\documentclass[journal]{IEEEtran}
\usepackage{dsfont,times,graphicx}
\usepackage{amssymb,amsmath,amsfonts}
\DeclareMathOperator*{\argmax}{argmax}
\DeclareMathOperator*{\arginf}{arginf}
\newtheorem{Lemma}{Lemma}
\newtheorem{Theorem}{Theorem}

\newtheorem{Definition}{Definition}

\newtheorem{Statement}{Statement}
\newtheorem{Remark}{Remark}
\newtheorem{Property}{Property}
\newtheorem{Assumption}{Assumption}

\begin{document}

\title{Dynamic Control of Tunable Sub-optimal Algorithms for Scheduling of Time-varying Wireless Networks}

\author{\authorblockN{Mahdi~Lotfinezhad, Ben~Liang, Elvino S. Sousa
}
\thanks{A short version of this submission is published in \cite{mahdi:iwqos08:m}.}
\thanks{The authors are with the Department of
Electrical and Computer Engineering, University of Toronto. E-mail:
\{mlotfinezhad, liang\}@comm.utoronto.ca,
es.sousa@utoronto.ca.}
\vspace{-.2in}}

\maketitle

\begin{abstract}
It is well known that for ergodic channel processes the Generalized Max-Weight Matching (GMWM) scheduling policy stabilizes the network for any supportable arrival rate vector within the network capacity region. This policy, however, often requires the solution of an NP-hard optimization problem. This has motivated many researchers to develop sub-optimal algorithms that approximate the GMWM policy in selecting schedule vectors.  One implicit assumption commonly shared in this context is that during the algorithm runtime, the channel states remain effectively unchanged.  This assumption may not hold as the time needed to select near-optimal schedule vectors usually increases quickly with the network size.  In this paper, we incorporate channel variations and the time-efficiency of sub-optimal algorithms into the scheduler design, to dynamically tune the algorithm runtime considering the tradeoff between algorithm efficiency and its robustness to changing channel states. Specifically, we propose a Dynamic Control Policy (DCP) that operates on top of a given sub-optimal algorithm, and dynamically but in a large time-scale adjusts the time given to the algorithm according to queue backlog and channel correlations. This policy does not require knowledge of the structure of the given sub-optimal algorithm, and with low overhead can be implemented in a distributed manner. Using a novel Lyapunov analysis, we characterize the throughput stability region induced by DCP and show that our characterization can be tight. We also show that the throughput stability region of DCP is at least as large as that of any other static policy. Finally, we provide two case studies to gain further intuition into the performance of DCP.
\end{abstract}

\thispagestyle{empty}


\vspace{-.1in}
\section{Introduction}\label{sec:intro}

The problem of scheduling of wireless networks has been extensively investigated in the literature. A milestone in this context is
the seminal work by Tassiulas and Ephremides \cite{tassiulas:ephremides:jac:92}, where the authors characterized the \emph{network-layer capacity region} of constrained queueing systems, including wireless networks, and designed a \emph{throughput-optimal} scheduling policy, commonly referred to as the GMWM scheduling. In this context, capacity region by definition is the largest region that can be stably supported using any policy, including those with the knowledge of future arrivals and channel states. A throughout-optimal policy is a policy that stabilizes the network for any input rate that is within the capacity region and, thus, has the largest stable throughput region. In general \cite{neely:jsac05}\cite{eryilmaz:sirkant:ton:05}, the GMWM scheduling should maximize the sum of backlog-rate products at each timeslot given channel states, which can be considered as a GMWM problem. This problem has been shown to be, in general, complex and NP-hard \cite{lin:shroff:ton06}\cite{eryilmaz:sirkant:ton:05}\cite{mazumdar:shroff:06}. Even in those cases where the optimization problem can be solved polynomially, distributed implementation becomes a major obstacle. These issues, naturally, motivated researchers to study and develop suboptimal centralized or distributed algorithms that can stabilize a fraction of the network-layer capacity region \cite{wu:srikant:05}\cite{lin:shroff:ton06}\cite{mazumdar:shroff:06}\cite{mahdi:ben:ToN:09}\cite{shroff:mazummdar:info:07}.

One implicit but major assumption in this context is that the time required to find an appropriate scheduling vector, \emph{search-time}, is negligible compared to the length of a timeslot, or otherwise, during this search-time, channel states remain effectively unchanged. Since many algorithms take polynomial time with the number of users to output a solution \cite{lin:shroff:ton06}\cite{mazumdar:shroff:06}\cite{shroff:mazummdar:info:07}, we see that this assumption may not hold in practice for networks with large number of users. In particular, it is possible that once an optimal solution corresponding to a particular channel state is found, due to channel variations, it becomes outdated to the point of being intolerably far away from optimality.

 Intuitively, for many suboptimal algorithms, the solution found becomes a better and more \emph{efficient} estimate of the optimal solution as the number of iterations increases or more \emph{time} is given to the algorithm, e.g., see PTAS in \cite{mazumdar:shroff:06}. This inspires us to consider this \emph{time-efficiency} correspondence as a classifying tool for sub-optimal algorithms. 
 As mentioned earlier, however, the solution found might become outdated due to channel variations.
 This poses a challenging problem as how the search-time given to sub-optimal algorithms should be adjusted to ensure an efficient scheduling with a large stable throughput region when channels states are time-varying.


 Our work in this paper addresses the above challenge by joint consideration of channel correlation and  time-efficiency of sub-optimal algorithms. In particular, we propose a dynamic control policy (DCP) that operates on top of a given sub-optimal algorithm $A$, where the algorithm is assumed to provide an approximate solution to the GMWM problem. Our proposed policy dynamically tunes the length of scheduling frames as the search-time given to the algorithm $A$ so as to maximize the time average of backlog-rate product, improving the throughput stability region. This policy does not require the knowledge of input rates or the structure of the algorithm $A$, works with a general class of sub-optimal algorithms, and with low-overhead can be implemented in a distributed manner. We analyze the performance of DCP in terms of its associated throughput stability region, and prove that this policy enables the network to support all input rates that are within $\theta_{\infty}$-scaled version of the capacity region. The scaling factor $\theta_{\infty}$ is a function of the interference model, algorithm $A$, and channel correlation, and we prove that in general this factor can be tight. We also show that the throughput stability region of DCP is at least as large as the one for any other static scheme that uses a fixed frame-length, or search-time, for scheduling.

As far as we are aware, our study is the first that jointly incorporates the time-efficiency of sub-optimal algorithms and channel variations into the scheduler design and stability region analysis.
One distinguishing feature of our work, apart form its practical implications, is the use of a Lyapunov drift analysis that is based on a \emph{random } number of steps. Therefore, to establish stability results, we use a method recently developed for Markov chains \cite{fralix}, and modify it such that it is also applicable to our network model.

The rest of this paper is organized as follows. We review the related work in the next section. Network model including details of arrival and channel processes is presented in Section~\ref{sec:net}. Structures of the sub-optimal algorithms and DCP policy are discussed in Section~\ref{sec:dcp}. We then provide performance analysis and the related discussion in Section~\ref{sec:performance}, followed by two case studies in Section~\ref{sec:casestudy}. Finally, we conclude the paper in Section~\ref{sec:conclusion}.

\vspace{-.1in}
\section{Related Work}\label{sec:relatedw}
Previous work on throughput-optimal scheduling includes the studies in \cite{tassiulas:ephremides:jac:92}\cite{Shakkottai:Stolyar:00}\cite{neely:jsac05}\cite{yeh:berry:it:07}. In particular, in \cite{tassiulas:ephremides:jac:92}, Tassiulas and Ephremides characterized the throughput capacity region for multi-hop wireless networks, and developed the GMWM scheduling as a throughput-optimal scheduling policy. This result has been further extended to general network models with ergodic channel and arrival processes \cite{neely:jsac05}. Due to its applicability to general multi-hop networks, the GMWM scheduling has been employed, either directly or in a modified form, as a key component in different setups and many cross-layer designs. Examples include control of cooperative relay networks \cite{yeh:berry:it:07}, rate control \cite{neely:modiano:fair:ton:08}, energy
efficiency \cite{neely:opt:ene:delay:JIT07}\cite{wanshi:neely:it:08}, and congestion control
\cite{eryilmaz:sirkant:fair:ton:07}\cite{chen:doyle:infocom:06}. This scheduling policy has also inspired pricing strategies
maximizing social welfare \cite{neely:infocom07}, and fair resource allocation \cite{eryilmaz:sirkant:fair:ton:07}.

Another example of the
throughput optimal control is the exponential rule proposed in
\cite{Shakkottai:Stolyar:00}. In addition to the exponential rule scheduling, there are other approaches that use queue
backlog, either explicitly or implicitly, for scheduling  \cite{peter:activequeue} \cite{Jiang:EECS-2008-38}\cite{goyal:tran:it:08}. For instance, in \cite{peter:activequeue},
active queue management is used that implements CSMA protocol with backlog dependent transmission
probabilities. It is shown that such an approach can implement a distributed fair buffer. In one other work \cite{Jiang:EECS-2008-38}, an adaptive CSMA algorithm is proposed that iteratively adjust nodes' aggressiveness based on nodes' (simulated) queue backlog.

The GMWM scheduling despite its optimality, in every timeslot, requires the solution of the GMWM problem, which can be,
in general, NP-hard and Non-Approximable \cite{mazumdar:shroff:06}. Thus, many studies has focused on developing sub-optimal constant factor approximations to the GMWM scheduling.
One interesting study addressing the complexity issue is the work in
\cite{tassiulas:infocom98}, where sub-optimal algorithms are modeled as randomized algorithms, and it is shown that throughput-optimality can be achieved with linear complexity. In a more recent work \cite{modiano:sig06}, the authors propose distributed schemes
to implement a randomized policy similar to the one in \cite{tassiulas:infocom98} that can stabilize the entire capacity region. These results, however, assume non-time-varying channels. Other recent studies in \cite{eryilmaz:sirkant:ton:05}\cite{chaporkar:tran:control:08} generalize the approach in \cite{tassiulas:infocom98} to time-varying networks, and prove its throughput-optimality. This optimality, as expected, comes at the price of requiring excessively large amount of other valuable resources in the network, which in this case is memory storage. Specifically, the memory requirement in \cite{eryilmaz:sirkant:ton:05}\cite{chaporkar:tran:control:08} increases exponentially with the
number of users, making the generalized approach hardly amenable to
practical implementation in large networks.

Another example of sub-optimal approximation is the work in \cite{lin:shroff:ton06}, where the authors
assume that the controller can use
only an \emph{imperfect} scheduling component, and as an example they use maximal matching to design a distributed scheduling that is
within a constant factor of optimality. This scheduling algorithm under the name of \emph{Maximal Matching} (MM) scheduling and its variants have been widely studied in the literature \cite{wu:srikant:05}\cite{mazumdar:shroff:06}\cite{wu:sriknat:tmc:07}\cite{shroff:mazummdar:info:07}\cite{neely:info08}\cite{chaporkar:tran:it:08}.
In \cite{wu:srikant:05}\cite{lin:shroff:ton06}, it is shown that under simple interference models, MM scheduling can achieve a throughput (or stability region) that is at least
half of the throughput achievable by a throughput-optimal algorithm (or the capacity region).
Extended versions of these results for more general interference models are presented in \cite{mazumdar:shroff:06}\cite{shroff:mazummdar:info:07}, where in \cite{shroff:mazummdar:info:07} randomized distributed algorithms are proposed for implementing MM scheduling, being a constant factor away from the optimality.
This result has been further strengthened recently \cite{joo:shroff:infocom:08} stating that the worst-case efficiency ratio of Greedy Maximal Matching scheduling in geometric network graphs under the $\kappa$-hop interference model is between $1/6$ and $1/3$. All of the mentioned proposals so far either do not consider channel variations, or assume the search-time is relatively small compared to the length of a timeslot.

The closest work to ours in this paper is \cite{mahdi:ben:ToN:09}, where based on the linear-complexity algorithm in \cite{tassiulas:infocom98}, the impact of channel memory on the stability region of a general class of sub-optimal algorithms is studied. Despite its consideration for channel variations, this work still does not model the search-time, and implicitly assumes it is negligible.

In this paper, we consider the problem of scheduling from a new perspective. We assume a sub-optimal algorithm $A$ is given that can approximate the solution of the GMWM problem, and whose efficiency
naturally improves as the search-time increases. We then devise a dynamic control policy which tunes the search-time,
as the length of scheduling frames, according to queue backlog levels in the network, and also based on channel correlations. As far as we are aware, our study is the
first that explicitly models the time-efficiency of sub-optimal approaches, and uses this concept along with channel correlation in the scheduler design.
\vspace{-.05in}
\section{Network Model}\label{sec:net}
We consider a wireless network with $N$ one-hop source-destination pairs, where each pair represents a data flow\footnote{ Extension to multi-hop flows is possible using the methods in \cite{tassiulas:ephremides:jac:92}\cite{neely:jsac05}.}.
Associated with each data flow, we consider a separate queue, maintained at the source of the flow, that holds packets to be
transmitted over a wireless link. Examples of this type of network
include downlink or uplink of a cellular or a mesh network.

\vspace{-.15in}
\subsection{Queueing}\label{sec:net:queue}

 We assume the system is time-slotted, and
channels hold their state during a timeslot but may change from one
timeslot to another. Let $\mathbf{s}(t)$ be the matrix of all channels states from any given node $i$ to any other node $j$ in the network at time $t$. For instance, when the network is the downlink or uplink of a cellular network, $\mathbf{s}(t)$ will reduce to the vector
of user-base-station channel states, i.e., $\mathbf{s}(t)=(s_{1}(t), \dots, s_{N}(t))$, where $s_{i}(t)$ is the state of the $i_{\text{th}}$ link (corresponding to the $i_{\text{th}}$ data flow) at time $t$. Throughout the chapter, we use bold face to denote vectors or matrices. Let
$\mathcal{S}$ represent the set of all possible channel state matrices with
finite cardinality $|\mathcal{S}|$. Let $D_{i}(t)$ denote the rate over the
$i_{\text{th}}$ link corresponding to the $i_{\text{th}}$ data flow at time $t$, and $\mathbf{D}(t)$ be the corresponding vector of rates, i.e., $\mathbf{D}(t)=(D_{1}(t), \dots, D_{N}(t))$. In addition, let $I_{i}(t)$ represent the amount of resource used by the $i_{\text{th}}$ link at time $t$, and $\mathbf{I}(t)$ be the corresponding vector, i.e., $\mathbf{I}(t)=(I_{1}(t), \cdots, I_{N}(t))$. The vector $\mathbf{I}(t)$ contains both scheduling and resource usage information, and hereafter, we refer to it simply as the schedule vector. Let $\mathcal{I}$ denote the set containing all possible schedule vectors, with finite cardinality $|\mathcal{I}|$.


Note that the exact specification of the scheduling vector $\mathbf{I}(t)$ is system dependent. For instance, in CDMA systems, it may represent the vector of power levels associated with wireless links; in OFDMA systems, it may represent the number of sub-channels allocated to each physical link; and when interference is modeled as the K-hop interference model \cite{mazumdar:shroff:06}, the vector can be a link activation vector representing a sub-graph in the network. Assuming that transmission rates are completely characterized given channel states, the schedule vector, and the interference model, we have
\begin{align*}
\mathbf{D}(t)=\mathbf{D}(\mathbf{s}(t),\mathbf{I}(t)).
\end{align*}
We assume that transmission rates are bounded, i.e., for all $\mathbf{s}\in\mathcal{S}$ and $\mathbf{I}\in\mathcal{I}$,
\begin{align*}
D_{i}(\mathbf{s},\mathbf{I})<D_{max},  \ 1\leq i\leq N,
\end{align*}
for some large $D_{max}>0$.

Let $A_{i}(t)$ be the number of packets arriving in timeslot $t$ associated with the $i_{\text{th}}$ link (or data flow), and $\mathbf{A}(t)$ be the vector of arrivals, i.e.,
$\mathbf{A}(t)=(A_{1}(t), \cdots, A_{N}(t))$. We assume arrivals are i.i.d.\footnote{This assumption is made to simplify the analysis, and our results can be extended to non i.i.d arrivals.} with mean vector
\begin{align*}
\mathds{E}[\mathbf{A}(t)]=\mathbf{a}=(a_{1}, \dots, a_{N}),
 \end{align*}
and bounded above:
\begin{align*}
A_{i}(t)<A_{max}, \ 1\leq i\leq N,
\end{align*}
for some large $A_{max}$.

Finally, let $\mathbf{X}(t)=(X_{1}(t), \dots, X_{N}(t))$ be the vector of queue lengths, where $X_{i}(t)$ is the queue length associated with the $i_{\text{th}}$ link (or data flow). Using the preceding definitions, we see that $\mathbf{X}(t)$ evolves according to the following equation
\begin{eqnarray*}
\mathbf{X}(t+1)=\mathbf{X}(t)+\mathbf{A}(t)-\mathbf{D}(t)+\mathbf{U}(t),
\label{X(t+1)}
\end{eqnarray*}
 where $\mathbf{U}(t)$ represents the
wasted service vector with non-negative elements; the service is
wasted when in a queue the number of packets waiting for
transmission is less than the number that can be transmitted, i.e.,
when $X_{i}(t)<D_{i}(t)$.
\vspace{-.15in}
\subsection{Channel State Process}\label{sec:channelp}
We assume the channel state process is stationary and ergodic. In particular, for all $\mathbf{s}\in \mathcal{S}$, as $k \to \infty$, we have
\begin{align*}
\frac{1}{k}\sum_{i=0}^{k-1}\mathbf{1}_{\mathbf{s}(t+i)=\mathbf{s}} \to \pi(\mathbf{s}), \ \ a.s.,
\end{align*}
where $\mathbf{1}_{(\cdot)}$ denotes the indicator function associated with a given event, and $\pi(\mathbf{s})$ is the steady-state probability of state $\mathbf{s}$. Let $\mathcal{P}_{t}$ represent the past history of the channel process and be defined by $\mathcal{P}_{t}=\{\mathbf{s}(i); 0\leq i \leq t\}$. The above almost surely convergence implies that for
 any $\epsilon>0$ and $\zeta>0$, we can find a sufficiently large $K_{\epsilon,\zeta,t}>0$ such that \cite{borovkov}
\begin{align}
P\Big(\sup_{k>K_{\epsilon,\zeta,t}} \big|\frac{1}{k}\sum_{i=0}^{k-1}\mathbf{1}_{\mathbf{s}(t+i)=\mathbf{s}}-\pi(\mathbf{s})\big|>\epsilon\ \ \big|
 \mathcal{P}_{t}\Big)<\zeta. \label{channel}
\end{align}
We assume that the almost surely convergence is \emph{unform} in the past history and $t$ in the sense that regardless of $\mathcal{P}_{t}$ and $t$, there exists a $K_{\epsilon,\zeta}$ such that (\ref{channel}) holds with $K_{\epsilon,\zeta,t}=K_{\epsilon,\zeta}$\footnote{Examples of this channel model include but are not limited to Markov chains.}.
\vspace{-.15in}
\subsection{Capacity Region}
In our context, capacity region, denoted by $\Gamma$, is defined as the closure of the set of all input rates that can be
\emph{stably} supported by the network using any scheduling
policy including those that use the knowledge of future arrivals and channel states.
In \cite{tassiulas:ephremides:jac:92}\cite{tassiulas:infoth:97} and
recently under general conditions in \cite{neely:jsac05}, it has
been shown that the capacity region $\Gamma$ is given by
 \begin{eqnarray}
\Gamma=\sum_{\mathbf{s}\in \mathcal{S}} \pi(\mathbf{s})\
\text{Convex-Hull}\{\mathbf{D}(\mathbf{s},\mathbf{I})|\mathbf{I}\in\mathcal{I}\}. \nonumber
 \end{eqnarray}
\vspace{-.1in}
\section{Dynamic Control Policy}\label{sec:dcp}
As mentioned in the introduction, DCP controls and tunes the search-time given to a sub-optimal algorithm to improve the stability region. The considered sub-optimal algorithms are assumed to provide a sub-optimal solution to the GMWM problem. In the following, we first elaborate on the structure of the sub-optimal algorithms, and then, describe the operation of DCP.
\vspace{-.15in}
\subsection{Sub-optimal Algorithms Approximating GMWM Problem}\label{sec:algorithmA}
It is well known that the GMWM scheduling is throughput-optimal in that it stabilizes the network for all input rates interior to capacity region $\Gamma$. This policy in each timeslot uses the schedule vector $\mathbf{I}^{*}(t)$ that is argmax to the following GMWM problem:
\begin{align}
&\max\sum_{l=1}^{N}X_{l}(t)D_{l}(\mathbf{s}(t), \mathbf{I}) ,
\ \quad \text{subject to } \mathbf{I}\in \mathcal{I}.
\label{Dstar}
\end{align}
However, as mentioned in Section~\ref{sec:intro}, this optimization problem can be in general NP-hard. We therefore assume that there exists an algorithm $A$ that can provide suboptimal solutions to the max-weight problem given in (\ref{Dstar}).
To characterize the structure of algorithm $A$, let $\mathbf{I}^{*}(\mathbf{X},\mathbf{s})$ be the argmax to (\ref{Dstar}) by setting $\mathbf{X}(t)=\mathbf{X}$ and $\mathbf{s}(t)=\mathbf{s}$. Thus,
\begin{align*}
\mathbf{I}^{*}(\mathbf{X},\mathbf{s})=\argmax_{\mathbf{I}\in \mathcal{I}} \mathbf{X}\mathbf{D}(\mathbf{s}, \mathbf{I}),
\end{align*}
where $\mathbf{X}\mathbf{D}(\mathbf{s}, \mathbf{I})$ is the scalar product of the two vectors, and for ease of notation, we have dropped the transpose symbol required for $\mathbf{D}(\mathbf{s}, \mathbf{I})$. In the rest of this paper, we use the same method to show the scalar products. Associated with $\mathbf{I}^{*}(\mathbf{X},\mathbf{s})$, let $\mathbf{D}^{*}(\mathbf{X},\mathbf{s})$ be defined as
\begin{align}
\mathbf{D}^{*}(\mathbf{X},\mathbf{s})=\mathbf{D}(\mathbf{s},\mathbf{I}^{*}(\mathbf{X},\mathbf{s})).\label{dstarv}
\end{align}
Thus, $\mathbf{D}^{*}(\mathbf{X},\mathbf{s})$ is the optimal rate, in the sense of (\ref{Dstar}), when the backlog vector is $\mathbf{X}$ and the channel state is $\mathbf{s}$.

Let $\mathbf{I}^{(n)}$ be the \emph{output} schedule vector of algorithm $A$ when it is given an amount of time equal to $n$ timeslots, $\mathbf{X}(t)=\mathbf{X}$, and $\mathbf{s}(t)=\mathbf{s}$. We therefore assume that the time given to algorithm $A$ can be programmed or tuned as desired, or simply, the algorithm can continue or iterate towards finding better solutions over time. We assume that $\mathbf{I}^{(n)}$ is in general a random vector with distribution $\mu_{\mathbf{X},\mathbf{s}}^{(n)}$.
Since the objective function in (\ref{Dstar}) is a continuous function of $\mathbf{X}(t)$, we naturally assume that algorithm $A$ characterized by the distribution of $\mathbf{I}^{(n)}$, 
for all $n\geq 1$, and all values of $\mathbf{X}$ and $\mathbf{s}$, has the following property:
 \begin{Assumption}\label{p1}
 For all $\mathbf{I}\in \mathcal{I}$, $\mathbf{s} \in \mathcal{S}$, and $n$, we have that
 \begin{eqnarray}
|\mu_{\mathbf{X}_{1},\mathbf{s}}^{(n)}(\mathbf{I}^{(n)}=\mathbf{I})-\mu_{\mathbf{X}_{2},\mathbf{s}}^{(n)}(\mathbf{I}^{(n)}=\mathbf{I})|\to 0, \nonumber
\end{eqnarray}
as $\mathbf{X}_{1}\to \mathbf{X}_{2}$. In addition, assuming and keeping $\|\mathbf{X}_{1}-\mathbf{X}_{2}\|<C$ for a given $C>0$, the above convergence also holds when $\|\mathbf{X}_{1}\| \to \infty$. Moreover, the convergence becomes equality if $\mathbf{X}_{1}= \beta \mathbf{X}_{2}$, for some $\beta>0$.
 \end{Assumption}

In the following, we discuss concrete models that provide further details on the structure of algorithm $A$. Note that these models serve only as examples, and our results do not depend on any of these models; what required is only Assumption~\ref{p1}.

The first model arises from the intuition that the distribution $\mu^{(n)}_{\mathbf{X},\mathbf{s}}$ should improve as $n$ increases. More precisely, we can define the sequence
$\{\mu_{\mathbf{X},\mathbf{s}}^{(n)}, \ n=1,2,3,\cdots\}$ to be an \emph{improving} sequence if for all $n>1$,
\begin{align*}
\mathds{E}[\mathbf{X}\mathbf{D}(\mathbf{s},\mathbf{I}^{(n)})]\geq \mathds{E}[\mathbf{X}\mathbf{D}(\mathbf{s},\mathbf{I}^{(n-1)})]\geq \cdots \geq \mathds{E}[\mathbf{X}\mathbf{D}(\mathbf{s},\mathbf{I}^{(1)})].
\end{align*}
The first model uses the above and defines a \emph{natural} algorithm to be the one for which the above inequalities hold for all values of $\mathbf{X}$ and $\mathbf{s}$.

As for the second model, we may have that $\mathbf{I}^{(n)}$ is such that
\begin{align}
\mathbf{X}\mathbf{D}(\mathbf{s},\mathbf{I}^{(n)})\geq g(n)\mathbf{X}\mathbf{D}(\mathbf{s},\mathbf{I}^{*}(\mathbf{X},\mathbf{s})),\label{algorithma}
\end{align}
where the function $g(n)$ is a non-decreasing function of $n$, and less than or equal to one. For instance, if the optimization problem can be approximated to a convex problem \cite{boyd}, then $g(n)=\xi(1-\zeta^{n})$, where $0<\xi\leq 1$ and $0\leq \zeta <1$. Another possible form for $g(n)$ is
\begin{align*}
\Big(1-\beta\frac{\ln N}{\ln n}\Big),
\end{align*}
where $\beta$ is a positive constant. This form of $g(n)$ may stem from cases where the optimization problem associated with (\ref{Dstar}) admits Polynomial-Time Approximation Scheme (PTAS) \cite{mazumdar:shroff:06}.

The last model that we consider is a generalization of the previous model, where we assume that (\ref{algorithma}) holds with probability $h(n)$ as a non-decreasing function of $n$. This specification can model algorithms that use randomized methods to solve (\ref{Dstar}), and without its consideration for the improvement over $n$, is similar to the ones developed in \cite{tassiulas:infocom98}\cite{mahdi:ben:ToN:09}.
\vspace{-.15in}
\subsection{Dynamic Control Policy and Scheduling}\label{sec:dcp:sch}
The dynamic control policy in this paper interacts with scheduling component, and through some measures, which will be defined later, dynamically tunes the time spent by the scheduler, or more precisely algorithm $A$, to find a schedule vector. In what follows, we describe the joint operation of DCP and the scheduler.

\begin{figure}[t]
\centering
\includegraphics[width=2.8in]{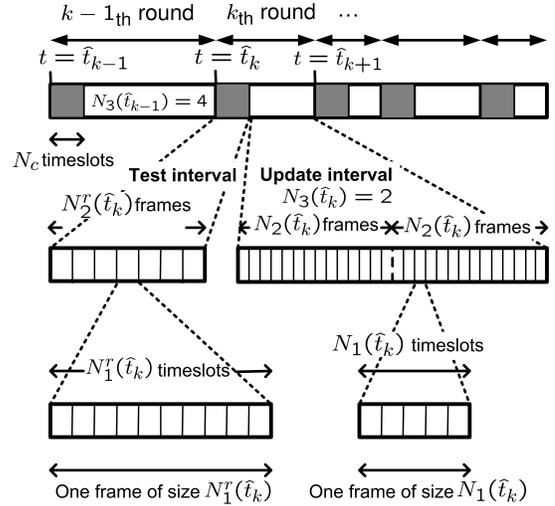}
\caption{Illustration of scheduling rounds, test intervals, update intervals, and frames.}\label{fig:rounds}
 \vspace{-.2in}
\end{figure}

  As DCP operates, the time axis becomes partitioned to a sequence of \emph{scheduling rounds}, where each round might consist of a different number of timeslots. An illustrative example is provided in Fig.~\ref{fig:rounds}. Let $\hat{t}_{k}$ denote the start time of the $k_{\text{th}}$ round. Each round begins with a \emph{test interval} followed by an \emph{update interval}. In the beginning of the test interval of each round, a \emph{candidate} value for the number of timeslots given to the algorithm $A$ to solve (\ref{Dstar}) is selected by DCP. Let $N^{r}_{1}(\hat{t}_{k})$ denote this candidate value for the $k_{\text{th}}$ round, and assume $N^{r}_{1}(\hat{t}_{k}) \in \mathcal{N}_{1}$, where $\mathcal{N}_{1}$ has a finite cardinality. In the rest, we use $N^{r}_{1}$ instead of $N^{r}_{1}(\hat{t}_{k})$ where appropriate. The algorithm that chooses the candidate value might be in general a randomized algorithm. Thus, we use the superscript $r$ to make this point clear. We assume $N^{r}_{1}$ takes an optimal value with probability $\delta>0$, where optimality will be defined later by (\ref{tildeN1}) and its following discussion.

   We set the length of the test interval to be
 \begin{align*}
N^{r}_{1}N^{r}_{2}=N_{c}=\text{const.},
\end{align*}
 a multiple of $N^{r}_{1}$, where $N^{r}_{2}$ is adjusted accordingly so that the test interval has a fixed length $N_{c}$. Therefore, given $N^{r}_{1}$, the test interval becomes partitioned into $N^{r}_{2}$ consecutive frames of $N^{r}_{1}$ timeslots. In the beginning of each frame, e.g., at time $t$, the current backlog vector $\mathbf{X}(t)$ and channel state $\mathbf{s}(t)$ are provided to the algorithm $A$. The algorithm then spends $N^{r}_{1}$ timeslots to find a schedule vector. Depending on the properties of a particular instance of algorithm $A$, this vector is used by the scheduler to update scheduling decisions in the next frame in a variety of methods.

  In the first method, the schedule vector found after $N^{r}_{1}$ timeslots in the frame starting at time $t$ is used throughout the next frame of $N^{r}_{1}$ timeslots starting at time $t+N^{r}_{1}$. Thus, the schedule vector used in any frame is obtained by using backlog and channel state information at the beginning of its previous frame. This method is general and can be applied to all types of algorithm $A$.

  We can apply a second method where algorithm $A$ is capable of outputting schedule vectors in intermediate steps, and not only after the planned $N^{r}_{1}$ timeslots. Consider the $i_{\text{th}}$ timeslot of a given frame of $N^{r}_{1}$ timeslots started at time $t$, where $i\leq N^{r}_{1}$. Suppose $\mathbf{I}^{c}_{t+i}$ is the intermediate solution found by the algorithm $A$ after $i$ timeslots in the considered frame, and $\mathbf{I}_{p}$ is the vector found at the end of its previous frame. Then, we may assume that with some probability, $\mathbf{I}^{c}_{t+i}$ is used if
 \begin{align*}
\mathbf{X}(t+i)\mathbf{D}(\mathbf{s}(t+i),\mathbf{I}^{c}_{t+i})>\mathbf{X}(t+i)\mathbf{D}(\mathbf{s}(t+i),\mathbf{I}_{p}),
 \end{align*}
 otherwise $\mathbf{I}_{p}$ is used in the timeslot following the $i_{\text{th}}$ timeslot. The update rule in \cite{mahdi:ben:ToN:09} provides an example where two schedule vectors are compared, and the best is selected with a well-defined probability.

  As for the third method, we may assume algorithm $A$ can accept an initial schedule vector. In this case, we can assume that the algorithm $A$ at a given frame accepts the schedule vector found in the previous frame as the initial point to the optimization problem of (\ref{Dstar}). Note that many graph-inspired algorithms do not start from a given initial vector (as a sub-graph), but instead, gradually progress towards a particular solution. These algorithms\footnote{Adaption of these algorithm to time-varying networks is an interesting problem, and is left for the future research.}, therefore, do not belong to the class of algorithms considered for this method. A forth method can also be considered by mixing the second and the third method if algorithm $A$ has the corresponding required properties. Our results in this paper extend to these methods as long as Property~\ref{p2} and Property~\ref{p3} in Section~\ref{sec:func} hold.

  Given $N^{r}_{1}$, and a method to use the output of algorithm $A$, DCP evaluates scheduling performance resulting from the value for $N^{r}_{1}$. The performance criterion is the normalized time-average of the backlog-rate (scalar) product. To define the criterion precisely, let $\varphi(\cdot,\cdot,\cdot)$ be defined as
\begin{eqnarray}
\varphi(t,n_{1},n_{2})=\sum_{j=0}^{n_{2}-1}\sum_{i=0}^{n_{1}-1}
\frac{\mathbf{X}_{t+jn_{1}+i}\mathbf{D}_{t+jn_{1}+i}}{n_{1}n_{2}\|\mathbf{X}_{t}\|}.\nonumber
\end{eqnarray}
If $\|\mathbf{X}_{t}\|=0$, we set $\varphi(t,n_{1},n_{2})=0$.
Based on the above definition, the criterion associated with the test interval of the $k_{\text{th}}$ scheduling round, which is computed by DCP, is denoted by $\varphi^{r}(\hat{t}_{k})$, where
\begin{eqnarray*}
\varphi^{r}(\hat{t}_{k})=\varphi(\hat{t}_{k},N^{r}_{1}(\hat{t}_{k}),
N^{r}_{2}(\hat{t}_{k})).
\end{eqnarray*}
This quantity is then used to determine the length of frames in the update interval of the $k_{\text{th}}$ round.

Update intervals are similar to the test intervals in that they are consisted of a multiple number of fixed-length frames. More precisely, we assume that the update interval in the $k_{\text{th}}$ round becomes partitioned into $N_{2}(\hat{t}_{k})N_{3}(\hat{t}_{k})$ consecutive frames of $N_{1}(\hat{t}_{k})$ timeslots. Integers $N_{1}(\hat{t}_{k})$ and $N_{2}(\hat{t}_{k})$ are such that
\begin{align}
N_{1}(\hat{t}_{k})N_{2}(\hat{t}_{k})=N_{c}. \label{nc}
\end{align}
Therefore, the length of the $k_{\text{th}}$ update interval is $N_{3}(\hat{t}_{k})$ times the length of a test interval.
Moreover, we see that $N_{1}(\hat{t}_{k})$ in the $k_{\text{th}}$ update interval takes the role of $N^{r}_{1}(\hat{t}_{k})$ in the $k_{\text{th}}$ test interval. Assuming the same method is applied to all test and update intervals to use the output of algorithm $A$, we can properly define
$\varphi(\hat{t}_{k})$ as
\begin{eqnarray*}
\varphi(\hat{t}_{k})=\varphi(\hat{t}_{k}+N_{c},N_{1}(\hat{t}_{k}),
N_{2}(\hat{t}_{k})N_{3}(\hat{t}_{k})).
\end{eqnarray*}
The quantity $\varphi(\hat{t}_{k})$ is similar to $\varphi^{r}(\hat{t}_{k})$, and measures the normalized time-average of backlog-rate product in the $k_{\text{th}}$ update interval.

DCP , on top of algorithm $A$, uses $\varphi(\hat{t}_{k-1})$ and $\varphi^{r}(\hat{t}_{k})$ to dynamically control the value of $N_{1}(\hat{t}_{k})$ and $N_{3}(\hat{t}_{k})$ over time. Specifically, in the $k_{\text{th}}$ round, at the \emph{end} of the test interval, the policy chooses either the $N_{1}$ used in the previous update interval, $N_{1}(\hat{t}_{k-1})$, or
the newly chosen value of $N_{1}$ in the current test interval, $N^{r}_{1}(\hat{t}_{k})$, according to the following update rule:
\begin{eqnarray*}
N_{1}(\hat{t}_{k})=\left \{
\begin{array}{ll}
N_{1}^{r}(\hat{t}_{k})& \text{if} \ \varphi^{r}(\hat{t}_{k})>\varphi(\hat{t}_{k-1})+\alpha \\
N_{1}(\hat{t}_{k-1})  & \text{otherwise}
 \end{array} \right.,
\end{eqnarray*}
where $\alpha$ is a suitably small but otherwise an \emph{arbitrary} positive constant. At the same time, the value of $N_{3}(\hat{t}_{k})$, is updated according to the following:
\begin{eqnarray*}
N_{3}(\hat{t}_{k})=\left \{
\begin{array}{ll}
\max(1,\frac{N_{3}(\hat{t}_{k-1})}{2}) & \text{if} \ \varphi^{r}(\hat{t}_{k})>\varphi(\hat{t}_{k-1})+\alpha \\
\min(L_{1},2N_{3}(\hat{t}_{k-1}))  & \text{otherwise,}
 \end{array} \right.
\end{eqnarray*}
where $L_{1}$ is a suitably large but otherwise an \emph{arbitrary} positive constant. Note that $N_{2}(\hat{t}_{k})$ becomes updated such that (\ref{nc}) holds.
 Once the values of $N_{1}$, $N_{2}$, and $N_{3}$ are updated, in the rest of the scheduling round, which by definition is the update interval, the policy proceeds with computing the time average $\varphi(\hat{t}_{k})$. When the $k_{\text{th}}$ round finishes, the $k+1_{\text{th}}$ round starts with a test interval, and DCP proceeds with selecting $N^{r}(\hat{t}_{k+1})$, and applying the update rule at the end of the $k+1_{\text{th}}$ test interval. This completes the description of joint operation of DCP and the scheduling component.

 Considering the above description, we see that DCP keeps trying new values for $N_{1}$. Once a good candidate is found for $N_{1}$, the update rule with high probability uses this value for longer periods of time by doubling the length of update intervals. In case the performance in terms of the backlog-rate product degrades, the length of update intervals are halved to expedite trying new values for $N_{1}$. Note that $\alpha$ can be arbitrarily small, but should be a positive number. This avoids fluctuations between different values of $N_{1}$ performing closely, thus preventing short update intervals. In addition, it limits incorrect favoring towards new values of $N_{1}$ in the test intervals, where due to atypical channel conditions, the normalized backlog-rate product deviates from and goes beyond its expected value. Finally, note that $L_{1}$ can be arbitrarily large, but should be a finite integer. This assumption is mainly analysis-inspired but is also motivated by the fact that a larger $L_{1}$ can lead to a larger delay.
\vspace{-.15in}
\section{Performance Analysis}\label{sec:performance}
In this section, we evaluate the performance of DCP in terms of its associated stability region. We first introduce several key definitions and functions, and then state the main theorem of the paper.
\vspace{-.15in}
\subsection{Definitions}
Since the backlog vector is non-Markovian, we consider the following definition for the stability of a process.
\subsubsection{Stability}\label{sec:per:stab}
Suppose there are a bounded closed region $\mathcal{C}$ around the origin, and a real-valued function $F(\cdot)\geq 0$ such the following holds: For any $t$, and $\sigma_{\mathcal{C}}$ defined by
\begin{align*}
\sigma_{\mathcal{C}}=\inf\{i\geq 0:\mathbf{X}_{t+i}\in
\mathcal{C}\},
\end{align*}
we have
\begin{align*}
\mathds{E}[\sigma_{\mathcal{C}}]\leq F(\mathbf{X}(t)) \mathbf{1}_{\mathbf{X}(t)\notin \mathcal{C}}.
\end{align*}
Then, the system is said to be stable.

This definition implies that when $\mathbf{X}(t)\notin \mathcal{C}$, e.g., when $\|\mathbf{X}(t)\|$ is larger than a threshold, the conditional expectation of the time required to return to $\mathcal{C}$, e.g., so that $\|\mathbf{X}(t)\|$ becomes less than or equal to the threshold, is bounded by a function of only $\mathbf{X}(t)$, uniformly in the past history and $t$. This definition further implies that if the sequence $\mathbf{X}(t)$ is stable, then \cite{buche:kushber:04}
\begin{align*}
\lim_{k\to \infty} \sup_{t} P(|\mathbf{X}(t)|>k)=0.
\end{align*}
\vspace{-.15in}
\subsubsection{$\theta$-scaled Region and Maximal Stability}
Suppose $0\leq\theta \leq1$. A region is called $\theta$-scaled of the region $\Gamma$, and denoted by $\theta \Gamma$, if it contains
all rates that are $\theta$-scaled of the rates in $\Gamma$, i.e.,
\begin{align*}
\theta \Gamma =\{\mathbf{a}_{1}: \ \mathbf{a}_{1}=\theta \mathbf{a}_{2},\ \ \text{for some}\ \mathbf{a}_{2}\in \Gamma\}.
\end{align*}
Further, the $\theta$-scaled region is called \emph{maximally} stable if for all arrival rate vectors interior to $\theta \Gamma$, the system can be stabilized, and for all $\epsilon>0$ there exists at least one rate vector interior to $(\theta+\epsilon) \Gamma$ that makes the system unstable, both under the same given policy. Thus, maximal stability determines the largest \emph{scaled} version of $\Gamma$ that can be stably supported under a given policy.
\vspace{-.15in}
\subsection{Auxiliary Functions and Their Properties}\label{sec:func}
To define the first function, hypothetically suppose for all $t$, $\mathbf{X}(t)=\mathbf{X}$ for a given $\mathbf{X}$, $\mathbf{X}\neq \mathbf{0}$, and thus, $\mathbf{X}(t)$ does not get updated. In addition, assume that $N_{1}$ has a fixed value over time. Considering these assumptions and an update interval of infinite number of frames\footnote{Here, we assume the channel evolves, and that the algorithm $A$ is used in the same manner as it is used in an ordinary update interval with a finite $N_{c}$, as discussed in Section~\ref{sec:dcp:sch}. }, each consisting of $N_{1}$ timeslots, we can see that in the \emph{steady state}, the expected normalized backlog-rate product, averaged over one frame, is equal to
\begin{eqnarray}
\phi(\mathbf{X},N_{1})=\mathds{E}_{\mathbf{s},A}\frac{(\sum_{i=1}^{N_{1}}\mathbf{X}
\mathbf{D}_{i})}{N_{1} \|\mathbf{X}\|},
\end{eqnarray}
where $\mathbf{D}_{i}$ is the rate vector in the $i_{\text{th}}$ timeslot of a given frame in the steady state. This expectation is over the steady-state
distribution of channel process, and possibly over the randomness introduced by the algorithm $A$.

Intuitively, $\phi(\mathbf{X},N_{1})$ states how well a particular choice for $N_{1}$ performs, in terms of backlog-rate product, when queue-length changes are ignored. This is exactly what we need to study since the stability region often depends on the behavior of scheduling at large queue-lengths, where in a finite window of time the queue-lengths do not change significantly.

To simplify notation, where appropriate, we use
$t$ as the first argument of $\phi(\cdot,\cdot)$; by that we mean\footnote{By definition of $\phi(\cdot, \cdot)$, here we hypothetically assume the backlog vector $\mathbf{X}(t_{1})$ for all times $t_{1}$ is equal to $\mathbf{X}(t)$. }
\begin{align*}
\phi(t,N_{1})=\phi(\mathbf{X}(t),N_{1}).
\end{align*}
Having defined $\phi(\mathbf{X},N_{1})$, we define $\tilde{N}_{1}(\mathbf{X})$
and $\tilde{\phi}(t)$ by
\begin{align}
\tilde{N_{1}}(\mathbf{X})&=\argmax_{N_{1}\in \mathcal{N}_{1}}\phi(\mathbf{X},N_{1}), \label{tildeN1}
\end{align}
and
\begin{align}
 \tilde{\phi}(t)&=\tilde{\phi}(\mathbf{X}(t)) =\phi\Big(\mathbf{X}(t),\tilde{N}_{1}\big(\mathbf{X}(t)\big)\Big). \nonumber
\end{align}
Finally, for a given $\mathbf{X}$ with $\|\mathbf{X}\| \neq 0$, we define
\begin{eqnarray}
\chi(\mathbf{X})=\mathds{E}_{\mathbf{s}}\Big[\frac{\mathbf{X}\mathbf{D}^{*}(\mathbf{X},\mathbf{s})}{\|\mathbf{X}\|}\Big],\nonumber
\end{eqnarray}
where $\mathbf{D}^{*}(\mathbf{X},\mathbf{s})$ is defined in (\ref{dstarv}), and
the expectation is over the \emph{steady-state} distribution of the channel process.

According to the above definitions, we see that when variations in the backlog vector are ignored after time $t$, and $N_{1}$ is confined to have a fixed value, $\tilde{N}_{1}(\mathbf{X}(t))$ becomes the optimal value for $N_{1}$ in terms of the normalized backlog-rate product, and $ \tilde{\phi}(t)$ represents the corresponding expected value. In particular, note that $\tilde{N}_{1}(\mathbf{X})$ is a function of $\mathbf{X}$ and may take different values for different $\mathbf{X}$'s. The quantity $\chi(\mathbf{X})$, on the other hand, is the expected normalized backlog-rate product if for all states we could find the optimal schedule vector. This quantity, therefore, can serve as a benchmark to measure performance of sub-optimal approaches.

Note that $\chi(\mathbf{X})$ is continuous function of $\mathbf{X}$ and does not depend on $\|\mathbf{X}\|$. Similarly, by Assumption~\ref{p1}, $\phi(\mathbf{X},N_{1})$ does not depend on $\|\mathbf{X}\|$, and is expected to have the following property.
\begin{Property}\label{p2}
 Suppose $\|\mathbf{X}_{1}-\mathbf{X}_{2}\|<C$ for a given $C>0$. For any given $\epsilon>0$, there exists a sufficiently large $M>0$ such that if $\|\mathbf{X}_{1}\|>M$, then for all $N_{1}\in \mathcal{N}_{1}$
\begin{eqnarray}
|\phi(\mathbf{X}_{1},N_{1})-\phi(\mathbf{X}_{2},N_{1})|< \epsilon. \nonumber
\end{eqnarray}
\end{Property}
If the first or the second method in Section~\ref{sec:dcp:sch} is used, this property holds since by Assumption~\ref{p1}, algorithm $A$ statistically finds similar schedule vectors when two backlog vectors are close and large. In case the third or the forth method is used, it is possible to consider explicit restrictions for algorithm $A$ such that $\phi(\mathbf{X},N_{1})$ is well-defined and Property~\ref{p2} holds. However, in this paper, we simply assume that algorithm $A$ is well-structured, in terms of the distribution of $\mathbf{I}^{(n)}$, so that by the ergodicity of the channel process this property also holds for these methods.

Recall that $\varphi^{r}(\hat{t}_{k})$ is the normalized time-average of backlog-rate product over the $k_{\text{th}}$
test interval. If we assume that the backlog vector is kept fixed at $\mathbf{X}(\hat{t}_{k})$, by ergodicity of the channel process as explained in Section~\ref{sec:channelp}, we expect $\varphi^{r}(\hat{t}_{k})$ to converge to $\phi(\hat{t}_{k},N^{r}_{1}(\hat{t}_{k}))$. Hence, when the number of frames is large, which is the case when $N_{c}$ is large, $\varphi^{r}(\hat{t}_{k})$ should be close to $\phi(\hat{t}_{k},N^{r}_{1}(\hat{t}_{k}))$ with high probability. However, the backlog vector is not fixed and changes over time. But by Assumption~\ref{p1}, algorithm $A$ statistically responds similarly to different backlog vectors if they are close and sufficiently large. This can be exactly our case since arrivals and departures are limited, and thus, for a fixed $N_{c}$, the changes in the norm of backlog vector are bounded over one test interval. Therefore, by Assumption~\ref{p1}, if $\|\mathbf{X}(\hat{t}_{k})\|$ is sufficiently large, the changes in the backlog have little impact on the distribution of $\varphi^{r}(\hat{t}_{k})$. Applying a similar discussion to $\varphi(\hat{t}_{k})$ while noting that the length of update intervals is bounded by $L_{1}N_{c}$, we expect the following property.
\begin{Property}\label{p3}
There exist $\varrho_{\varphi}>0$ and $\theta_{\varphi}>0$ such that
for any given $\epsilon>0$, there exists $M>0$ such
that if $\|\mathbf{X}_{\hat{t}_{k}}\|>M$, then regardless of $k$ and the past
history, up to and including time $\hat{t}_{k}$, with probability at least $(1-\varrho_{\varphi})$
\begin{align}
\big|\varphi^{r}(\hat{t}_{k})-\phi(\hat{t}_{k},N^{r}_{1}(\hat{t}_{k}))\big|
<\theta_{\varphi}+\epsilon.\nonumber
\end{align}
Similarly, regardless of $k$ and the past
history, up to and including time $\hat{t}_{k}+N_{c}$, with probability at least $(1-\varrho_{\varphi})$
\begin{align}
\big|\varphi(\hat{t}_{k})-\phi(\hat{t}_{k}+N_{c},N_{1}(\hat{t}_{k}))\big|<\theta_{\varphi}+\epsilon.\nonumber
\end{align}
 Moreover,
\begin{align*}
\lim_{N_{c}\to \infty} \varrho_{\varphi} =\lim_{N_{c}\to \infty} \theta_{\varphi}=0.
\end{align*}
\end{Property}

 According to the preceding discussion, we can see that $\theta_{\varphi}$ and $\varrho_{\varphi}$ mainly measure how fast the time-averages
converge to their expected value, and $\epsilon$ models the error due
to variations in the backlog vector $\mathbf{X}_{\hat{t}_{k}+i}$. Thus, as stated above,
$\varrho_{\varphi}$ and $\theta_{\varphi}$ can be made arbitrarily
small by assuming a sufficiently large value for $N_{c}$. In a
practical implementation, however, $N_{c}$ is a limited integer, and
therefore, $\theta_{\varphi}>0$ and $\varrho_{\varphi}>0$. Note that when the first or the second method in Section~\ref{sec:dcp:sch} is used, Property~\ref{p3} holds as a result of its preceding discussion, uniform convergence of the channel process, and finiteness of $|\mathcal{I}|$. Similar to Property~\ref{p2}, in the case of the third or the forth method, we assume this property results from the well-structuredness of algorithm $A$.

As the final step towards the main theorem, we define several random variables that are indirectly used in the theorem statement. Specifically, let $i_{\delta}$ be a geometric random variable with success probability $\delta^{'}$, where
\begin{align*}
\delta^{'}=(1-\varrho_{\varphi})^{2}\delta,
\end{align*}
where $\delta$ is defined in Section~\ref{sec:dcp:sch}.
In addition, let $i_{\varphi}$ be a r.v. with the following distribution.
\begin{align*}
P(i_{\varphi}=0)=\varrho_{\varphi},
\end{align*}
and
\begin{align*}
P(i_{\varphi}=k)=(1-\varrho_{\varphi})^{2k-1}(1-(1-\varrho_{\varphi})^{2}), \ k\geq 1.
\end{align*}
We also define the random sequence $\{N^{'}_{3}(i), \ i\geq 1\}$ as\footnote{Here, $ \wedge$ and $\vee$ are the \emph{and} and \emph{or} operators, respectively.}
\begin{align}
&N^{'}_{3}(i)
=\left\{\begin{array}{l l}
L_{1} & (1 \leq i \leq i_{\delta}) \ \vee
\\ & \qquad \quad  (i=i_{\delta}+i_{\varphi}+1) \\
1 & (i=i_{\delta}+1) \wedge (i_{\varphi}=1) \\
2 & (i=i_{\delta}+1) \wedge (i_{\varphi}>1) \\
\min(\frac{2^{i}}{2^{i_{\delta}+2}},L_{1})  & (i_{\delta}+2 \leq i \leq i_{\delta}+i_{\varphi})  \wedge
 \nonumber \\ & \qquad \qquad \qquad \qquad
     (i_{\varphi}>1) \\
     0 & i>i_{\delta}+i_{\varphi}+1
\end{array}
\right.\nonumber.
\end{align}
Using the above sequence, we define $R_{\infty}$ as
\begin{align}
R_{\infty}=\frac{\mathds{E}\big[\sum_{i=i_{\delta}+1}^{i_{\delta}+i_{\varphi}}\ N^{'}_{3}(i)\big]}{\mathds{E}\big[\sum_{i=1}^{i_{\delta}+i_{\varphi}+1}\ (1+N^{'}_{3}(i))]},\label{r:infty}
\end{align}
which plays a key role in theorem statement and its proof. Note that for a fixed $\delta>0$, we have
\begin{align*}
\lim_{\varrho_{\varphi}\to 0} R_{\infty}=\frac{L_{1}}{1+L_{1}}.
\end{align*}
As mentioned earlier, we can make $\varrho_{\varphi}$ and $\theta_{\varphi}$ arbitrarily small by choosing a sufficiently large value for $N_{c}$. We are now ready to state the theorem.

\subsection{Main Theorem on Stability of DCP}
We have the following theorem:
\begin{Theorem}\label{theorem1}
Consider a network as described in Section~\ref{sec:net}. For this network,
let $\theta$ be a constant defined by
\begin{align*}
\theta=R_{\infty} \inf_{\|\mathbf{X}\|=1}\ \frac{(\tilde{\phi}(\mathbf{X})-\alpha-3\theta_{\varphi})}{\chi(\mathbf{X})}.
\end{align*}
In addition, let $\theta_{\infty}$ be
\begin{align*}
\theta_{\infty}=\inf_{\|\mathbf{X}\|=1}\ \ \frac{\tilde{\phi}(\mathbf{X})}{\chi(\mathbf{X})}.
\end{align*}
\begin{itemize}
\item[{(}a{)}] If $6 \theta_{\varphi}<\alpha$ and $2\alpha \leq \inf_{\|\mathbf{X}\|=1} \tilde{\phi}(\mathbf{X})$, then the network is stable under DCP if the mean arrival rate vector, $\mathbf{a}$, lies strictly inside
the region $\theta \Gamma$.
\item[{(}b{)}] For any input rate strictly inside $\theta_{\infty} \Gamma$, there exist a sufficiently small value for $\alpha$, and sufficiently large values for $L_{1}$ and $N_{c}$ such that the network becomes stabilized under DCP. In other words, we can expand the sufficient stability region $\theta \Gamma$ arbitrarily close to $\theta_{\infty} \Gamma$ by choosing appropriate values for
for $\alpha$, $L_{1}$, and $N_{c}$.
\item[{(}c{)}] There exist instances of networks, as described in Section~\ref{sec:net}, for which their associated region $\theta_{\infty} \Gamma$ is maximally stable under DCP.
\end{itemize}
\end{Theorem}
\begin{proof}
The proof is provided in the Appendix.
\end{proof}
\vspace{-.15in}
\subsection{Discussion}
\subsubsection{Intuitive Explanation of $\theta$}\label{sec:dis:theta}
Theorem~\ref{theorem1} states that all input rates interior to $\theta \Gamma$ can be stably supported under DCP. In particular, it implicitly quantifies $\theta$ as a function of the sub-optimality of algorithm $A$ and channel state correlation. Clearly, the value of $\theta$ is not fixed, and can vary from a particular network setup to another. As expected, for a fixed $\mathbf{X}$, as algorithm $A$ finds better schedule vectors in shorter times, and as the channel states become more correlated, $\tilde{\phi}(\mathbf{X})$ gets closer to ${\chi(\mathbf{X})}$, and $\theta$ gets closer to one, expanding the region $\theta \Gamma$ to the capacity region $\mathcal{C}$.

In addition, Theorem~\ref{theorem1} shows how the stability region is directly affected by the choices for $\alpha$ and $L_{1}$, and the values for $\theta_{\varphi}$ and $\varrho_{\varphi}$. The impact of $\alpha$ on $\theta$ could be predicted by noting that the update rule uses $N^{r}_{1}$ in an update interval only when the normalized average backlog-rate product increases at least by $\alpha$. Thus, we expect to see a decrease of the type $\frac{\alpha}{\chi(\mathbf{X})}$ in the stability region scaling. The effect of $\theta_{\varphi}$ and $\varrho_{\varphi}$ is less obvious, but can be roughly explained as follows. Suppose at the $k_{\text{th}}$ round the optimal $N_{1}$ is selected, i.e., $N^{r}_{1}(\hat{t}_{k})=\tilde{N}_{1}(\hat{t}_{k})$. In this case, to have a proper comparison, $\varphi^{r}(\hat{t}_{k})$ and $\varphi(\hat{t}_{k-1})$ should satisfy their corresponding inequalities in Property~\ref{p3}. Moreover, to make sure that $N^{r}_{1}(\hat{t}_{k})$ or a near optimal $N_{1}$ is used in the $l_{\text{th}}$ round after the $k_{\text{th}}$, we at least require $\varphi^{r}(\hat{t}_{l})$ satisfy its corresponding inequality in Property~\ref{p3}. Therefore, there are at least three inequalities of the form in Property~\ref{p3} that should be satisfied, which results in the term $3 \theta_{\varphi}$ in the expression for $\theta$.

The factor $R_{\infty}$ in a sense measures the least fraction of time in update intervals where \emph{near} optimal values for $N_{1}$ is used. To better understand $R_{\infty}$, suppose $\varrho_{\varphi}$ is small, and the backlog vector is large. Once the optimal value for $N_{1}$ is found in a round, as long as the inequalities in Property~\ref{p3} hold for the subsequent rounds, $N_{1}$ gets updated for only a few times. By the update rule, this means that $N_{3}$ gets doubled in most of the rounds, and is likely equal to $L_{1}$. Thus, the update intervals constitute $\frac{L_{1}}{1+L_{1}}$ fraction of time. At the same time, in these intervals, near optimal values for $N_{1}$ are being used. Thus, we expect to see $\frac{L_{1}}{1+L_{1}}$ as a multiplicative factor in $\theta$.

The above discussion and Theorem~\ref{theorem1} also state that DCP successfully adapts $N_{1}$ in order to keep $\varphi(\hat{t}_{k}+N_{c},N_{1}(\hat{t}_{k}))$ close to $\tilde{\phi}(\mathbf{X}(\hat{t}_{k}+N_{c}))$\footnote{This statement is in fact a direct result of Lemma~4.}. Note that for a given $\mathbf{X}$ finding $\tilde{N}_{1}(\mathbf{X})$, or equivalently, $\tilde{\phi}(\mathbf{X})$, in general, is a difficult problem. Specifically, it requires the exact knowledge of the channel state and arrival process statistics, and the structure of algorithm $A$. Even when this knowledge is available, as the number of users increases, finding $\tilde{N_{1}}(\mathbf{X})$ demands computation over a larger number of dimensions, which becomes exponentially complex. Hence, we see that DCP dynamically solves a difficult optimization problem, without requiring the knowledge of input rates or the structure of algorithm $A$\footnote{DCP also does not require the exact knowledge of channel state statistics. However, a practical implementation of DCP requires $N_{c}$ to be related to the convergence-rate of channel process to its steady state.}.

\subsubsection{Comparison with Static Policies, Minmax v.s. Maxmin}\label{sec:dis:maxmin}
Part (b) of the theorem gives the region $\theta_{\infty} \Gamma$ as the fundamental lower-bound on the limiting performance of DCP. It also implicitly states that this lower-bound depends on the solution to a minmax problem. To see this, recall that by definition $\tilde{\phi}(\mathbf{X})$ is the maximum of $\phi(\mathbf{X}, N_{1})$ over all choices for $N_{1}$. Thus, we have that
\begin{align*}
\theta_{\infty}=\inf_{\|\mathbf{X}\|=1}\ \max_{N_{1}\in \mathcal{N}_{1}}\ \frac{\phi(\mathbf{X},N_{1})  }{\chi(\mathbf{X})}.
\end{align*}

Now, consider a \emph{static} policy that assumes a fixed value for $N_{1}$. This policy partitions the time axis into a set of frames each consisting of $N_{1}$ timeslots, with the $i_{\text{th}}$ frame starting at time $(i-1)N_{1}$. The static policy, in the beginning of each frame, e.g., the $i_{\text{th}}$ frame, provides algorithm $A$ with vectors $\mathbf{X}((i-1)N_{1})$ and $\mathbf{s}((i-1)N_{1})$. Algorithm $A$ uses these vectors as input, and after spending $N_{1}$ timeslots, returns a schedule vector as the output. This output vector is then used to schedule users in the next following frame.

It is not difficult to show that the above static policy stabilizes the network for all rates interior to $\theta^{s}_{N_{1}}\Gamma$, where
\begin{align*}
\theta^{s}_{N_{1}}=\inf_{\|\mathbf{X}\|=1}\frac{\phi(\mathbf{X},N_{1})}{\chi(\mathbf{X})}.
\end{align*}
Thus, the best static policy, in terms of the region $\theta^{s}_{N_{1}}\Gamma$, is the one that maximizes $\theta^{s}_{N_{1}}$. Let $\theta^{s}_{o}$ be the maximum value. We have that
\begin{align*}
\theta^{s}_{o}=\max_{N_{1}\in \mathcal{N}_{1}}\inf_{\|\mathbf{X}\|=1}\frac{\phi(\mathbf{X},N_{1})}{\chi(\mathbf{X})}.
\end{align*}
Therefore, the best static policy corresponds to a maxmin problem.
 Considering the definition of $\theta_{\infty}$ and $\theta^{s}_{o}$, and that the minmax of a function is always larger than or equal to the maxmin, we have that $\theta^{s}_{o}\Gamma \subseteq \theta_{\infty} \Gamma$. More generally, using the above definitions and a simple drift analysis, we can show that the stability region of static policies is not larger than the limiting stability
 region of DCP.

 \subsubsection{Tightness of $\theta_{\infty}$ and $\theta_{o}^{s}$}Note that parts (a) and (b) of the theorem do not exclude the possibility of networks being stable under DCP for rates outside of $\theta \Gamma$ or $\theta_{\infty} \Gamma$. Part (c) of the theorem, on the other hand, compliments parts (a) and (b), and shows that for some networks the region $\theta_{\infty}\Gamma$ is indeed the largest \emph{scaled} version of $\Gamma$ that can be stably supported under DCP. This for instance may happen when the channel state is statistically symmetric with respect to users as the ones in Section~\ref{sec:casestudy}. Proof of part~(c) of the theorem provides conditions for cases that lead to the maximal stability of the region $\theta_{\infty} \Gamma$, and in particular, shows that the symmetric examples in Section~\ref{sec:casestudy} meet such conditions. Note that the same discussion also applies to $\theta^{s}_{N_{1}}$ and the stability region of static policies. We therefore have $\theta_{\infty}$ and $\theta^{s}_{o}$ both as tight measures, stating that for some networks, including the ones in the next section, DCP can increase throughput efficiency of static policies by a factor of $\frac{\theta_{\infty}-\theta^{s}_{o}}{\theta^{s}_{o}}$.

\subsubsection{Delay}
Note that getting close to the boundary of $\theta_{\infty}\Gamma$ increases delay. This follows from part (b) of the theorem stating that for input rates close to the boundary, $L_{1}$ and $N_{c}$ should be large. These choices, as expected, increase the length of test and update intervals, which can potentially be large intervals of sub-optimal transmissions in terms of the value used for $N_{1}$. This in turn makes data wait in queues before transmission, thus increasing the delay.

\subsubsection{Distributed Implementation}
Assuming algorithm $A$ is decentralized \cite{wu:srikant:05}\cite{lin:shroff:ton06}\cite{mazumdar:shroff:06}\cite{shroff:mazummdar:info:07}, DCP can be implemented in a distributed manner with low overhead. This is possible since consistent implementation of DCP in all nodes requires updates of only queue backlog and nodes' time-average of backlog-rate product, and such updates are needed only over \emph{long} time intervals.

More specifically, two conditions are required to be met for distributed implementation. First, nodes should generate the same sequence of random candidates for $N_{1}$ over time, which can be met by assuming the same number generator is employed by all nodes. Second, nodes should have the knowledge of backlog-rate product in the test and its preceding update interval in order to individually and consistently apply the update rule.

 The second condition can also be met, for instance, by requiring each source node perform the following. Every node, e.g., the $i_{\text{th}}$ node, records its own backlog, $X_{i}(t)$, only at the beginning of the test and update intervals. During these intervals, the $i_{\text{th}}$ node also computes its own \emph{individual} time-average of backlog-rate product $X_{i}D_{i}$. Here, we assume the time-averages in the test intervals are computed up to the last $N_{d}$ timeslots, where $N_{d}\ll N_{c}$. Then, once an update interval ends, the $i_{\text{th}}$ node has all the duration of a test interval, consisting of $N_{c}$ timeslots, to send all the other nodes $X_{i}$ and time-average of $X_{i}D_{i}$ for that update interval. Similarly, when the last $N_{d}$ timeslots in a test interval are reached, the $i_{\text{th}}$ node starts sending all the other nodes $X_{i}$ and time-average of $X_{i}D_{i}$ of that test interval, hence, having $N_{d}$ timeslots for communication. Since for each interval, data of each node, backlog in the beginning of the associated interval and the time-average, consists of at most a few bytes, we see that the overhead can be made arbitrarily small by choosing $N_{c}$ and $N_{d}$ large. At the same time, we can make the ratio $\frac{N_{d}}{N_{c}}$ sufficiently small, by choosing $N_{c}$ large, to ensure that not consideration of the last $N_{d}$ timeslots in the test intervals has little impact on the stability region.


\vspace{-.15in}
\section{Case Studies}\label{sec:casestudy}
In this section, we present two examples that provide further insight into our analytical results and the performance of DCP. To be able to compare the simulation results with analytical ones, we consider a small network consisting of two data flows in the downlink of a wireless LAN or a cellular network. In this case, $\mathbf{s}(t)$ is the vector of channel gains, and we assume the schedule vector is the power allocation vector, i.e., $\mathbf{I}=\mathbf{P}=(p_{1},p_{2})$, with constraint
\begin{align*}
p_{1}+p_{2}=P_{t},
\end{align*}
where $P_{t}$ is total power budget.
Assuming super-position coding is used in the downlink, if $s_{1}(t)<s_{2}(t)$, then \cite{tse}
\begin{align*}
D_{1}(\mathbf{s}(t),\mathbf{P})=\log\Big(1+\frac{p_{1}|s_{1}|^{2}}{p_{2}|s_{1}|^{2}+n_{0}}\Big),
\end{align*}
and
\begin{align*}
D_{2}(\mathbf{s}(t),\mathbf{P})=\log\Big(1+\frac{p_{2}|s_{2}|^{2}}{n_{0}}\Big).
\end{align*}
If $s_{1}(t)\geq s_{2}(t)$, we obtain similar expressions for user rates by swapping the role of one user for another.

For illustration purposes, we assume that algorithm $A$ in every step, i.e., during each timeslot, reduces the gap to the optimal backlog-rate product. Specifically, if the initial gap corresponding to the initial power vector $\mathbf{I}^{(0)}$, assumed to be chosen randomly, is $\Delta_{0}$, then after $i$ steps the gap is decreased to $\Delta_{i}$, where
\begin{align*}
\Delta_{i}&=\mathbf{X}\mathbf{D}^{*}(\mathbf{X},\mathbf{s})-\mathbf{X}\mathbf{D}(\mathbf{s},\mathbf{I}^{(n)})\\
          &=\frac{1}{\beta^{i}}\big(\mathbf{X}\mathbf{D}^{*}(\mathbf{X},\mathbf{s})-\mathbf{X}\mathbf{D}(\mathbf{s},\mathbf{I}^{(0)})\big)
          =\frac{\Delta_{0}}{\beta^{i}},
\end{align*}
where $\beta>1$. This case corresponds to $g(n)=(1-\zeta^{i})$ with $\zeta=\frac{1}{\beta}$, where $g(n)$ is introduced in Section~\ref{sec:algorithmA}.

Having specified rates and algorithm $A$, as the first example, we assume that the channel state is Markovian with two possible state vectors, namely, $\mathbf{s}_{1}=(1,5)$ and $\mathbf{s}_{2}=(5,1)$, where the channel vector in each transition takes a different state with probability $p_{t}=0.3$. For this case, we set $\alpha=0.06$, $N_{c}=12000$, $L_{1}=32$, $\beta=1.7$, $\mathcal{N}_{1}=\{N_{1}: 1\leq N_{1}\leq 6\}$, $n_{0}=10$, and $p_{t}=50$. To study the stability region, we consider the rate vector $\mathbf{a}=(2.4181,2.4181)$ which belongs to the boundary of $\Gamma$ corresponding to this example. We then assume the arrival vector is $\gamma\mathbf{a}$, where $\gamma$ is the load factor, and varies from $0.84$ to $0.92$. Fig.~\ref{fig:2state} depicts the resulting average queue sizes. For loads larger than $0.93$, the queue sizes increase with time implying network instability. The range selected for $\gamma$ is motivated by noting that $\theta_{\infty}=0.9447$, which is computed numerically. Considering the growth of average queue sizes in Fig.~\ref{fig:2state}, we therefore see that for this example $\theta_{\infty}$ is indeed an upper bound for capacity region scaling.
In fact, part~(c) of Theorem~\ref{theorem1} applies to this example, and any rate of the form $(\theta_{\infty}+\epsilon) \mathbf{a}$, $\epsilon>0$, makes the network unstable.

As for the second example, we increase the number of states to six corresponding to the following state vectors:
\begin{align*}
\mathbf{s}_{1}=(1,5),&\quad \mathbf{s}_{2}=(5,1),\\
\mathbf{s}_{3}=(1,2),&\quad \mathbf{s}_{4}=(2,1),\\
\mathbf{s}_{5}=(2,5), &\quad \mathbf{s}_{6}=(5,2),
\end{align*}
and having the following symmetric transition matrix:
\begin{eqnarray}
 \mathbf{T_{m}}= \left( \begin{array}{cccccc}
0.3 &0.1 &0.2 &0.1 &0.2 &0.1 \\
 0.1 &0.3 &0.1 &0.2 &0.1 &0.2\\
\vdots & \vdots & \vdots & \vdots & \ddots \\
 0.1 &0.2 &0.1 &0.2 &0.1 &0.3
\end{array} \right).
\end{eqnarray}
For this case, we keep the same $N_{c}$, $L_{1}$, and $\mathcal{N}_{1}$, but assume $\alpha=0.02$, $\beta=1.5$, $n_{0}=50$, and $p_{t}=10$. Similar to the previous example, to vary arrival rate vector, we consider the rate vector $\mathbf{a}=(0.6952,0.6952)$ which belongs to the boundary of $\Gamma$ associated with this example. Then, the arrival vector is assumed to be $\gamma\mathbf{a}$, where the load factor $\gamma$ varies from $0.67$ to $0.76$. The resulting average queue sizes are also shown in Fig.~\ref{fig:2state}. In this case, for load factors larger than $0.76$, the queue sizes increase with time, suggesting network instability. This result is consistent with our analytical results since the numerically computed value of $\theta_{\infty}$ is $0.7762$. Note that part~(c) of Theorem~\ref{theorem1} also applies to this example, and any rate of the form $(\theta_{\infty}+\epsilon) \mathbf{a}$, $\epsilon>0$, makes the network unstable.

Finally, in Fig.~\ref{fig:theta}, for the two examples, we have shown $\theta^{s}_{N_{1}}$ as a function of $N_{1}$, and also shown the value of $\theta_{\infty}$ for DCP. As expected and the figure suggests, since DCP adapts $N_{1}$ according to queue backlog, it outperforms the best static policy. We also see that the optimal stationary policy for the first example is the one with $N_{1}=3$ and $\theta^{s}_{o}=0.9122$, and for the second example is the one with $N_{1}=2$ and $\theta^{s}_{o}=0.7511$. Note that characterization of the best static policy requires computation of $\tilde{\phi}(\mathbf{X})$, which, as briefly discussed in Section~\ref{sec:dis:theta}, can be computationally intensive. From the figure, we also observe that the performance of a suboptimal static policy can be substantially less than DCP if the static policy does not assume a proper value for $N_{1}$.



\begin{figure}[t]
\centering
\includegraphics[width=2.3in]{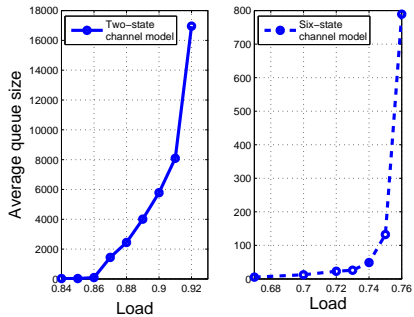}
\caption{Average queue size as a function of load factor.}\label{fig:2state}
 \vspace{-.1in}
\end{figure}
\begin{figure}[t]
\centering
\includegraphics[width=2.3in]{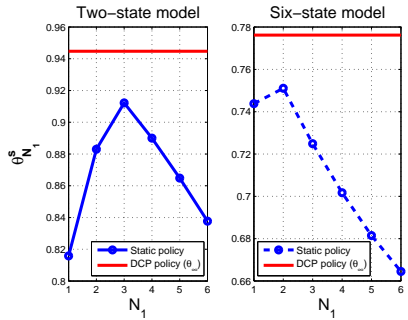}
\caption{Comparison of capacity region scaling for DCP and static policies.}\label{fig:theta}
 \vspace{-.1in}
\end{figure}

\vspace{-.1 in}
\section{Conclusion}\label{sec:conclusion}
In this paper, to improve the stable throughput region in practical network setups, we have considered the problem of scheduling in time-varying networks from a new perspective. Specifically, in contrast to previous research which assumes the search-time to find schedule vectors is negligible, we have considered this time, based on which we
modeled the time-efficiency of sub-optimal algorithms. Inspired by this modeling, we have proposed a dynamic control policy that dynamically but in a large time-scale tunes the time given to an available sub-optimal algorithm according to queue backlog and channel correlation. Remarkably, this policy does not require knowledge of input rates or the structure of available sub-optimal algorithms, nor it requires exact statistics of the channel process. We have shown that this policy can be implemented in a distributed manner with low overhead. In addition, we have analyzed the throughput stability region of the proposed policy and shown that its throughput region is at least as large as the one for any other, including the optimal, static policy. We believe that study and design of similar policies opens a new dimension in the design of scheduling policies, and in parallel to the efforts to improve the performance of sub-optimal algorithms, can help boost the throughput performance to the capacity limit.

\appendix[Proof of Theorem~\ref{theorem1}]
\begin{proof}[Proof of part (a)]

The proof of part (a) consists of two main parts. First, using several lemmas, we obtain a negative drift with a random number of steps. In the second part, we use the negative drift analysis to show that the return time to a bounded region has a finite expected value, and conforms to the properties required for network stability, according to the definition given in Section~\ref{sec:per:stab}.

We start by noting that $\theta\leq 1$, and since $\mathbf{a}$ is strictly inside $\theta\Gamma$, there must be some non-negative constants $\beta_{\mathbf{s},\mathbf{I}}$ with the property that for all $\mathbf{s}\in
\mathcal{S}$
\begin{align}
\sum_{\mathbf{I}\in\mathcal{I}}\beta_{\mathbf{s},\mathbf{I}} <\theta\leq 1,\label{beta}
\end{align}
such that
\begin{align}
 \mathbf{a}=\sum_{\mathbf{s}\in
\mathcal{S}}\pi(\mathbf{s})\sum_{\mathbf{I}\in\mathcal{I}}
\beta_{\mathbf{s},\mathbf{I}}\mathbf{D}_{\mathbf{s},\mathbf{I}}.\label{a:vector}
\end{align}
Considering (\ref{beta}), we can define positive $\xi'$ as
\begin{align*}
\xi'=\theta-\max_{\mathbf{s} \in \mathcal{S}}\sum_{\mathbf{I} \in
\mathcal{I}}\beta_{\mathbf{s},\mathbf{I}}.
\end{align*}
Since $\xi'>0$, by the definition of $\theta$, for $\|\mathbf{X}_{t}\|\neq 0$, we have that
\begin{align}
\frac{R_{\infty}(\tilde{\phi}(t)-\alpha-3\theta_{\varphi})}{\chi(\mathbf{X}_{t})}
-\max_{\mathbf{s}\in \mathcal{S}}\sum_{\mathbf{I} \in
\mathcal{I}}\beta_{\mathbf{s},\mathbf{I}} \  >\xi'>0.\label{xi:1}
\end{align}

To proceed with the proof, associated with a given time $t$,
 we define a sequence of random variables $\{\tau_{i}\}_{i=-1}^{\infty}$, where $\tau_{-1}$ and $\tau_{0}$
 denote the number of timeslots to the last timeslot of the previous and the current scheduling round, respectively, and $\tau_{i}$, $i\geq 1$, is
 the number of timeslots to the last timeslot of the $i_{\text{th}}$ subsequent scheduling round. Let $\mathcal{H}_{t}$ denote the past history of the system up to and including time $t$. Thus, given $\mathcal{H}_{t}$, the value of $\mathbf{X}_{t}$ is known. Let $f(\cdot)$ be defined as
 \begin{align*}
 f(\mathbf{X})=\|\mathbf{X}\|^{2},
 \end{align*}
 Considering a $\tau_{K}+1$-step drift with function $f(\cdot)$, we can write
 \begin{align}
\Delta(\tau_{K}+1)&=\mathds{E}[f(\mathbf{X}_{t+\tau_{K}+1})-f(\mathbf{X}_{t})|\mathcal{H}_{t}]
\nonumber \\
&=\mathds{E}[\sum_{k=0}^{\tau_{K}}f(\mathbf{X}_{t+k+1})-f(\mathbf{X}_{t+k})|\mathcal{H}_{t}] \nonumber \\
&=\mathds{E}[\sum_{k=0}^{\tau_{K}}(\mathbf{X}_{t+k+1}+\mathbf{X}_{t+k})(\mathbf{X}_{t+k+1}-\mathbf{X}_{t+k})|\mathcal{H}_{t}]. \nonumber
 \end{align}

  Using the fact that arrivals and departures are bounded, after performing some
 preliminary steps, we can show that
\begin{align*}
&\Delta(\tau_{K}+1)
\nonumber \\&\leq\mathds{E}\bigg[(\tau_{K}+1)C_{1}+(\tau_{K}+1)^{2}C_{2}
\nonumber \\ &\quad \qquad \qquad \qquad+2
\sum_{k=0}^{\tau_{K}}(\mathbf{X}_{t}\mathbf{A}_{t+k}-\mathbf{X}_{t+k}\mathbf{D}_{t+k})\ \big|\mathcal{H}_{t}\bigg],
\end{align*}
for appropriate constants $C_{1}$ and $C_{2}$. Since $\mathbf{X}_{t+k}\mathbf{D}_{t+k}\geq 0$, we have
\begin{align*}
\Delta(\tau_{K}+1)&\leq\mathds{E}\bigg[(\tau_{K}+1)C_{1}+(\tau_{K}+1)^{2}C_{2}
 \nonumber\\&\quad \qquad \qquad
 +2\sum_{k=0}^{\tau_{K}}(\mathbf{X}_{t}\mathbf{A}_{t+k}-\mathbf{X}_{t}\mathbf{a})
  \nonumber \\&\quad \qquad \qquad +2\sum_{k=0}^{\tau_{K}}(\mathbf{X}_{t}\mathbf{a}- \mathbf{X}_{t+k}\mathbf{D}^{*}_{t+k})
  \\\nonumber&
\qquad \qquad\quad  + 2\sum_{k=0}^{\tau_{K}}\mathbf{X}_{t+k}\mathbf{D}^{*}_{t+k}
 \\\nonumber&
\qquad \qquad\quad   - 2\sum_{k=\tau_{0}+1}^{\tau_{K}}\mathbf{X}_{t+k}\mathbf{D}_{t+k}
\ \big|\mathcal{H}_{t}\bigg],
\end{align*}
where $\mathbf{D}^{*}_{t+k}=\mathbf{D}^{*}(\mathbf{X}(t+k),\mathbf{s}(t+k))$. In the following, we derive an upper bound for $\Delta(\tau_{K}+1)$.

As mentioned in Section~\ref{sec:net:queue}, arrivals are i.i.d with mean vector $\mathbf{a}$. We can therefore
apply the same method used to prove Lemma~\ref{prelemma} to obtain
\begin{align*}
\mathds{E}\big[\|\sum_{k=0}^{\tau_{K}}\mathbf{A}_{t+k}-(\tau_{K}+1)\mathbf{a}\| \ \big|\mathcal{H}_{t}\big]\leq \epsilon \mathds{E}\big [(\tau_{K}+1)|\mathcal{H}_{t}\big],
\end{align*}
where $\epsilon>0$, and can be made arbitrarily small by choosing a sufficiently large $K$.

Using the above inequality, Lemma~\ref{lemma:D*}, Lemma~\ref{lemma:2}, and Lemma~\ref{lemma:3}, all with the same choice for $\epsilon$,
we can show that
\begin{align}
&\Delta(\tau_{K}+1)\leq
\mathds{E}\Big[(\tau_{K}+1)\|\mathbf{X}_{t}\|\chi(\mathbf{X}_{t})
\nonumber \\& \Big(  \epsilon_{1}+2\big(\max_{\mathbf{s}\in \mathcal{S}}\sum_{\mathbf{I} \in
\mathcal{I}}\beta_{\mathbf{s},\mathbf{I}}-\frac{R_{\infty}(\tilde{\phi}(t)-\alpha-3\theta_{\varphi})}{\chi(\mathbf{X}_{t})}\ \big) \Big) \big|\mathcal{H}_{t}\Big]\label{dd:final},
\end{align}
where
\begin{align}
\epsilon_{1}=\frac{1}{\chi(\mathbf{X}_{t})}\Big(\frac{C_{1}}{\|\mathbf{X}_{t}\|}
+\frac{C_{2}(\tau_{K}+1)}{\|\mathbf{X}_{t}\|}+8\epsilon\Big).
\end{align}
Note that according to the lemmas, $\epsilon$ can take any given positive real number if $K$ and $\|\mathbf{X}_{t}\|$ are sufficiently large.

Similarly, $\epsilon_{1}$ can assume any given positive value. To see this, first note that since $\mathbf{a} \in
\theta \Gamma $, we have $\mathbf{a} \in
\Gamma $. Thus, for any user, e.g. the $i_{\text{th}}$ user, for which $\mathbf{a}_{i}>0$, there has to be a
state $\mathbf{s}$ and a schedule $\mathbf{I}$ satisfying
\begin{align*}
\pi(\mathbf{s})\mathbf{D}(\mathbf{s},\mathbf{I})_{i}>0,
\end{align*}
where $\mathbf{D}(\mathbf{s},\mathbf{I})_{i}$ is the $i_{\text{th}}$ element of vector $\mathbf{D}(\mathbf{s},\mathbf{I})$.
Otherwise, $\mathbf{a}_{i}$ should be zero, contradicting the assumption. Therefore, assuming $\mathbf{a}\neq \mathbf{0}$, we can define positive
$\upsilon$ as
\begin{eqnarray*}
\upsilon=\min_{i\in N} \max_{\mathbf{s},\mathbf{I}} \pi(\mathbf{s})\mathbf{D}(\mathbf{s},\mathbf{I})_{i}>0. \nonumber
\end{eqnarray*}
Thus,
\begin{eqnarray}
\mathds{E}[\mathbf{X}_{t}\mathbf{D}^{*}(\mathbf{X}_{t},\mathbf{s})]\geq \upsilon \max_{i\in
N}\mathbf{X}(t)_{i}\geq \frac{\upsilon}{\sqrt{N}}
\|\mathbf{X}_{t}\|.\label{EXt}
\end{eqnarray}
This implies that for all nonzero $\mathbf{X} \in \mathds{R}^{N}$
\begin{eqnarray}
\chi(\mathbf{X})\geq \frac{\upsilon}{\sqrt{N}}.\label{chi}
\end{eqnarray}
On the other hand, since departure rates are bounded above by $D_{max}$, we have
\begin{align}
\chi(\mathbf{X})\leq \sqrt{N}D_{max}.
\end{align}
Now consider any positive $\epsilon_{2}$, and suppose $K$ is sufficiently large such that for large $\|\mathbf{X}_{t}\|$ we have
\begin{align*}
\frac{8 \epsilon}{\chi(\mathbf{X}_{t})} \leq \frac{8\sqrt{N}\epsilon}{\upsilon}<\frac{\epsilon_{2}}{3},
\end{align*}
where the first inequality follows from (\ref{chi}).
This upper-bounds the third term in $\epsilon_{1}$.
 Since for any $K$, and in particular, the chosen one, we have $\tau_{K}+1\leq (K+1)(1+L_{1})N_{c}$, we see that if $\|\mathbf{X}_{t}\|$ is appropriately large, the first and second terms in $\epsilon_{1}$ can also be less than $\frac{\epsilon_{2}}{3}$. Thus, for any given positive $\epsilon_{2}$, we can find
an appropriately large $K$ such that for sufficiently large $\|\mathbf{X}_{t}\|$, (\ref{dd:final}) holds with $\epsilon_{1}<\epsilon_{2}$.

Suppose $K$ is sufficiently large, and $\|\mathbf{X}_{t}\|>M_{K}$ for
appropriately large $M_{K}$ such that $\epsilon_{1}<\xi'$. We can use (\ref{dd:final}) and (\ref{xi:1})
to show that
\begin{align}
&\Delta(\tau_{K}+1) <-\mathds{E}\Big[\xi'\|\mathbf{X}_{t}\|(\tau_{K}+1)\chi(\mathbf{X}_{t}) \ \big|\mathcal{H}_{t}\big]
.\nonumber
\end{align}
This inequality and (\ref{chi}) further imply that
\begin{align}
&\Delta(\tau_{K}+1)<
-\mathds{E}\Big[\xi(\tau_{K}+1)
\|\mathbf{X}_{t}\|\big|\mathcal{H}_{t}\Big],\label{last:ineq}
\end{align}
where $\xi=\frac{\upsilon}{\sqrt{N}}\xi'>0$. We, therefore, have obtained the negative drift expression, completing the first part of the proof.

Note that in above $\tau_{K}$ is a random variable, and in fact, is a stopping time with respect to the filtration
$\mathcal{H}=\{\mathcal{H}_{t}\}_{t=0}^{\infty}$. This means that we have obtained a drift expression that is based on a random number of steps. Proofs of stability in the literature, however, are often based on a negative drift with a fixed number of steps. This contrast has motivated us to adopt an interesting method recently developed in \cite{fralix}. This method is general since it can be applied in both cases, and also leads to an intuitive notation of stability. However,
it has been originally developed for Markov chains. Therefore, as well
as using less technical notations, in what follows, we apply minor modifications to the method so
that it is appropriate in our context.

We now, in the second part of the proof, use the negative drift, and prove that the expected value of the return time to some bounded region
is finite in a manner that renders network stable. Let $\mathcal{C}$ denote the bounded region, and be defined as
\begin{align*}
 \mathcal{C}=\{X\in
\mathds{R}^{N}, \|X\|\leq M_{K}\}.
\end{align*}
Associated with $\mathcal{C}$, we define
$\sigma_{\mathcal{C}}$ to be the
number timeslots after which the process $\{\mathbf{X}_{t+i}\}_{i=0}^{\infty}$
\emph{enters} $\mathcal{C}$, i.e.,
\begin{align*}
\sigma_{\mathcal{C}}=\inf\{i\geq 0:\mathbf{X}_{t+i}\in
\mathcal{C}\}.
\end{align*}
Similarly, we let $\tau_{\mathcal{C}}$ be
\begin{align*}
\tau_{\mathcal{C}}=\inf\{i\geq 1:\mathbf{X}_{t+i}\in
\mathcal{C}\}.
\end{align*}
 Therefore, $\tau_{\mathcal{C}}$, in contrast $\sigma_{\mathcal{C}}$, characterizes the first time that the process $\{\mathbf{X}_{t+i}\}_{i=1}^{\infty}$ \emph{returns} to $\mathcal{C}$.

 Back to the drift expression in (\ref{last:ineq}), let $\eta$ be a random variable defined by
\begin{align}
\eta=\xi (\tau_{K}+1)\|\mathbf{X}_{t}\|.\nonumber
\end{align}
We obtain, for $K$ sufficiently large,
\begin{eqnarray}
\mathds{E}[f(\mathbf{X}_{t+\tau_{K}+1})+\eta |\mathcal{H}_{t}]\leq
f(\mathbf{X}_{t}), \label{main:ineq}
\end{eqnarray}
provided that $\|\mathbf{X}_{t}\|>M_{K}$.
 Let $\eta_{0}=\eta$, and $\tau_{K,0}=\tau_{K}$, where $\eta$ and $\tau_{K}$ are random variables defined by considering time $t$.
 We now consider time $t_{K}^{(1)}=t+\tau_{K,0}+1$. For this particular time, we can define another pair $\tau_{K,1}$ and $\eta_{1}$ and  such that if $\|\mathbf{X}_{t_{K}^{(1)}}\|>M_{K}$, then
\begin{eqnarray}
\mathds{E}[f(\mathbf{X}_{t_{K}^{(1)}+\tau_{K,1}+1})+\eta_{1} |\mathcal{H}_{t_{K}^{(1)}}]\leq
f(\mathbf{X}_{t_{K}^{(1)}}),\nonumber
\end{eqnarray}
where $\tau_{K,1}$ is the number of timeslots from time $t_{K}^{(1)}$ to the last timeslot of the $K_{\text{th}}$ subsequent scheduling round, and
\begin{align*}
\eta_{1}=\xi (\tau_{K,1}+1)\|\mathbf{X}_{t_{K}^{(1)}}\|.
\end{align*}
Note that the definition of $\tau_{K,1}$ and $\eta_{1}$ is independent of whether the previous inequality holds.

 We can continue this process by considering the drift criteria for time $t_{K}^{(i)}=t_{K}^{(i-1)}+\tau_{K,i-1}+1$, and defining random variables $\tau_{K,i}$ and $\eta_{i}$. The random variables $\tau_{K,i}$ and $\eta_{i}$ have a similar definition as $\tau_{K,1}$ and $\eta_{1}$, respectively, except that they are associated with time $t_{K}^{(i)}$. Using these definitions, we can define $t_{K}^{(i)}$ more precisely by
\begin{align*}
&t_{K}^{(0)}=t,\nonumber \\
&t_{K}^{(i)}=t_{K}^{(i-1)}+(\tau_{K,i-1}+1)=t+\sum_{j=0}^{i-1}(\tau_{K,j}+1).
\end{align*}
Note that $t_{K}^{(i)}$ is a stopping time with respect
to $\mathcal{H}$.
Using $t_{K}^{(i)}$, we set
\begin{align}
\mathbf{\bar{X}}_{i}=\mathbf{X}_{t_{K}^{(i)}}, \ i\geq
0,
\end{align}
and define $\mathcal{H}^{\tau}$ as the filtration
given by $\mathcal{H}^{\tau}=\{\mathcal{H}_{
t_{K}^{(i)}}\}_{i=0}^{\infty}$.
 In addition, associated with $\eta_{i}$, which is given by
 \begin{align*}
\eta_{i}=\xi (\tau_{K,i}+1)\|\mathbf{X}_{t_{K}^{(i)}}\|,
 \end{align*}
we define $\eta^{(i)}$ as
\begin{eqnarray}
\eta^{(0)}=0,\ \ \eta^{(i)}=\sum_{j=0}^{i-1}\eta_{j}.\label{eta}
\end{eqnarray}
We also define $\nu$ as
\begin{align}
\nu=\inf\{i\geq 0:
t_{K}^{(i)}\geq t+\sigma_{C}\},\label{nu}
\end{align} which is a stopping
time with respect to $\mathcal{H}^{\tau}$. Intuitively, $\nu$ marks the first time $t_{K}^{(i)}$ at or before which the process $\{\mathbf{X}_{t+i}\}_{i=0}^{\infty}$ enters $\mathcal{C}$. We finish the chain of definitions by introducing
the sequence $\{Z_{i}\}_{i=0}^{\infty}$,
where
\begin{align}
Z_{i}=f(\mathbf{\bar{X}}_{i})+\eta^{(i)}.
\end{align}
For $i < \nu$, using (\ref{eta}), we have
\begin{align}
\mathds{E}[Z_{i+1}|\mathcal{H}_{t_{K}^{(i)}}]
&=
\mathds{E}\big[f(\mathbf{\bar{X}}_{i+1})+\eta_{i}|\mathcal{H}_{t_{K}^{(i)}}\big]
+\eta^{(i)} \nonumber
\\& \leq
f(\mathbf{\bar{X}}_{i})+\eta^{(i)}
=Z_{i},
\label{sup:martingle}
\end{align}
where the first equality follows from the fact that $\eta^{(i)}$ is completely determined given
$\mathcal{H}_{t_{K}^{(i)}}$, and the inequality is simply an immediate result of
(\ref{main:ineq}) and the assumption $i<\nu$. To simplify the notation, let $\nu \wedge i$ denote
\begin{align}
\nu \wedge i= \min(\nu,i).\nonumber
\end{align}
It now follows directly from (\ref{sup:martingle}) that the sequence
$\{Z_{\nu\wedge i}\}_{i=0}^{\infty}$ is an
$\mathcal{H}^{\tau}$-supermartingale. Since $f(\cdot)$ is non-negative, we have
\begin{align*}
\mathds{E}[\eta^{(\nu \wedge
i)}|\mathcal{H}_{t}]\leq \mathds{E}[Z_{\nu \wedge
i}|\mathcal{H}_{t}]
.
\end{align*}
But $\mathcal{H}_{t}=\mathcal{H}_{t_{K}^{(0)}}$, and $\{Z_{\nu\wedge i}\}_{i=0}^{\infty}$ is a supermartingale. Hence,
\begin{align*}
\mathds{E}[Z_{\nu \wedge
i}|\mathcal{H}_{t}]& =\mathds{E}[Z_{\nu \wedge
i}|\mathcal{H}_{t_{K}^{(0)}}]
\\ &\leq Z_{0}=f(\mathbf{X}_{t}).
\end{align*}
Considering the last two inequalities, we obtain
\begin{align}
\mathds{E}[\eta^{(\nu \wedge
i)}|\mathcal{H}_{t}]\leq f(\mathbf{X}_{t}).\label{tau:pre}
\end{align}
In addition, using
the definition of $\eta^{(i)}$ and $\eta_{j}$ while assuming $M_{K}>1$, it is easy to see that
\begin{align}
\eta^{(\nu \wedge i)}=\sum_{j=0}^{
i-1}\eta_{j}\mathbf{1}_{(j<\nu)}&\geq \xi\sum_{j=0}^{
i-1}(\tau_{K,j}+1)\mathbf{1}_{(j<\nu)} \nonumber
\\&=\xi(t_{K}^{(\nu \wedge i)}-t). \label{tau:pre2}
\end{align}
Applying the monotone convergence theorem \cite{borovkov}, we can take the limit in (\ref{tau:pre}) and (\ref{tau:pre2}) as $i\rightarrow \infty$ yielding
\begin{eqnarray*}
\mathds{E}[t_{K}^{(\nu)}-t|\mathcal{H}_{t}]\leq
\xi^{-1} f(\mathbf{X}_{t}).
\end{eqnarray*}
But by definition in (\ref{nu}), $\sigma_{\mathcal{C}} \leq t_{K}^{(\nu)}-t$. Thus, for $\mathbf{X}_{t}\notin \mathcal{C}$
\begin{eqnarray}
\mathds{E}[\sigma_{\mathcal{C}}|\mathcal{H}_{t}]\leq \xi^{-1}
f(\mathbf{X}_{t})\label{sigma_C}.\nonumber
\end{eqnarray}
If $\mathbf{X}_{t}\in \mathcal{C}$, we have $\sigma_{\mathcal{C}}=0$. Hence, we have that
\begin{eqnarray}
\mathds{E}[\sigma_{\mathcal{C}}|\mathcal{H}_{t}]\leq \xi^{-1}
f(\mathbf{X}_{t})\mathbf{1}_{\mathbf{X}_{t}\notin \mathcal{C}},\nonumber
\end{eqnarray}
showing that the expected $\sigma_{\mathcal{C}}$ is bounded by a function of $\mathbf{X}_{t}$ uniformly in the past history and $t$, as required. This completes the proof of part (a) of the theorem.
\end{proof}

 \begin{proof}[Proof of part (b)]
 Part~(b) follows directly from part~(a) of the theorem as a corollary by noting that $\theta_{\varphi}$ and $\varrho_{\varphi}$ can be made arbitrarily small by assuming a sufficiently large $N_{c}$, as stated in Property~\ref{p3}. This allows us to select arbitrarily small values for $\alpha$. In addition, we can chose a sufficiently large value for $L_{1}$ such that for sufficiently small values for $\theta_{\varphi}$ and $\varrho_{\varphi}$, $R_{\infty}$ is arbitrarily close to one. Considering these choices, we see that we can make $\theta$ arbitrarily close to $\theta_{\infty}$, as required.
\end{proof}
 \begin{proof}[Proof of part (c)]
 Since part (c) of the theorem only concerns existence of such networks for which the region $\theta_{\infty}\Gamma$ is maximally stable under DCP, for simplicity of exposition, we consider a network consisting of two users, i.e., two data flows. Note that
our approach can be extended to more general networks with $N$ data flows. Here, we adopt a direct method and show that with positive
probability norm of the backlog vector approaches infinity. Therefore, the expected value of the return time to any bounded region becomes infinity, implying network instability. We start by introducing several definitions followed by four conditions sufficient for network instability.

Let $\bar{\mathbf{D}}$ and $\bar{\mathbf{D}}^{*}$ be defined by
\begin{align*}
\bar{\mathbf{D}}(\mathbf{X})&=\mathds{E}\Big[\lim_{n \to \infty}\frac{1}{n}\sum_{i=0}^{n-1}(\mathbf{D}_{t+i}-\mathbf{U}_{t+i})\ \big| \mathbf{X}_{t+i}=\mathbf{X}, i\geq 0\Big],
\end{align*}
and
\begin{align*}
\bar{\mathbf{D}}^{*}(\mathbf{X})&=\mathds{E}_{\mathbf{s}}\big[\mathbf{D^{*}}(\mathbf{X},\mathbf{s})\big],
\end{align*}
where $\mathbf{D^{*}}(\mathbf{X},\mathbf{s})$ is defined in (\ref{dstarv}).
In addition, let $\mathbf{X}_{min}$ be\footnote{Note that here infimum can be achieved since the functions $\tilde{\phi}(\mathbf{X})$ and $\chi(\mathbf{X})$ are continuous functions of $\mathbf{X}$, and the infimum is taken over a closed interval.}
\begin{align*}
\mathbf{X}_{min}=\arginf_{\|\mathbf{X}\|=1} \ \frac{\tilde{\phi}(\mathbf{X})}{\chi(\mathbf{X})}.
\end{align*}
Note that in the definition of $\bar{\mathbf{D}}$, we hypothetically assume that the backlog vector after time $t$ is fixed and does not change. This is similar to the method used to define $\phi(\mathbf{X},N_{1})$ except that here we do not assume a fixed value for $N_{1}$, and instead, assume DCP adapts $N_{1}$ as if the backlog vector was changing. In addition, note that by the ergodicity of the channel process $\bar{\mathbf{D}}$ does not depend on $t$, and moreover, by Property~\ref{p2}, $\bar{\mathbf{D}}$ does not depend on $\|\mathbf{X}\|$. To simplify the subsequent analysis, we also consider the following definitions:
\begin{Definition}
For a given $\mathbf{X}$ and a given $\epsilon>0$, the $\epsilon$-neighborhod of $\mathbf{X}$ is defined by
\begin{align*}
\mathcal{N}(\mathbf{X},\epsilon)=\{\mathbf{X}_{1}:
\|\mathbf{X}_{1}-\mathbf{X}\|<\epsilon \}.
\end{align*}
\end{Definition}
\begin{Definition}
For a given $\mathbf{X}$ with $\|\mathbf{X}\|=1$, and a given $\epsilon>0$, the \emph{normalizing} region $\mathcal{R}(\mathbf{X},\epsilon)$ is defined by
\begin{align*}
\mathcal{R}(\mathbf{X},\epsilon)=\{\mathbf{X}_{1}:\|\mathbf{X}_{1}\|\neq 0,
\Big\|\frac{\mathbf{X}_{1}}{\|\mathbf{X}_{1}\|}-\mathbf{X}\Big\|<\epsilon \}\cup \{\mathbf{X}_{1}=\mathbf{0}\}.
\end{align*}
\end{Definition}
\begin{Definition}
Consider a region $\mathcal{R}$ and a vector $\mathbf{X}$ inside $\mathcal{R}$. We define $\xi(\mathbf{X},\mathcal{R})$ as the \emph{supremum} of the angular deviation of the vectors in $\mathcal{R}$ from $\mathbf{X}$, i.e.,
\begin{align*}
\xi(\mathbf{X},\mathcal{R})=\sup_{\mathbf{Y}\in \mathcal{R}} \arccos\Big(\frac{\mathbf{X}\mathbf{Y}}{\|\mathbf{X}\| \ \|\mathbf{Y}\|}\Big).
\end{align*}
\end{Definition}

Now suppose the following conditions hold:
\begin{enumerate}
\item[C1)] $\mathbf{X}_{min}=\gamma_{1}\bar{\mathbf{D}}(\mathbf{X}_{min})=\gamma_{2}\bar{\mathbf{D}}^{*}(\mathbf{X}_{min})$, for some $\gamma_{1},\gamma_{2}>0$.
\item[C2)] For any $N_{1,1}\in \mathcal{N}_{1}$ and $N_{1,2}\in \mathcal{N}_{1}$ with $N_{1,1}\neq N_{1,2}$, we have $\phi(\mathbf{X}_{min},N_{1,1})\neq \phi(\mathbf{X}_{min},N_{1,2})$.
\item[C3)] For any $\beta_{1}>0$ and $\beta_{2}>0$, there exists a sufficiently small $\epsilon>0$ such that if $\mathbf{X}\in \mathcal{R}(\mathbf{X}_{min},\epsilon)$, then
\begin{align}
\bar{\mathbf{D}}(\mathbf{X})-\bar{\mathbf{D}}(\mathbf{X}_{min})=\lambda_{1}\bar{\mathbf{D}}
(\mathbf{X}_{min})+\lambda_{2}
(\frac{\mathbf{X}}{\|\mathbf{X}\|}-\mathbf{X}_{min}),
\nonumber \end{align}
for some $\lambda_{1}$ and $\lambda_{2}$ satisfying $|\lambda_{1}|<\beta_{1}$ and $0<\lambda_{2}<\beta_{2}$.
\item[C4)] For any $\mathbf{X} \in \mathds{N}^{N}$, for some $t$
\begin{align}
P(\mathbf{X}_{t}=\mathbf{X})>0.\nonumber
\end{align}
\end{enumerate}

 Condition C1 may be met by assuming a statistically symmetric channel states as the ones in Section~\ref{sec:casestudy}. Condition C2 simply requires the function $\phi(\mathbf{X}_{min},N_{1})$ to be a one-to-one function of $N_{1}$ at $\mathbf{X}_{min}$. Condition C3 intuitively states that the average departure rates should be a \emph{continuous} function\footnote{As opposed to traditional definitions which usually use $\mathcal{N}(\mathbf{X},\epsilon)$ to define continuity, here, the region $\mathcal{R}(\mathbf{X},\epsilon)$ is used to characterize continuity. } of $\mathbf{X}$ around $\mathbf{X}_{min}$, and in particular, when $\mathbf{X}$ deviates from $\mathbf{X}_{min}$, these rates should deviate from $\bar{\mathbf{D}}(\mathbf{X}_{min})$ in a similar manner. This is in fact expected as increasing the backlog vector in one dimension should increase the expected departure rate in that dimension, which can be considered as a result of the approximation to the GMWM problem through the use of algorithm $A$. Note that in C3 where appropriate the vector $\mathbf{X}$ is normalized by its norm since $\bar{\mathbf{D}}(\mathbf{X})$ does not depend on $\|\mathbf{X}\|$. Finally, C4 simply requires the process $\{\mathbf{X}_{t}\}$ to be able to reach all vectors in $\mathds{N}^{N}$, although what we need for the proof is a relaxed version of this assumption. Using the numerical results for $\phi(\mathbf{X}_{min},N_{1})$ and $\bar{\mathbf{D}}^{*}(\mathbf{X}_{min})$, and the symmetry of channel states, it is easy to verify that the conditions C2-C4 also hold for the examples in Section~\ref{sec:casestudy}. Therefore, there are examples for which the conditions C1-C4 hold. Next, we show that these conditions are sufficient for network instability, completing the proof of part~(c).

 First, note that $\bar{\mathbf{D}}^{*}(\mathbf{X}_{min})\in \Gamma$, which directly follows from the definition
  of $\bar{\mathbf{D}}^{*}(\mathbf{X}_{min})$ and $\Gamma$. Second, the rate $\bar{\mathbf{D}}^{*}(\mathbf{X}_{min})$ belongs to the boundary of $\Gamma$,
otherwise we could find another vector $\mathbf{D}$ inside $\Gamma$ and within a small neighborhood of $\bar{\mathbf{D}}^{*}(\mathbf{X}_{min})$ with larger backlog-rate product, in contradiction with the definition of $\bar{\mathbf{D}}^{*}(\mathbf{X}_{min})$. Hence, we see that the rate $\theta_{\infty}\bar{\mathbf{D}}^{*}(\mathbf{X}_{min})$ belongs to the boundary of $\theta_{\infty}\Gamma$. Third, we can see that by the definition of $\theta_{\infty}$ and $\mathbf{X}_{min}$
\begin{align*}
\theta_{\infty}=\frac{\tilde{\phi}(\mathbf{X}_{min})}{\chi(\mathbf{X}_{min})}\geq \frac{\mathbf{X}_{min}
\bar{\mathbf{D}}(\mathbf{X}_{min})}{\chi(\mathbf{X}_{min})}.
\end{align*}
This is because DCP may use sub-optimal values for $N_{1}$, which by C2 make $\mathbf{X}_{min}
\bar{\mathbf{D}}$ less than $\tilde{\phi}(\mathbf{X}_{min})$ when $N_{c}$ is large.
Using C1 and the above inequality, we have
\begin{align*}
\theta_{\infty}\geq \frac{\|\bar{\mathbf{D}}(\mathbf{X}_{min})\|}{\|\bar{\mathbf{D}}^{*}(\mathbf{X}_{min})\|},
\end{align*}
which implies that
\begin{align*}
\bar{\mathbf{D}}(\mathbf{X}_{min})\leq \ \theta_{\infty} \ \bar{\mathbf{D}}^{*}(\mathbf{X}),
\end{align*}
where the inequality is component-wise. Without of loss of generality, we assume that
\begin{align}
\bar{\mathbf{D}}(\mathbf{X}_{min})=\ \theta_{\infty} \ \bar{\mathbf{D}}^{*}(\mathbf{X}).\label{dinf}
\end{align}

\begin{figure}[t]
\centering
\includegraphics[width=2.4in]{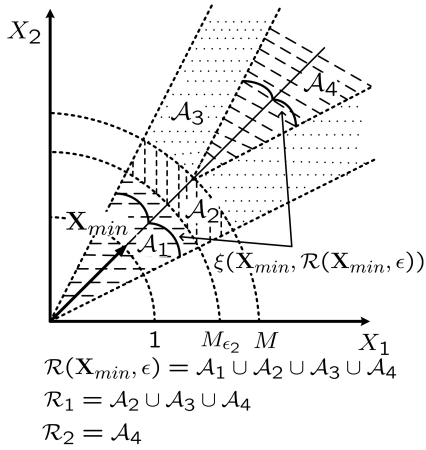}
\caption{Illustration of regions $\mathcal{R}(\mathbf{X}_{min},\epsilon)$, $\mathcal{R}_{1}$, and $\mathcal{R}_{2}$.}\label{fig:m_r}
 \vspace{-.0in}
\end{figure}

Let the input rate be
\begin{align}
\mathbf{a}=(\theta_{\infty}+\varsigma)\bar{\mathbf{D}}^{*}(\mathbf{X})\label{inrate},
\end{align}
for some $\varsigma>0$, which is clearly outside of the region $\theta_{\infty}\Gamma$.
 Let $\Delta_{t,n}$ be the drift vector defined by
\begin{eqnarray}
\Delta_{t,n}=\frac{1}{n}\sum_{i=0}^{n-1}\mathbf{A}_{t+i}-\frac{1}{n}
\sum_{i=0}^{n-1}(\mathbf{D}_{t+i}-\mathbf{U}_{t+i}).\nonumber
\end{eqnarray}
In addition, let
\begin{align*}
 \bar{\Delta}_{\mathbf{X}}=\mathbf{a}-\bar{\mathbf{D}}(\mathbf{X}).
 \end{align*}
Note that $\bar{\Delta}_{\mathbf{X}}$ does not depend on $\|\mathbf{X}\|$ since $\bar{\mathbf{D}}(\mathbf{X})$ has the same property as pointed out earlier.
 Suppose for a given $\epsilon_{1}$, the values for $\beta_{1}$ and $\beta_{2}$ are chosen such that by C3 if $\mathbf{X}_{t}\in \mathcal{R}(\mathbf{X}_{min},\epsilon)$, for appropriately small $\epsilon$, then the following holds
 \begin{align*}
\| \bar{\Delta}_{\mathbf{X}_{t}} -\bar{\Delta}_{\mathbf{X}_{min}}\|<\epsilon_{1}.
 \end{align*}
Using Assumption~\ref{p1}, condition C2, channel ergodicity as stated in Section~\ref{sec:channelp}, and that arrivals are i.i.d, it is not hard to see that\footnote{A similar discussion similar to the one for Property~\ref{p3} applies here.} when $\epsilon$ is sufficiently small, for any positive $\epsilon_{2}$ and $0<\zeta\leq 1$, we can first chose $n$ large and then $M_{\epsilon_{2}}$ sufficiently large, and define the region $\mathcal{R}_{1}$ as
\begin{align}
\mathcal{R}_{1}&=\{\mathbf{X}: \|\mathbf{X}\|\geq M_{\epsilon_{2}}, \mathbf{X}\in \mathcal{R}(\mathbf{X}_{min},\epsilon) \}\nonumber
\end{align}
such that
\begin{align}
&P\big(\|\Delta_{t,n}-\bar{\Delta}_{\mathbf{X}_{t}}\|<\epsilon_{2} | \mathcal{H}_{t},\mathbf{X}_{t}\in \mathcal{R}_{1}\big)
   >(1-\zeta),\label{pdelta2}
\end{align}
where in above
\begin{align}
\bar{\Delta}_{\mathbf{X}_{t}}=\frac{(\varsigma-\lambda_{1}\theta_{\infty})}
{\theta_{\infty}+\varsigma}\mathbf{a}-\lambda_{2}\big(\frac{\mathbf{X}_{t}}{\|\mathbf{X}_{t}\|}-\mathbf{X}_{min}\big).
\label{deltabar}
\end{align}
The above equality is obtained by using condition C3, equality (\ref{dinf}), and considering that the input rate is given by (\ref{inrate}). In particular, we have that
\begin{align*}
\bar{\Delta}_{\mathbf{X}_{min}}=\frac{\varsigma}
{\theta_{\infty}+\varsigma}\mathbf{a}.
\end{align*}
Since $\epsilon_{2}$ can be made arbitrarily small by choosing sufficiently large $n$ and $M_{\epsilon_{2}}$, we assume that for all $\mathbf{X}\in \mathcal{R}_{1}$
\begin{align}
\xi(\bar{\Delta}_{\mathbf{X}},\mathcal{N}(\bar{\Delta}_{\mathbf{X}},\epsilon_{2}))
<\frac{\xi(\mathbf{X}_{min},\mathcal{R}(\mathbf{X}_{min},\epsilon))}{2}.\label{angle2}
\end{align}
Hence, according to (\ref{pdelta2}) and (\ref{angle2}), for $\mathbf{X}_{t}\in \mathcal{R}_{1}$, with probability larger than $(1-\zeta)$ the drift $\Delta_{t,n}$ is close to $\bar{\Delta}_{\mathbf{X}_{t}}$ with a supremum angular deviation that is half of the supremum angular deviation of $\mathbf{X}$'s in $\mathcal{R}(\mathbf{X}_{min},\epsilon)$ from $\mathbf{X}_{min}$.


To continue, let the region $\mathcal{R}_{2}$ be defined as
\begin{align}
\mathcal{R}_{2}&=\{\mathbf{X}: \mathbf{X}-M\mathbf{X}_{min} \in \mathcal{R}(\mathbf{X}_{min},\epsilon)\},
\end{align}
for some $M\geq M_{\epsilon_{2}}$.
This region is a shifted version of $\mathcal{R}(\mathbf{X}_{min},\epsilon)$ with the origin shifted to $M\mathbf{X}_{min}$, and therefore, $\mathcal{R}_{2} \subset \mathcal{R}_{1}$. Fig.~\ref{fig:m_r} provides a graphical demonstration of regions $\mathcal{R}(\mathbf{X}_{min},\epsilon)$, $\mathcal{R}_{1}$, and $\mathcal{R}_{2}$. In the figure, the vector $\mathbf{X}_{min}$ is shown by a unit arrow-vector.
Now we are in a position to show that starting at $\mathbf{X}_{t}=M\mathbf{X}_{min}$, for some appropriately chosen $M$, with positive probability $\{\mathbf{X}_{t+i}, \ i\geq 0\}$ stays in $\mathcal{R}_{2}$ with ever growing norm.

\begin{figure}[t]
\centering
\includegraphics[width=3.2in]{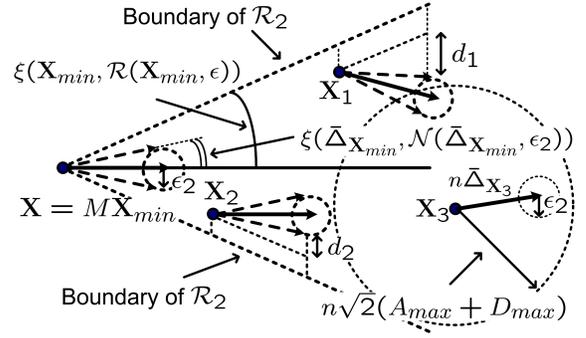}
\caption{Examples where $\mathbf{X}_{t+ni} \in \mathcal{R}_{2}$ explaining cases where $\mathcal{A}_{t+ni,n}=0$ as in the points $\mathbf{X}$, $\mathbf{X}_{1}$, and $\mathbf{X}_{2}$, and the cases where $\mathcal{A}_{t+ni,n}=1$ as in the point $\mathbf{X}_{3}$. In this figure, the region $\mathcal{R}_{2}$ is rotated clockwise.}\label{fig:detals}
 \vspace{-.2in}
\end{figure}

Consider the sequence $\{\mathbf{X}_{t+ni}\}_{i=0}^{\infty}$ with $\mathbf{X}_{t}=M\mathbf{X}_{min}$. Recall that $n$ is chosen sufficiently large according to the value of $\epsilon_{2}$.
Let $\mathcal{A}_{t+ni,n}$ be a r.v. defined by
\begin{align*}
\mathcal{A}_{t+ni,n}=\left \{
\begin{array}{ll}
 1 & \text{if} \ \|\Delta_{t+ni,n}
-\bar{\Delta}_{\mathbf{X}_{t+ni}}\|\leq \epsilon_{2}, \\
0  & \text{otherwise}
 \end{array} \right.
\end{align*}
Provided that $\mathbf{X}_{t+ni}\in\mathcal{R}_{2} $, where  $\mathcal{R}_{2} \subset \mathcal{R}_{1}$, and
 assuming a small $\epsilon_{1}$ and a sufficiently large $M$, it is not hard to see that
  if $\mathcal{A}_{t+ni,n}=1$, then the following hold as a result
   of (\ref{deltabar}) and (\ref{angle2}). First, $\mathbf{X}_{t+n(i+1)}\in \mathcal{R}_{2}$. Second, the distance of the vector $\mathbf{X}_{t+n(i+1)}$ from the boundary of $\mathcal{R}_{2}$ becomes the distance of $\mathbf{X}_{t+ni}$ plus at least $n\delta_{\mathcal{A}}$. Third,
\begin{align*}
 \|\mathbf{X}_{t+n(i+1)}\|\geq \|\mathbf{X}_{t+ni}\|+n\delta_{\mathcal{A}},
 \end{align*}
 where $\delta_{\mathcal{A}}$ is an appropriately small positive constant.
  Fig.~\ref{fig:detals} shows the region
    $\mathcal{R}_{2}$ rotated clockwise, and provides examples for the case where $\mathcal{A}_{t+ni,n}=1$. Specifically, when $\mathbf{X}_{t+ni}$ equals one of the points $\mathbf{X}$, $\mathbf{X}_{1}$,
    and $\mathbf{X}_{2}$, the figure assumes the
    drift vector $\Delta_{t+ni,n}$ is within the $\epsilon_{2}$-neighborhood of
    $\bar{\Delta}_{\mathbf{X}_{t+ni}}$. For points $\mathbf{X}_{1}$ and
    $\mathbf{X}_{2}$, the figure also shows the increases in their distance from the boundary of $\mathcal{R}_{2}$, and denotes them by $d_{1}$ and $d_{2}$, respectively. These values, as mentioned above,
are lower-bounded by $n\delta_{\mathcal{A}}$ as a result of (\ref{deltabar}) and (\ref{angle2}). To see this note that, as shown in the
figure and suggested by (\ref{deltabar}), when $\mathbf{X}_{t+ni}$ deviation from $\mathbf{X}_{min}$, i.e.,
when it deviates from the central line in the figure, the vector
$\bar{\Delta}_{\mathbf{X}_{t+ni}}$ gets a component towards the central line.
This and the assumption that the angular deviations in the $\epsilon_{2}$-neighborhoods are less than half of the one defining region $\mathcal{R}_{2}$, as assumed in (\ref{angle2}), ensure that after $n$ steps the backlog vector remains in $\mathcal{R}_{2}$, and that the distance from the boundary of $\mathcal{R}_{2}$ increases when $\mathcal{A}_{t+ni,n}=1$.
 Using a similar argument, it is easy to see that when $\epsilon_{1}$ is small, an event of the type $\mathcal{A}_{t+ni,n}=1$ increases the norm of backlog vector more than $n\delta_{\mathcal{A}}$.
On the other hand, if $\mathcal{A}_{t+ni,n}=0$ with at most probability $\zeta$, both
  the distance of $\mathbf{X}_{t+n(i+1)}$ from $\mathcal{R}_{2}$ and $\|\mathbf{X}_{t+n(i+1)}\|$, compared to the distance of $\mathbf{X}_{t+ni}$ and $\|\mathbf{X}_{t+ni}\|$, respectively,
   decrease at most by $n\sqrt{2}(A_{max}+D_{max})$. In Fig.~\ref{fig:detals}, the point $\mathbf{X}_{3}$ is an example of this case, where the vector $\Delta_{t+ni,n}$ can be anywhere inside the outer circle, centered at $\mathbf{X}_{3}$, but outside the inner circle defining the $\epsilon_{2}$-neighborhood of the vector $\mathbf{X}_{3}+n\bar{\Delta}_{\mathbf{X}_{3}}$.

  In the rest of the proof, as the worst case, we assume that for $\mathbf{X}_{t+ni}\in\mathcal{R}_{2}$, $i\geq 0$, the event $\{\mathcal{A}_{t+ni,n}=1\}$ occurs with probability $(1-\zeta)$. Note that $\mathcal{R}_{2} \subset \mathcal{R}_{1}$, and when $\mathbf{X}_{t+ni}\in\mathcal{R}_{2} $, the inequality (\ref{pdelta2}) holds regardless of the past history $\mathcal{H}_{t+ni}$. Let the event that $\mathcal{A}_{t+ni,n}=1$ be a success. Based on the previous assumption, for $\mathbf{X}_{t+ni}\in\mathcal{R}_{2} $, this success event occurs with probability $(1-\zeta)$ regardless of the past.
Now consider the sequence $\{\mathbf{X}_{t+ni}\}$, $0\leq i\leq m-1$, and let $m_{(1-\zeta)}$ be the number successes of the type $\{\mathcal{A}_{t+ni,n}=1\}$ out of the $m$ associated trials. The above observations imply that if $\mathbf{X}_{t+ni} \in \mathcal{R}_{2}$, for $0 \leq i\leq m-1$, and if
\begin{align}
(m-m_{(1-\zeta)})\sqrt{2}(A_{max}+D_{max}) < m_{(1-\zeta)}\delta_{\mathcal{A}},\nonumber
\end{align}
then $\mathbf{X}_{t+nm}\in \mathcal{R}_{2}$, and
\begin{align}
\|\mathbf{X}_{t+nm}\|&\geq \|\mathbf{X}_{t}\|+m_{(1-\zeta)}n\delta_{\mathcal{A}}
\nonumber \\&-(m-m_{(1-\zeta)})n\sqrt{2}(A_{max}+D_{max}).\nonumber
\end{align}
Using the above, we see that a sufficient condition for the sequence $\{\mathbf{X}_{t+nm},  m\geq 1\}$ to stay within $\mathcal{R}_{2}$, and
\begin{align}
\|\mathbf{X}_{t+nm}\|&\geq \|\mathbf{X}_{t}\|+n\ m \ \epsilon_{3}(\delta_{\mathcal{A}}+\sqrt{2}(A_{max}+D_{max})),\label{norm:e}
\end{align}
for some $\epsilon_{3}$ with
\begin{align*}
\epsilon_{3}<\frac{\delta_{\mathcal{A}}}{\sqrt{2}(A_{max}+D_{max})+\delta_{\mathcal{A}}}
\end{align*}
is that for all $m\geq 1$,
\begin{align}
r_{m}&\triangleq 1-\frac{m_{(1-\zeta)}}{m}\nonumber \\ &<\frac{\delta_{\mathcal{A}}}{\sqrt{2}(A_{max}+D_{max})+\delta_{\mathcal{A}}}-\epsilon_{3}.\label{rm}
\end{align}
In what follows, we show that with positive probability the above inequality holds for all $m\geq 1$.

 Starting at $M\mathbf{X}_{min}$, let $\tau_{\mathcal{R}_{2}}$ be the first time that the ratio $r_{m}$ does not satisfy (\ref{rm}). Consider the sequence
 \begin{align}
 \{\mathcal{A}_{t+nm,n},\ 0\leq m\leq \tau_{\mathcal{R}_{2}}-2\}
 . \label{sequence}
 \end{align}
For $ 0\leq m\leq \tau_{\mathcal{R}_{2}}-1$, the discussion leading to (\ref{norm:e}) and (\ref{rm}) implies that  $\mathbf{X}_{t+nm}\in\mathcal{R}_{2}$. Furthermore, this discussion shows that the sequence can be considered as a truncated Bernoulli process with success probability $(1-\zeta)$. An intuitive yet important observation is that for an infinite sequence of Bernoulli trials $\{B_{i}, i\geq 0 \}$ with success probability $(1-\zeta)$, for any given $\epsilon_{4}>0$, with positive probability the ratio of failures \emph{never} reaches $\zeta+\epsilon_{4}$. This is the key to prove $\tau_{\mathcal{R}_{2}}=\infty$, or equivalently, (\ref{rm}) holds for all $m\geq 1$, with positive probability. Let the notation $r_{m}$ be re-used as the failure ratio for the infinite Bernoulli process, i.e.,
 \begin{align}
r_{m}=1-\frac{1}{m}\sum_{i=1}^{m}B_{i}.\nonumber
 \end{align}
 Using large deviation results \cite{jeff:03}, we have
\begin{align}
P(r_{m}-\zeta>\epsilon_{4})<\rho^{m},\label{larged}
\end{align}
where
\begin{align}
\rho=\inf_{s>0}M_{Z_{\zeta}}(s)\ <1,\nonumber
\end{align}
where $Z_{\zeta}=1-B_{1}-\zeta -\epsilon_{4}$, and $M_{Z_{\zeta}}(s)$ is the characteristic function of $Z_{\zeta}$. The above inequality indicates that with probability at least $(1-\rho^{m})$, the ratio of failures after $m$ trials is less than or equal to $\zeta+\epsilon_{4}$.

 To further study $r_{m}$, we consider the infinite Bernoulli process in a sequence of stages. In the first and second stages, we consider $m$ Bernoulli trials. However, after the second stage, for the $i_{th}$ stage, we consider the next subsequent $2^{i-2}m$ trials. Since trials are independent, with probability $(1-\zeta)^{m}$, we can have only successes for the first $m$ trials, and thus, the ratio $r_{m}$ never goes beyond zero, i.e.,
 \begin{align}
r_{j}=0, \ 1\leq j\leq m.\nonumber
 \end{align}
  For the second stage with the next $m$ trials, using (\ref{larged}), we see that with probability at least $(1-\zeta)^{m}(1-\rho^{m})$
\begin{align}
\max_{0\leq j \leq 2m}r_{j}\leq \frac{0+m(\zeta+\epsilon_{4})}{m+m(\zeta+\epsilon_{4})}<2(\zeta+\epsilon_{4}),\nonumber
\end{align}
where the first inequality refers to the worst case where in the second stage of $m$ trials, the failures happen in the beginning of the stage, i.e., when the $(m+1)_{\text{th}}$, $(m+2)_{\text{th}}$,..., and $(m+m(\zeta+\epsilon_{3}))_{\text{th}}$ trials are all failures. Inductively, considering the $(l+2)_{\text{th}}$ stage, we see that with probability at least $(1-\zeta)^{m}\Pi_{p=0}^{l}(1-\rho^{2^{p}m})$
\begin{align}
\max_{1\leq j \leq 2^{(l+1)m}}r_{j}\leq \frac{(2^{l+1}-1)m(\zeta+\epsilon_{4})}{2^{l}m+2^{l}m(\zeta+\epsilon_{4})}<2(\zeta+\epsilon_{4}),
\end{align}
where the numerator is the total number of failures up to the end of $(l+2)_{\text{th}}$ stage, and the denominator corresponds to the worst case where the failures in the $(l+2)_{\text{th}}$ stage all occur in the beginning of the stage. Therefore, with probability at least
\begin{align}
p_{\zeta}=(1-\zeta)^{m}\Pi_{p=0}^{\infty}(1-\rho^{2^{p}m})\nonumber
\end{align}
the ratio $r_{m}$, $m\geq 1$, always stays below $2(\zeta+\epsilon_{4})$. But
\begin{align}
p_{\zeta}>(1-\zeta)^{m}\Pi_{p=0}^{\infty}(1-(\rho^{m})^{p+1}).\nonumber
\end{align}
This and Lemma~\ref{lemma:l5} indicate that
\begin{align*}
p_{\zeta}>0.
\end{align*}

The above discussion implies that with a positive probability, not less than $p_{\zeta}$, the ratio $r_{m}$ associated with the sequence in (\ref{sequence}) stays below $2(\zeta+\epsilon_{4})$. Hence,
if $\zeta$ and $\epsilon_{4}$ are chosen such that
\begin{align}
\zeta+\epsilon_{4}<\frac{1}{2}\Big(\frac{\delta_{\mathcal{A}}}{\sqrt{2}(A_{max}+D_{max})+\delta_{\mathcal{A}}}-\epsilon_{3}\Big),\label{ineq:last}
\end{align}
then starting at $\mathbf{X}_{t}=M\mathbf{X}_{min}$, the inequality in (\ref{rm}) holds for all $m\geq 1$ with positive probability. Since this latter statement can be generalized to the case where $\mathbf{X}_{t}\in \mathcal{R}_{2}$, we have that
\begin{align}
P(\forall m\geq 0, \ \mathbf{X}_{t+nm}\in \mathcal{R}_{2} \ \text{and (\ref{norm:e}) holds} \ |\mathbf{X}_{t}\in \mathcal{R}_{2} )>0\label{pxr2}
\end{align}
if (\ref{ineq:last}) holds. But (\ref{ineq:last}) can be satisfied since the choice for a positive $\epsilon_{4}$ is arbitrary, and as mentioned in the discussion leading to (\ref{pdelta2}), $\zeta$ can be chosen arbitrarily small. Hence, for an appropriate choice of parameters, (\ref{pxr2}) holds, which suggests that with positive probability $\mathbf{X}_{t+nm}$ stays in $\mathcal{R}_{2}$, and its norm increases (at least) linearly with $m$. Since by C4 with positive probability $\mathbf{X}_{t}\in \mathcal{R}_{2}$ for some $t$, and arrivals and departures are bounded implying for $0\leq j\leq n$, $\|\mathbf{X}_{t+mn-j}\|\geq \|\mathbf{X}_{t+nm}\|-C$, for some $C>0$, we see that (\ref{pxr2}) indicates that with positive probability
 \begin{align*}
\lim_{i \to \infty } \|\mathbf{X}_{t+i}\|=\infty.
  \end{align*}
 This shows that when the input rate is given by (\ref{inrate}), and thus, when it is outside the region $\theta_{\infty}\Gamma$, for any bounded region $\mathcal{C}$, with positive probability the process $\mathbf{X}_{t+i}$ never returns to $\mathcal{C}$, and hence, $\mathds{E}[\tau_{\mathcal{C}}]=\infty$, implying network instability. This completes the proof of part~(c) of the theorem.

 \end{proof}

\appendix[Lemmas]

\begin{Lemma}\label{prelemma}
For any $\epsilon>0$, regardless of the past history $\mathcal{H}_{t}$, there exists a sufficiently large $K_{\epsilon}$
such that for all $\mathbf{s} \in \mathcal{S}$ and $K>K_{\epsilon}$
\begin{align*}
\Big|\mathds{E}\big[(\tau_{K}+1 ) \pi(\mathbf{s})-\sum_{k=0}^{\tau_{K}}\mathbf{1}_{\mathbf{s}(t+k)=\mathbf{s}}\big|\mathcal{H}_{t}\big]\Big|
<\epsilon\mathds{E}\Big[(\tau_{K}+1 )|\mathcal{H}_{t}\Big].
\end{align*}

\begin{proof}
Since $\tau_{i+1}-\tau_{i}>2N_{c}$, it is easy to verify that
$\tau_{K} \to \infty, \ a.s.$, as $K \to \infty$. This almost surely convergence and the ergodicity of channel process, as stated in Section~\ref{sec:channelp}, imply that as $K \to \infty$
\begin{align}
\frac{1}{\tau_{K}+1}\sum_{k=0}^{\tau_{K}}\mathbf{1}_{\mathbf{s}(t+k)=\mathbf{s}} \to \pi(\mathbf{s}), \ a.s.
\end{align}
 Moreover, since the channel convergence in Section~\ref{sec:channelp} is uniform in the past history and $t$, and since the number of channel states is finite, we see that the above convergence is uniform in $t$, $\mathcal{H}_{t}$, and $\mathbf{s}$. Thus, for any $\epsilon^{'}>0$ and $\zeta>0$, we can find a sufficiently large $K_{\epsilon^{'},\zeta}$ independent of the past history $\mathcal{H}_{t}$ and $\mathbf{s}$ such that \cite{borovkov}
\begin{align}
P(\sup_{K>K_{\epsilon^{'},\delta}}|\pi(s)
-\frac{1}{\tau_{K}+1}\sum_{k=0}^{\tau_{K}}
\mathbf{1}_{\mathbf{s}(t+k)=\mathbf{s}}|>\epsilon^{'}|\mathcal{H}_{t})<\zeta.\label{prelemma:e1}
\end{align}
Given $\mathcal{H}_{t}$, let $\mathcal{A}_{K_{\epsilon^{'},\zeta},\epsilon^{'}}$ denote the set of all $\omega \in \Omega$ with the property that
 \begin{align}
\sup_{K>K_{\epsilon^{'},\delta}}|\pi(s)-\frac{1}
{\tau_{K}+1}\sum_{k=0}^{\tau_{K}}\mathbf{1}_{\mathbf{s}(t+k)=\mathbf{s}}|> \epsilon^{'}.\nonumber
 \end{align}
By (\ref{prelemma:e1}), we have that
\begin{align}
P(\mathcal{A}_{K_{\epsilon^{'},\zeta},\epsilon^{'}}|\mathcal{H}_{t})<\zeta. \label{p:delta}
\end{align}

Suppose $K>K_{\epsilon^{'},\zeta}$ and let
\begin{align}
\Delta=\mathds{E}\Big[(\tau_{K}+1 ) \pi(\mathbf{s})-\sum_{k=0}^{\tau_{K}}\mathbf{1}_{\mathbf{s}(t+k)=\mathbf{s}}\ \big|\mathcal{H}_{t}\Big].\nonumber
\end{align}
 Using conditional expectations and the definition of $\mathcal{A}_{K_{\epsilon^{'},\zeta},\epsilon^{'}} $, and considering the fact that
$0\leq \pi(\mathbf{s}) \leq 1$ and $\tau_{K}\geq 0$, we can show that
\begin{align}
\Delta
&\leq P\big(\omega \notin  \mathcal{A}_{K_{\epsilon^{'},\zeta},\epsilon^{'}}\big)\mathds{E}\big[\epsilon^{'}(\tau_{K}+1)\ |\omega \notin \mathcal{A}_{K_{\epsilon^{'},\zeta},\epsilon^{'}},\mathcal{H}_{t}\big]
\nonumber \\ &
 +P\big(\omega \in  \mathcal{A}_{K_{\epsilon^{'},\zeta},\epsilon^{'}}\big)\mathds{E}\big[(\tau_{K}+1)\ |\omega \in \mathcal{A}_{K_{\epsilon^{'},\zeta},\epsilon^{'}},\mathcal{H}_{t}\big]\label{pre1:p1}
\end{align}
Similarly, we obtain
\begin{align}
&\mathds{E}\big[(\tau_{K}+1)|\mathcal{H}_{t}\big]
\nonumber \\&= P\big(\omega \notin  \mathcal{A}_{K_{\epsilon^{'},\zeta},\epsilon^{'}}\big)\mathds{E}\big[(\tau_{K}+1)\ |\omega \notin \mathcal{A}_{K_{\epsilon^{'},\zeta},\epsilon^{'}},\mathcal{H}_{t}\big]
\nonumber \\ &
 +P\big(\omega \in  \mathcal{A}_{K_{\epsilon^{'},\zeta},\epsilon^{'}}\big)\mathds{E}\big[(\tau_{K}+1) \ |\omega \in \mathcal{A}_{K_{\epsilon^{'},\zeta},\epsilon^{'}},\mathcal{H}_{t}\big]. \nonumber
\end{align}
Since $\tau_{K}\geq 0$, the above implies that
\begin{align}
&P\big(\omega \notin  \mathcal{A}_{K_{\epsilon^{'},\zeta},\epsilon^{'}}\big)\mathds{E}\big[(\tau_{K}+1)\ |\omega \notin \mathcal{A}_{K_{\epsilon^{'},\zeta},\epsilon^{'}},\mathcal{H}_{t}\big]
\nonumber \\ &\leq
\mathds{E}\big[(\tau_{K}+1)|\mathcal{H}_{t}\big]\label{pre1:p6}
\end{align}
In addition, w.p.1, $\tau_{K}+1\leq (K+1)(1+L_{1})N_{c}$. It thus follows from (\ref{p:delta}), (\ref{pre1:p1}), and (\ref{pre1:p6}) that
\begin{align}
\Delta &\leq
\epsilon^{'}\mathds{E}\big[(\tau_{K}+1)|\mathcal{H}_{t}\big]+\zeta (K+1)(1+L_{1})N_{c}
\nonumber
\end{align}
Noting the fact that $\tau_{K}\geq 2KN_{c}$, we obtain
\begin{align}
\Delta &\leq \mathds{E}\big[(\tau_{K}+1)|\mathcal{H}_{t}\big]\Big(\epsilon^{'}+\zeta \frac{(K+1)(1+L_{1})N_{c}}{\mathds{E}\big[(\tau_{K}+1)|\mathcal{H}_{t}\big]}\Big)
\nonumber \\
&< \mathds{E}\big[(\tau_{K}+1)|\mathcal{H}_{t}\big]\Big(\epsilon^{'}+\zeta \frac{(K+1)(1+L_{1})N_{c}}{2KN_{c}+1}\Big)
 \nonumber \\
&=\epsilon \mathds{E}\big[(\tau_{K}+1)|\mathcal{H}_{t}\big], \nonumber
\end{align}
where
\begin{align*}
\epsilon=\epsilon^{'}+\zeta \frac{(K+1)(1+L_{1})N_{c}}{2KN_{c}+1}
\end{align*}
can be made arbitrarily small by choosing sufficiently small values for $\epsilon^{'}$ and $\zeta$. A similar discussion holds for $-\Delta$ with the same $\epsilon$, completing the proof.
\end{proof}
\end{Lemma}

\begin{Lemma}\label{lemma:D*}
For any given $\epsilon>0$, there exists a sufficiently large constant $K_{\epsilon}>0$ such that for
all $K>K_{\epsilon}$, we can find a proper $M_{\epsilon,K}$ such that if $\|\mathbf{X}_{t}\|>M_{\epsilon,K}$, the
following holds
\begin{align*}
&\mathds{E}[\sum_{k=0}^{\tau_{K}}\mathbf{X}_{t}\mathbf{a}- \mathbf{X}_{t+k}\mathbf{D}^{*}_{t+k}|\mathcal{H}_{t}]
\nonumber\\& \leq \mathds{E}\Big[(\tau_{K}+1)\|\mathbf{X}_{t}\|\Big(\epsilon-\big(1-\max_{\mathbf{s}\in \mathcal{S}}\sum_{\mathbf{I} \in
\mathcal{I}}\beta_{\mathbf{s},\mathbf{I}}\big)\chi(\mathbf{X}_{t})\Big)
\big|\mathcal{H}_{t}\Big]
\end{align*}
\end{Lemma}
\begin{proof} To prove the lemma, we first note that by the definition of $\mathbf{D}^{*}(\mathbf{X}_{t},\mathbf{s})$, and the assumption that departures are bounded by $D_{max}$, we have
\begin{align}
&\!\!\!\mathbf{X}_{t+k}\mathbf{D}^{*}_{t+k}= \max_{\mathbf{I}\in
\mathcal{I}} \mathbf{X}_{t+k}\mathbf{D}(\mathbf{s}_{t+k},\mathbf{I})
 \nonumber \\& \geq \max_{\mathbf{I}\in
\mathcal{I}} \mathbf{X}_{t} \mathbf{D}(\mathbf{s}_{t+k},\mathbf{I})
-\max_{\mathbf{I}\in
\mathcal{I}}\Big(\sum_{i=0}^{k-1}\mathbf{D}_{t+i}\mathbf{D}(\mathbf{s}_{t+k},\mathbf{I})\Big)
\nonumber \\
&\geq  \max_{I\in \mathcal{I}} \mathbf{X}_{t}
\mathbf{D}(\mathbf{s}_{t+k},\mathbf{I})-kND_{max}^{2}.
\label{lemma1:p:1}
\end{align}
Using (\ref{a:vector}), we also observe that
\begin{align}
&\!\!\!\mathbf{X}_{t}\mathbf{a}-\mathds{E}_{\mathbf{s}}[\mathbf{X}_{t}\mathbf{D}^{*}(\mathbf{X}_{t},\mathbf{s})]
\nonumber
\\&=
\mathbf{X}_{t}\sum_{\mathbf{s}\in\mathcal{S}}\pi(\mathbf{s})
\sum_{\mathbf{I}\in\mathcal{I}}\beta_{\mathbf{s},\mathbf{I}}\mathbf{D}
(\mathbf{s},\mathbf{I})-\sum_{\mathbf{s}\in\mathcal{S}}\pi(\mathbf{s})
\mathbf{X}_{t}\mathbf{D}^{*}(\mathbf{X}_{t},\mathbf{s}) \nonumber
\\&=
\sum_{\mathbf{s}\in\mathcal{S}}\pi(\mathbf{s})\bigg(\sum_{\mathbf{I}\in\mathcal{I}}
\beta_{\mathbf{s},\mathbf{I}}\Big(\mathbf{X}_{t} \mathbf{D}(\mathbf{s},\mathbf{I})
- \mathbf{X}_{t}\mathbf{D}^{*}(\mathbf{X}_{t},\mathbf{s})
\Big)
\nonumber\\&\qquad \qquad \quad  \qquad-\Big(\big(1-\sum_{\mathbf{I}\in\mathcal{I}}\beta_{\mathbf{s},\mathbf{I}}\big)
\mathbf{X}_{t}\mathbf{D}^{*}(\mathbf{X}_{t},\mathbf{s})\Big)
\bigg) .
\end{align}
Since by definition for all $\mathbf{I}\in\mathcal{I}$
\begin{align*}
\mathbf{X}_{t}\mathbf{D}^{*}(\mathbf{X}_{t},\mathbf{s})\geq \mathbf{X}_{t}\mathbf{D}(\mathbf{s},\mathbf{I}),
\end{align*}
we have
\begin{align}
&\mathbf{X}_{t}\mathbf{a}-\mathds{E}_{\mathbf{s}}[\mathbf{X}_{t}\mathbf{D}^{*}(\mathbf{X}_{t},\mathbf{s})]
 \nonumber \\&\qquad \leq
-\sum_{\mathbf{s}\in\mathcal{S}}\pi(\mathbf{s})\big(1-\sum_{\mathbf{I}\in\mathcal{I}}
\beta_{\mathbf{s},\mathbf{I}}\big)\mathbf{X}_{t}\mathbf{D}^{*}(\mathbf{X}_{t},\mathbf{s})
\nonumber \\ &\qquad\leq
-\big(1-\max_{\mathbf{s}\in \mathcal{S}}\sum_{\mathbf{I}\in\mathcal{I}}
\beta_{\mathbf{s},\mathbf{I}}\big) \sum_{\mathbf{s}\in\mathcal{S}}\pi(\mathbf{s})\mathbf{X}_{t}\mathbf{D}^{*}(\mathbf{X}_{t},\mathbf{s})
\nonumber \\ &\qquad =-\|\mathbf{X}_{t}\|\big(1-\max_{\mathbf{s}\in \mathcal{S}}\sum_{\mathbf{I}\in\mathcal{I}}
\beta_{\mathbf{s},\mathbf{I}}\big) \chi(\mathbf{X}_{t})
,\label{lemma1:p:2}
\end{align}
where the last equality follows from the definition of $\chi(\mathbf{X}_{t})$.

Back to the inequality in the lemma, using \eqref{lemma1:p:1}, we have
\begin{align*}
&\mathds{E}\big[\sum_{k=0}^{\tau_{K}}\mathbf{X}_{t}\mathbf{a}- \mathbf{X}_{t+k}\mathbf{D}^{*}_{t+k}|\mathcal{H}_{t}\big]\leq
\mathds{E}\big[(\tau_{K}+1)^{2}N D_{max}^{2}|\mathcal{H}_{t}\big]\nonumber\\
&\qquad+
\mathds{E}\Big[\sum_{k=0}^{\tau_{K}}\big(\mathbf{X}_{t}\mathbf{a}-\mathbf{X}_{t}\sum_{\mathbf{s}\in\mathcal{S}}
\mathbf{1}_{\mathbf{s}(t+k)=\mathbf{s}}\mathbf{D}^{*}(\mathbf{X}_{t},\mathbf{s})\big) \ \big|\mathcal{H}_{t}\Big]
\nonumber\\ &=\mathds{E}\big[(\tau_{K}+1)^{2}N D_{max}^{2}|\mathcal{H}_{t}\big]\nonumber\\
&\qquad+
\mathds{E}\Big[\sum_{k=0}^{\tau_{K}}\mathbf{X}_{t}\mathbf{a}-\mathbf{X}_{t}\sum_{\mathbf{s}\in\mathcal{S}}\mathbf{D}^{*}(\mathbf{X}_{t},\mathbf{s})
\sum_{k=0}^{\tau_{K}}\mathbf{1}_{\mathbf{s}(t+k)=\mathbf{s}}\ \big|\mathcal{H}_{t}\Big].
\end{align*}
Using Lemma~\ref{prelemma}, for $\epsilon_{1}>0$ and sufficiently large $K_{1}$, we have that for $K>K_{1}$
\begin{align}
&\mathds{E}\big[\sum_{k=0}^{\tau_{K}}\mathbf{X}_{t}\mathbf{a}- \mathbf{X}_{t+k}\mathbf{D}^{*}_{t+k}|\mathcal{H}_{t}\big]\nonumber \\&\leq
\mathds{E}\big[(\tau_{K}+1)^{2}N D_{max}^{2}|\mathcal{H}_{t}\big]+\mathds{E}\Big[(\tau_{K}+1)\mathbf{X}_{t}\mathbf{a}
\nonumber\\
&\qquad  -(\tau_{K}+1)\mathbf{X}_{t}\sum_{\mathbf{s}\in\mathcal{S}}
\mathbf{D}^{*}(\mathbf{X}_{t},\mathbf{s})
(\pi(\mathbf{s})-\epsilon_{1})\ \big|\mathcal{H}_{t}\Big]
\nonumber\\&=\mathds{E}\big[(\tau_{K}+1)^{2}N D_{max}^{2}|\mathcal{H}_{t}\big]\nonumber\\
&\qquad +
\epsilon_{1}\mathds{E}\big[(\tau_{K}+1)\|\mathbf{X}_{t}\||\mathcal{S}|\sqrt{N}D_{max}\big|\mathcal{H}_{t}\big]
\nonumber \\ &
 \qquad +\mathds{E}\Big[(\tau_{K}+1)\Big(\mathbf{X}_{t}\mathbf{a}-
 \mathds{E}_{\mathbf{s}}\big[\mathbf{X}_{t}\mathbf{D}^{*}(\mathbf{X}_{t},\mathbf{s})\big]\Big)\big|\mathcal{H}_{t}\Big].\label{lemma:p2}
\end{align}
Combining (\ref{lemma1:p:2}) and (\ref{lemma:p2}), we obtain the inequality in lemma with
\begin{align*}
\epsilon=\frac{(\tau_{K}+1)N D_{max}^{2}}{\|\mathbf{X}_{t}\|}+\epsilon_{1}|\mathcal{S}|\sqrt{N}D_{max}.
\end{align*}
The choice for a positive $\epsilon$ is arbitrary since one can first select $K_{\epsilon}\geq K_{1}$ such that for all $K>K_{\epsilon}$, $\epsilon_{1}$
is sufficiently small. After selecting $K$, because w.p.1 $\tau_{K}+1\leq (K+1)(1+L_{1})N_{c}$, one can chose $M_{\epsilon,K}$
such that for $\|\mathbf{X}_{t}\|>M_{\epsilon,K}$ the first term in $\epsilon$ is also sufficiently small, completing the proof.

\end{proof}

\begin{Lemma}\label{lemma:2}
For any given $\epsilon>0$, there exists a sufficiently large constant $K
_{\epsilon}>0$ such that for
all $K>K_{\epsilon}$, we can find a proper $M_{\epsilon,K}$ such that if $\|\mathbf{X}_{t}\|>M_{\epsilon,K}$, the
following holds
\begin{align}
&\mathds{E}\big[\sum_{i=0}^{\tau_{K}}\mathbf{X}_{t+i}\mathbf{D}^{*}_{t+i}|
\mathcal{H}_{t}\big]\nonumber \\&
\qquad \qquad \leq\mathds{E}\Big[(\tau_{K}+1)\|\mathbf{X}_{t}\|\big(\chi(\mathbf{X}_{t})+\epsilon\big)
\big|\mathcal{H}_{t}\Big]\nonumber.
\end{align}

\begin{proof}
Using the definition of $\mathbf{D}^{*}(\mathbf{X},\mathbf{s})$, for the LHS of the inequality in the lemma we can show that
\begin{align*}
&\text{LHS}=
\mathds{E}\Big[\sum_{i=0}^{\tau_{K}}\max_{\mathbf{I}\in\mathcal{I}}\Big(\big(\mathbf{X}_{t}
+\sum_{j=0}^{i-1}(\mathbf{A}_{t+j}
 -\mathbf{D}_{t+j}+\mathbf{U}_{t+j})\big)
\nonumber \\ & \qquad \qquad  \qquad \qquad \quad \qquad \qquad \qquad \qquad
 \mathbf{D}(\mathbf{s}_{t+i},\mathbf{I})\Big)\big|\mathcal{H}_{t}\Big].
\end{align*}
Since arrivals and departures are bounded by $A_{max} $ and $D_{max}$, respectively, we have that
\begin{align}
\text{LHS}&\leq
\mathds{E}\big[\sum_{i=0}^{\tau_{K}}\mathbf{X}_{t}
\mathbf{D}^{*}(\mathbf{X}_{t},\mathbf{s}_{t+i})\ |\mathcal{H}_{t}\big]
\nonumber
\\&+\mathds{E}\Big[\sum_{i=0}^{\tau_{K}}iNA_{max}D_{max}
+\sum_{i=0}^{\tau_{K}}i ND_{max}^{2}\
 \big|\mathcal{H}_{t}\Big].
\label{lemma2:lp}
\end{align}
Let $\Sigma$ be the first term of the RHS of the above inequality. We have
\begin{align*}
\Sigma
&=\mathds{E}\Big[ \sum_{i=0}^{\tau_{K}}\sum_{\mathbf{s}\in \mathcal{S}}\mathbf{X}_{t}
\mathbf{D}^{*}(\mathbf{X}_{t},\mathbf{s}) \mathbf{1}_{\mathbf{s}(t+i)=\mathbf{s}}\big| \mathcal{H}_{t} \Big]
\nonumber \\ &=
 \mathbf{X}_{t} \sum_{\mathbf{s}\in \mathcal{S}}\mathbf{D}^{*}(\mathbf{X}_{t},\mathbf{s})
 \mathds{E}\Big[\sum_{i=0}^{\tau_{K}}\mathbf{1}_{\mathbf{s}(t+i)=\mathbf{s}}
|\mathcal{H}_{t}\Big].
\end{align*}
Using Lemma~\ref{prelemma}, for any positive $\epsilon_{1}$, we can find a sufficiently large $K_{1}$ such that for $K>K_{1}$
\begin{align}
\Sigma&\leq \mathbf{X}_{t} \sum_{\mathbf{s}\in \mathcal{S}}\mathbf{D}^{*}(\mathbf{X}_{t},\mathbf{s})
\mathds{E} \Big[(\tau_{K}+1)\big(\pi(\mathbf{s})+\epsilon_{1}\big)\ \big|\mathcal{H}_{t}\Big]
\nonumber \\
&= \mathds{E}\Big[(\tau_{K}+1)  \sum_{\mathbf{s}\in \mathcal{S}}\pi(\mathbf{s})  \mathbf{X}_{t}\mathbf{D}^{*}(\mathbf{X}_{t},\mathbf{s})
\nonumber \\
 &\qquad \qquad \qquad  +\epsilon_{1}|\mathcal{S}|\sqrt{N}\|\mathbf{X}_{t}\|D_{max}(\tau_{K}+1)  \ \big| \mathcal{H}_{t} \Big]
\nonumber \\      &=
\mathds{E}\Big[(\tau_{K}+1)\|\mathbf{X}_{t}\|
\Big(\chi(\mathbf{X}_{t})+\epsilon_{1}\sqrt{N}|\mathcal{S}|D_{max}\Big)
\big|\mathcal{H}_{t}\Big],\label{sigma:pre}
\end{align}
where the last equality follows from the definition of $\chi(\mathbf{X}_{t})$.

 Considering inequalities (\ref{lemma2:lp}) and (\ref{sigma:pre}), we obtain
\begin{align}
\text{LHS}&
\leq
\mathds{E}\Big[(\tau_{K}+1)\|\mathbf{X}_{t}\| \big(
\chi(\mathbf{X}_{t})+\epsilon_{2}\big)\big|\mathcal{H}_{t}\Big],\label{lemma:ineq}
\end{align}
where
\begin{align}
\epsilon_{2}&=\epsilon_{1}|\mathcal{S}|\sqrt{N}D_{max}
 +
(\tau_{K}+1)N\frac{A_{max}D_{max}+D_{max}^{2}}
{\|\mathbf{X}_{t}\|}\nonumber.
\end{align}
To complete the proof, it remains to show that $\epsilon_{2}$ can be made arbitrarily small. Consider any positive $\epsilon$. We first choose $K_{\epsilon}$ such that for $K>K_{\epsilon}$ the value of $\epsilon_{1
}$ is sufficiently small to make the first term in $\epsilon_{2}$ less than $\frac{\epsilon}{2}$. Since $\tau_{k}+1\leq (K+1)(1+L_{1})N_{c}$, we see that for a given $K$ with $K>K_{\epsilon}$ if $\|\mathbf{X}_{t}\|>M_{\epsilon,K}$ for a sufficiently large $M_{\epsilon,K}$, then the second term in $\epsilon_{2}$ can also be less than $\frac{\epsilon}{2}$. Therefore, for any positive $\epsilon$, if $K>K_{\epsilon}$ and $\|\mathbf{X}_{t}\|>M_{\epsilon,K}$, for appropriate values of $K_{\epsilon}$ and $M_{\epsilon,K}$, then the inequity (\ref{lemma:ineq}) holds with $\epsilon_{2}<\epsilon$. But this means the inequality also holds for $\epsilon$, as required.
\end{proof}
\end{Lemma}


\begin{Lemma}\label{lemma:3}
Suppose $6\theta_{\varphi}<\alpha$, and let $\epsilon$ be a positive real number. For any given $\epsilon$, there exists a constant $K_{\epsilon}$ such that if
$K>K_{\epsilon}$, then for $\|\mathbf{X}_{t} \| >M_{\epsilon,K}$ the following holds
\begin{align}
&\mathds{E}\big[\sum_{i=\tau_{0}+1}^{\tau_{K}}\mathbf{X}_{t+i}\mathbf{D}_{t+i}|
\mathcal{H}_{t}\big] \nonumber \\ &\qquad \geq   \mathds{E}\Big[(\tau_{K}+1)\|\mathbf{X}_{t}\|
  \big(R_{\infty}(\tilde{\phi}(t)-\alpha-3\theta_{\varphi})-\epsilon\big)\big |\mathcal{H}_{t}\Big],\nonumber
\end{align}
where $M_{\epsilon,K}$ is a sufficiently large constant depending on $\epsilon$ and $K$, and $R_{\infty}$
is defined in (\ref{r:infty}).

\begin{proof}
The essence of the proof in this lemma is finding a lower-bound for the percentage of time that \emph{near optimal} values for $N_{1}$ are used by DCP. We prove that this percentage is close to $R_{\infty}$. First, we place a requirement on $\|\mathbf{X}_{t}\|$ for a given $K$. Later in the proof,
we find an appropriate lower-bound $K_{\epsilon}$ for $K$ according to the value of $\epsilon$. Note
that w.p.1, for any given $K$, $\tau_{K}\leq (K+1) (1+L_{1})N_{c}$.
Therefore, since departures and arrivals are bounded by $D_{max}$ and
$A_{max}$, respectively, we can easily see that for
 $0\leq i\leq \tau_{K}$, $
 \|\mathbf{X}_{t+i}-\mathbf{X}_{t}\| <C^{'}_{K}$, where $C^{'}_{K}$ is an appropriate constant depending on $K$.
Having this inequality, we can find an appropriate constant
$M^{'}_{K}$, depending on $K$, such that if
\begin{align}
\|\mathbf{X}_{t}\|>M^{'}_{K} \label{condition}
,
\end{align}
 then the following statements hold according to Property~\ref{p2} and Property~\ref{p3}, respectively, with
$\epsilon_{1}<\frac{1}{2}(\frac{\alpha}{6}-\theta_{\varphi})$.

 \begin{Statement} \label{st1}
 For $t \leq t_{1}\leq t+\tau_{K}$, $t \leq t_{2}\leq t+\tau_{K}$, and any
$N_{1}\in \mathcal{N}_{1}$,
\begin{eqnarray}
|\phi(\mathbf{X}_{t_{1}},N_{1})-\phi(\mathbf{X}_{t_{2}},N_{1})|<
\epsilon_{1}.
\label{ineq:assump:1}
\end{eqnarray}
\end{Statement}
\begin{Statement}\label{st2}
For any $\tau_{i}$, with $0\leq i \leq K$, and any
$N_{1}\in \mathcal{N}_{1}$, with probability $(1-\varrho_{\varphi})$, and regardless of $i$ and the past history at time $t+\tau_{i}+1$, $\mathcal{H}_{t+\tau_{i}+1}$, we have
\begin{align}
&|\varphi^{r}(t+\tau_{i}+1)-\nonumber \\ &\qquad \qquad \phi(t+\tau_{i}+1,N^{r}_{1}(t+\tau_{i}+1)|
<\theta_{\varphi}+\epsilon_{1}.\label{assum:phi:1}
\end{align}
Similarly, with probability $(1-\varrho_{\varphi})$, and regardless of $i$ and the past history at time $t+\tau_{i}+1+N_{c}$, $\mathcal{H}_{t+\tau_{i}+1+N_{c}}$, we have
\begin{align}
&|\varphi(t+\tau_{i}+1)-\nonumber \\& \qquad  \phi(t+\tau_{i}+1+N_{c},N_{1}(t+\tau_{i}+1))|
<\theta_{\varphi}+\epsilon_{1}.\label{assum:phi:2}
\end{align}
\end{Statement}
\begin{Remark}\label{remark0}
Property~\ref{p3} states inequalities in Statement~\ref{st2} may hold in general with \emph{different} probabilities all not less than $(1-\varrho_{\varphi})$. However, to consider the \emph{worst case} analysis, in Statement~\ref{st2}, we have assumed these inequalities, with the given conditions, hold with the \emph{same} probability $(1-\varrho_{\varphi})$ for all $i$, where $0\leq i \leq K$.
\end{Remark}
\begin{Remark}\label{remark1}
Consider the $i+1_{\text{th}}$ and the $j+1_{\text{th}}$ rounds, where $0\leq i,j \leq K$ and $i\neq j$.
Since inequalities (\ref{assum:phi:1}) and (\ref{assum:phi:2}) in Statement~\ref{st2} may hold in the $i+1_{\text{th}}$ round with probability $1-\varrho_{\varphi}$ regardless of $\mathcal{H}_{t+\tau_{i}+1}$ and $\mathcal{H}_{t+\tau_{i}+1+N_{c}}$, respectively, Statement~\ref{st2} implies that the event that (\ref{assum:phi:1}) or the one that (\ref{assum:phi:2}) holds in the $i+1_{\text{th}}$ round is independent of the inequality (\ref{assum:phi:1}) or (\ref{assum:phi:2}) holding in the $j+1_{\text{th}}$ round. In addition, the event that (\ref{assum:phi:1}) holds in the $i+1_{\text{th}}$ round is independent of (\ref{assum:phi:2}) holding in the same round.
\end{Remark}

Before going to the main part of the proof, we first derive two key inequalities. To obtain the first one, note that for any two time instants $t_{1}$ and $t_{2}$, with $t \leq t_{1}\leq t+\tau_{K}$ and $t \leq t_{2}\leq t+\tau_{K}$,
using (\ref{ineq:assump:1}), we have that
\begin{eqnarray}
|\phi(t_{1},\tilde{N}_{1}(\mathbf{X}_{t_{1}}))-\phi(t_{2},\tilde{N}_{1}(\mathbf{X}_{t_{1}}))|<\epsilon_{1},\label{inequal:3}
\end{eqnarray}
and
\begin{eqnarray*}
|\phi(t_{1},\tilde{N}_{1}(\mathbf{X}_{t_{2}}))-\phi(t_{2},\tilde{N}_{1}(\mathbf{X}_{t_{2}}))|<\epsilon_{1}.
\end{eqnarray*}
By the definition of $\tilde{N}_{1}(\mathbf{X})$ and the inequality in
(\ref{inequal:3}), we have
 \begin{align*}
&\!\!\!\!\!\!\!\!\!\!\!\!\!\!\!\!\!\!\!\!\!\!
\phi(t_{1},\tilde{N}_{1}(\mathbf{X}_{t_{1}}))-\phi(t_{2},\tilde{N}_{1}(\mathbf{X}_{t_{2}}))
\nonumber\\ \qquad \leq&
\phi(t_{1},\tilde{N}_{1}(\mathbf{X}_{t_{1}}))-\phi(t_{2},\tilde{N}_{1}(\mathbf{X}_{t_{1}}))<\epsilon_{1}
 \end{align*}
We can obtain the other direction of the inequality similarly.
Thus,
\begin{eqnarray}
|\phi(t_{1},\tilde{N}_{1}(\mathbf{X}_{t_{1}}))-\phi(t_{2},\tilde{N}_{1}(\mathbf{X}_{t_{2}}))|<\epsilon_{1}.\label{inequal:opt:phi}
\end{eqnarray}
This inequality shows that when backlog vector has a large absolute
value, the optimal $\phi$ does not vary significantly in a limited
time horizon. In particular, the variation approaches zero when
$\|\mathbf{X}_{t}\|$ approaches $\infty$.

To derive the second key inequality, first note that based on the definition of $\tau_{i}$ given in the proof of part(a) of the theorem,
the $i+1_{\text{th}}$ round after time $t$ begins at $t+\tau_{i}+1 $, and the time interval between $t+\tau_{0}+1$ and $t+\tau_{K}+1$
consists of $K$ scheduling rounds.
To simplify the notation, let $\tilde{N}_{1}$ be the optimal value of $N_{1}$ for the first round after time $t$, i.e.,
$\tilde{N_{1}}=\tilde{N_{1}}(\mathbf{X}_{t+\tau_{0}+1})$. In addition, let $N_{1}^{r}(j)$ be the candidate value for $N_{1}$ in the $j+1_{\text{th}}$ round, and let $N_{1}(j)$ be the value of $N_{1}$ used in the update interval of the $j+1_{\text{th}}$ round, i.e., $N_{1}^{r}(j)=N_{1}^{r}(t+\tau_{j}+1)$, and
$N_{1}(j)=N_{1}(t+\tau_{j}+1)$.

Now, consider the $i+1_{\text{th}}$ round, $i\geq 0$, and suppose the optimal $N_{1}$ is selected at this round,
i.e., $N_{1}^{r}(i)=\tilde{N}_{1}(\mathbf{X}_{t+\tau_{i}+1})$. Let
$\hat{N_{1}}=\tilde{N_{1}}(\mathbf{X}_{t+\tau_{i}+1})$.  Then the inequality in (\ref{assum:phi:1}) and the preceding inequality
imply that with probability $(1-\varrho_{\varphi})$
\begin{align}
&|\varphi^{r}(t+\tau_{i}+1)-\phi(t+\tau_{0}+1,\tilde{N_{1}})|
 <2\epsilon_{1}+\theta_{\varphi}.
\label{inequal:phi}\end{align}
 Let
 \begin{align}
\epsilon'=2\epsilon_{1}+\theta_{\varphi}.\label{def:eps'}
\end{align}
Based on the assumption $6\theta_{\varphi}<\alpha$ imposed by the Lemma and that $\epsilon_{1}<\frac{1}{2}(\frac{\alpha}{6}-\theta_{\varphi})$,
we have
\begin{eqnarray}
0<6\epsilon'<6\big((\frac{\alpha}{6}-\theta_{\varphi})+\theta_{\varphi}\big)=\alpha.
\label{ineq:assump:3}
\end{eqnarray}
The inequality (\ref{inequal:phi}) is the second key inequality required for the rest of the proof.

 We are now in a position to explain the essence of the proof, where we find a lower-bound for the fraction of time in the horizon of $K$ rounds in which near optimal values for $N_{1}$ are used. Towards this end, we first assume that the inequalities in (\ref{assum:phi:1}) and (\ref{assum:phi:2}) hold with probability one for
all $K$ scheduling rounds, thus assuming $\varrho_{\varphi}=0$ in Statement~\ref{st2}. We then extend our discussion to
 realistic cases where $\varrho_{\varphi}>0$.

 \emph{\textbf{Discussion assuming} $\varrho_{\varphi}=0$ \textbf{:}}
Suppose at the $i+1_{\text{th}}$ round, $i\geq1$, the optimal $N_{1}$ corresponding to $\mathbf{X}_{t+\tau_{i}+1}$ is selected, i.e., $N_{1}^{r}(i)=\hat{N}_{1}=\tilde{N}_{1}(\mathbf{X}_{t+\tau_{i}+1})$. Considering the scheduling policy, with respect to the update of $N_{1}$ in $i+1_{\text{th}}$ scheduling round,
there are two possible cases:\\

\textbf{Case~1:} In this case, we assume $\varphi^{r}(t+\tau_{i}+1)>\varphi(t+\tau_{i-1}+1)+\alpha$.
Thus, according to the update rule, $N_{1}$ gets updated at the $i+1_{\text{th}}$ round, and takes the value $N_{1}(i)=N_{1}^{r}(i)=\hat{N}_{1}$. However, it remains unchanged until the the $K+1_{\text{th}}$ round.
We can prove this statement by induction. To see this, assume that $N_{1}$ remains fixed after the $i+1_{\text{th}}$ but changes for the first time in the $j_{\text{th}}+1$ round, where $j>i$. Therefore, by the update rule, we must have
\begin{align}
\varphi^{r}(t+\tau_{j}+1)>\varphi(t+\tau_{j-1}+1)+\alpha.\label{assum:case1}
\end{align}
Since
\begin{align*}
&|\phi(t+\tau_{j}+1,N^{r}_{1}(j))-\phi(t+\tau_{0}+1,N^{r}_{1}(j))|<\epsilon_{1},
\end{align*}
and
\begin{align*}
&|\varphi^{r}(t+\tau_{j}+1)-\phi(t+\tau_{j}+1,N^{r}_{1}(j))|<\theta_{\varphi}+\epsilon_{1},
\end{align*}
which follow from (\ref{ineq:assump:1}) and (\ref{assum:phi:1}), respectively, and the assumption that $\varrho_{\varphi}=0$,
we have
\begin{align}
&\varphi^{r}(t+\tau_{j}+1)<\phi(t+\tau_{0}+1,N^{r}_{1}(j))+2\epsilon_{1}
+\theta_{\varphi}\nonumber
\\& \qquad \qquad \qquad \leq \phi(t+\tau_{0}+1,\tilde{N}_{1})+\epsilon', \label{ineq:case1}
\end{align}
where the last inequality follows from the definition of $\tilde{N}_{1}$.

Similarly, since by assumption $N(j-1)=\hat{N}_{1}$, we can use (\ref{ineq:assump:1}) and (\ref{assum:phi:2}) to show that
\begin{align}
|\varphi(t+\tau_{j-1}+1)-\phi(t+\tau_{i}+1,\hat{N_{1}})|<\epsilon'.\nonumber
\end{align}
Considering this inequality and (\ref{inequal:opt:phi}), we obtain
\begin{align}\varphi(t+\tau_{j-1}+1)>
\phi(t+\tau_{0}+1,\tilde{N_{1}})-\epsilon'-\epsilon_{1}.\label{ineq:case2}
\end{align}
Finally, considering (\ref{assum:case1}), (\ref{ineq:case1}), and (\ref{ineq:case2}), we obtain
\begin{align}
&\phi(t+\tau_{0}+1,\tilde{N_{1}})+\epsilon'> \nonumber\\
 & \qquad \qquad \qquad \qquad
\phi(t+\tau_{0}+1,\tilde{N_{1}})-\epsilon'-\epsilon_{1}+\alpha,\nonumber
\end{align}
which implies that $2\epsilon'+\epsilon_{1}>\alpha$. This is in
contradiction with (\ref{ineq:assump:3})
stating that $6\epsilon'<\alpha$. Therefore, $N_{1}(j)=\hat{N}_{1}$ for $i\leq j\leq K-1$, proving the claim.

A byproduct of the above discussion is that after the $i_{\text{th}}$ round, $\varphi(t+\tau_{j}+1)$ stays close to
$\phi(t+\tau_{0}+1,\tilde{N}_{1})$. More precisely, since $N_{1}(j)=\hat{N}_{1}$ for $i\leq j\leq K-1$, we have
\begin{align*}
&|\varphi(t+\tau_{j}+1)-\phi(t+\tau_{j}+1,\hat{N}_{1}))|<\theta_{\varphi}+\epsilon_{1}.
\end{align*}
Moreover, from (\ref{ineq:assump:1}) we have
\begin{align*}
&|\phi(t+\tau_{j}+1,\hat{N}_{1})-\phi(t+\tau_{i}+1,\hat{N}_{1})|<\epsilon_{1}.
\end{align*}
Using the last two inequalities and  (\ref{inequal:opt:phi}) , for $i\leq j\leq K-1$, we obtain
\begin{eqnarray}
|\varphi(t+\tau_{j}+1)-\phi(t+\tau_{0}+1,\tilde{N}_{1})|<\epsilon'+\epsilon_{1},
\label{inequal:ineff:1}
\end{eqnarray}
which shows how close is $\varphi(t+\tau_{j}+1)$ to $\phi(t+\tau_{0}+1,\tilde{N}_{1})$.

 \textbf{Case~2:} In this case, we assume $\varphi^{r}(t+\tau_{i}+1)\leq \varphi(t+\tau_{i-1}+1)+\alpha$.
Taking similar steps as in Case~1, we can show that
\begin{align}
\phi(t+\tau_{0}+1,\tilde{N_{1}})-\epsilon'\leq
\varphi^{r}(t+\tau_{i}+1) \nonumber
\end{align}
and
\begin{align}
\varphi(t+\tau_{i-1}+1) \leq\phi(t+\tau_{0}+1,N_{1}(i-1))+\epsilon',\nonumber
\end{align}
Hence, using the assumption, we obtain
\begin{align}
&\phi(t+\tau_{0}+1,\tilde{N_{1}})-2\epsilon'-\alpha\leq\phi(t+\tau_{0}+1,N_{1}(i-1))
\label{inequal:1}.
\end{align}
We next show that $N_{1}$ gets updated at most once in the rest of $K-(i+1)$ rounds.
Let the $j_{1}+1_{\text{th}}$ round, for $i<j_{1}\leq K-1$, be the first
round after the $i+1_{\text{th}}$ round that $N_{1}$ gets updated. Using similar arguments as the ones in Case~1, we have
\begin{align}
\varphi^{r}(t+\tau_{j_{1}}+1)<\phi(t+\tau_{0}+1,N_{1}^{r}(j_{1}))+\epsilon'.\nonumber
\end{align}
and
\begin{align}
\varphi(t+\tau_{j_{1}-1}+1)>\phi(t+\tau_{0}+1,N_{1}(i-1))-\epsilon',\nonumber
\end{align}
where in the above we have used the assumption that $N_{1}$ does not change before the $j_{1}+1$ round, and thus have set $N_{1}(j_{1}-1)=N_{1}(i-1)$. Since $N_{1}$ gets updated at the $j_{1}+1_{\text{th}}$, we have $N_{1}(j)=N_{1}^{r}(j)$. Using this, the update rule, and the last two inequalities, we have
\begin{align*}
\!\!\!\!\!\!\!\!\!\!\!\!\!\!\!\!\!&\phi(t+\tau_{0}+1,N_{1}(i-1))-\epsilon'+\alpha \nonumber \\
&\qquad \qquad \qquad \qquad
<\phi(t+\tau_{0}+1,N_{1}(j_{1}))+\epsilon' 
\end{align*}
This inequality and (\ref{inequal:1}) yield
\begin{eqnarray}
\phi(t+\tau_{0}+1,\tilde{N}_{1})-4\epsilon'<\phi(t+\tau_{0}+1,N_{1}(j_{1})).
\label{ineff:1}
\end{eqnarray}
Similarly, if there exists $j_{1}<j_{2}\leq K-1$, such that at $t+\tau_{j_{2}}+1$, $N_{1}$ becomes
updated for the second time, we can show that
\begin{align}
\!\!\!\!\!\!\!\!\!\!\!\!\!\!\!\!\!
&\phi(t+\tau_{0}+1,N_{1}(j_{1}))-\epsilon'+\alpha \nonumber
\\ &\qquad \qquad \qquad \qquad <\phi(t+\tau_{0}+1,N_{1}(j_{2}))+\epsilon'. \nonumber
\end{align}
In other words,
\begin{align*}
&\phi(t+\tau_{0}+1,N_{1}(j_{1}))+\alpha-2\epsilon'
\nonumber\\ & \qquad \qquad \qquad \qquad \quad<\phi(t+\tau_{0}+1,N_{1}(j_{2})).
\end{align*}
Therefore, every time that $N_{1}$ becomes updated, the
algorithm finds a better estimate for
$\phi(t+\tau_{0}+1,\tilde{N}_{1})$. More specifically, after each update, the gap between $\phi(t+\tau_{0}+1,N_{1}(j_{k}))$ and $ \phi(t+\tau_{0}+1,\tilde{N}_{1})$ is decreased
by $(\alpha-2\epsilon')>\frac{2}{3}\alpha$. However, (\ref{inequal:1}) shows that the initial gap is $\alpha+2\epsilon'$, which is less than or equal to $\frac{4}{3}\alpha$. Therefore, $N_{1}$ can be updated at most once in the rest of $K-i-1$ scheduling rounds.

In this case, similar to what we observed in Case~1, $\varphi(t+\tau_{j}+1)$ stays close to $\phi(t+\tau_{0}+1,\tilde{N}_{1})$. To see this, consider a scheduling round, e.g., $j+1_{\text{th}}$ for $i\leq j \leq K-1$
round, where $N_{1}(i-1)$ is used. By (\ref{ineq:assump:1}) and
(\ref{assum:phi:2}), we have
\begin{align*}
&\!\!\!\!|\varphi(t+\tau_{j}+1)-\phi(t+\tau_{0}+1,N_{1}(i-1))|
<\epsilon'.
\end{align*}
Considering the above inequality and (\ref{inequal:1}), we obtain
\begin{align}
&\!\!|\varphi(t+\tau_{j}+1)-\phi(t+\tau_{0}+1,\tilde{N}_{1})|
<\alpha+3\epsilon' \label{inequal:ineff:2}.
\end{align}
In the same manner, if instead of $N_{1}(i-1)$ an updated version of $N_{1}$ is used in an scheduling round, we can use
the inequality in (\ref{ineff:1}) to show that the
above inequality still holds. Hence, the inequality in (\ref{inequal:ineff:2}) holds for all $j$ with $ i\leq j \leq K-1$ since, as proved earlier, $N_{1}$ becomes updated at most once.

 Combining the
inequality (\ref{inequal:ineff:1}) associated with Case~1 and the
inequality (\ref{inequal:ineff:2}) associated with Case~2, we see
that regardless of which case happens, the following holds for $i\leq j \leq K-1$:
\begin{eqnarray}
|\varphi(t+\tau_{j}+1)-\phi(t+\tau_{0}+1,\tilde{N}_{1})|<\gamma,
\label{ineff:3}
\end{eqnarray}
where
\begin{align}
\gamma=\alpha+3\epsilon'. \label{gamma}
\end{align}

Inspired by the above inequality, we now define a new random variable $R_{K}$ as
the percentage of time that ``near optimal'' solution is used in the time
horizon consisting of $K$ rounds. By near optimal in a scheduling round, e.g., the $j+1_{\text{th}}$ round, we mean a choice of $N_{1}$
that ensures $\varphi(t+\tau_{j}+1)$ is close to $\phi(t+\tau_{0}+1,\tilde{N}_{1})$ in the sense of (\ref{ineff:3}). Intuitively, a larger $R_{K}$ results in a larger scaling factor, and thus, a better throughput performance. In the
following, using the preceding discussions provided in Case~1 and Case~2, we find a lower bound for
$R_{K}$.

 As explained in Section~\ref{sec:dcp:sch}, in the beginning of each round, e.g., the $j+1_{\text{th}}$ round, the optimal $N_{1}$, corresponding to $\mathbf{X}_{t+\tau_{j}+1}$, is chosen independently with probability $\delta$. Therefore, we see that with probability
$(1-\delta)^{i-1}\delta$, \emph{after} the first round, the optimal solution is selected for the first time in the $i+1_{\text{th}}$
round, $i\geq 1$. Suppose this event happens at $i+1_{\text{th}}$ round, $i\geq 1$. If Case~1 happens, we can partition the time interval between $t+\tau_{0}+1$
and $t+\tau_{K}+1$ into three sets. The first set consists of all test intervals. The second set consists of the
update intervals before the $i+1_{th}$ round. Finally, the third set consists of the update intervals after the $i_{th}$ round. Considering these sets in sequence, we can express the total number of timeslots between
$t+\tau_{0}+1$ and $t+\tau_{K}+1$ by
 \begin{align}
KN_{c}&+\sum_{j=0}^{i-1}N_{c}N_{3}(j) \nonumber
\\
&+\sum_{j=1}^{K-i}N_{c}\min(\max(1,\frac{N_{3}(i-1)}{2})2^{j-1},L_{1}) \label{total:t},
 \end{align}
 where $N_{3}(j)=N_{3}(t+\tau_{j}+1)$. To obtain the above expression, we have used the fact that
 when Case~1 happens, according to the update rule, at the $i+1_{\text{th}}$ round $N_{3}(i)$ becomes half of the previous value
 for $N_{3}$, but keeps doubling for each following round.
Recalling that (\ref{ineff:3}) holds after the $i_{\text{th}}$ round, and $N_{3}(j)\leq L_{1}$, we can use (\ref{total:t})
to show that for $K>i+1$, w.p.1,
\begin{align}
R_{K}\geq
\frac{2+\sum_{j=1}^{K-i-1}\min(2^{j-1},L_{1})}
{K+iL_{1}
+2+\sum_{j=1}^{K-i-1}\min(2^{j-1},L_{1})}.\label{r_k:ideal}.
\end{align}

For a given fixed $i$, the above fraction approaches
$\frac{L_{1}}{1+L_{1}}$ as $K$ approaches $\infty$. Therefore, for any given positive $\epsilon_{2}$, we
can choose $K$ sufficiently large such that for all $i$ with $1\leq i \leq i_{max}$, the above
fraction is larger than $\frac{L_{1}}{1+L_{1}}-\epsilon_{2}$.
Applying a similar argument to the second case, we can find a
sufficiently large $K$ such that the fraction of time over which
the near optimal solution is used is larger than
$\frac{L_{1}}{1+L_{1}}-\epsilon_{2}$. In addition, we can select $i_{max}$ such that
for a given positive $\zeta_{1}$
\begin{eqnarray*}
\sum_{i=1}^{i_{max}}(1-\delta)^{i-1}\delta>1-\zeta_{1}.
\end{eqnarray*}
Hence, if $K$ is sufficiently large, with probability larger than $1-\zeta_{1}$, we have
\begin{eqnarray}
R_{K}\geq
(\frac{L_{1}}{1+L_{1}}-\epsilon_{2}). \label{lowerb:RK}
\end{eqnarray}

This is an interesting observation. Since the choices for $\epsilon_{2}$ and $\zeta_{1}$ are arbitrary, this observation implies that
in the limit of large backlog vectors, the policy keeps the network operating at near optimal points for at least $\frac{L_{1}}{1+L_{1}}$ fraction of time. Hence, in the limit, at most only the time for selecting new values for $N_{1}$ and observing their performance is \emph{wasted}, which constitues $\frac{1}{1+L_{1}}$ fraction of total time. Recall that near optimality is defined in (\ref{ineff:3}), and $\phi(t+\tau_{0}+1,\tilde{N}_{1})=\tilde{\phi}(\mathbf{X}_{t+\tau_{0}+1})$, we therefore, as a result of the preceding inequality, expect the limiting scaling factor of the capacity region
to be a function of $\tilde{\phi}(\mathbf{X})$, and be proportional to $\frac{L_{1}}{1+L_{1}}$.

 Note that to obtain the above results, in particular those mentioned in Case~1 and Case~2, we assume that the inequalities in (\ref{assum:phi:1}) and (\ref{assum:phi:2}) hold for all $K$ scheduling rounds after time $t$. Therefore, the above discussion for $R_{K}$ holds only for the limiting case of $\varrho_{\varphi} = 0$. In the following, we extend the preceding discussions for a realistic situation where $\varrho_{\varphi}>0$, and obtain a general lower bound for $R_{K}$.

\emph{\textbf{Discussion assuming} $\varrho_{\varphi}>0$ \textbf{:}}We start by assuming that
\begin{align}
\|\mathbf{X}_{t}\|>M''_{K}>M'_{K},\label{another:assup}
\end{align}
for a sufficiently large $M''_{K}$ such that for a given $C$, $C\gg1$,
Statement~\ref{st1} holds for all times $t_{1}$ and $t_{2}$ greater than $t-1$ and less than $t+\tau_{(K+2)C}+1$, and Statement~\ref{st2} holds for all $i$ with $0\leq i\leq (K+2)C$. We partition the time between $t+\tau_{0}+1$ and $t+\tau_{(K+2)C}+1$ into a set of \emph{periods}, where each period consists of several
scheduling rounds. For the simplicity of discussion, we assume that the first period always starts at $t+\tau_{0}+1$.

 Corresponding to each period, e.g. the $j_{\text{th}}$ period, we define a \emph{positive} r.v. $i_{\delta,j}$. This r.v. takes value $i$, $i>0$, if the following conditions are met. First, in the $i+1_{\text{th}}$ round of the $j_{\text{th}}$ period, for the first time in that period the optimal value for $N_{1}$ is selected. Second, the inequality (\ref{assum:phi:1}) holds for $\varphi^{r}$ at the $i+1_{\text{th}}$ round as well as (\ref{assum:phi:2}) for $\varphi$ at the $i_{\text{th}}$ round, both in the $j_{\text{th}}$ period. Third, $i$ equals $C-1$ if the last two conditions do not hold for any of the second to the $(C-2)_{\text{th}}$ rounds in the $j_{\text{th}}$ period. Recall that the optimal $N_{1}$ is chosen independently in each round with probability $\delta$. Thus, using Remark~\ref{remark1} with $K $ replaced with $(K+2)C$, we see that $i_{\delta,j}$ becomes a truncated geometric r.v. with success probability
 \begin{align}
\delta'=(1-\varrho_{\varphi})^{2}\delta,\label{d11}
 \end{align}
 and with the property that
 \begin{align}
P(i_{\delta,j}=C-1)=1-\sum_{i=1}^{C-2}\delta' (1-\delta')^{i-1}.\label{d12}
 \end{align}

 Similarly, corresponding to the $j_{\text{th}}$ period, we define a \emph{non-negative} r.v. denoted by $i_{\varphi,j}$ that is zero if $i_{\delta,j}=C-1$, and otherwise, is the number of consecutive rounds immediately following the $i_{\delta,j}$$_{\text{th}}$ round in the $j_{\text{th}}$ period for all of which the inequalities in (\ref{assum:phi:1}) and (\ref{assum:phi:2}) hold. Similar to $i_{\delta,j}$, we limit $i_{\varphi,j}$ to be upper-bounded by $C-1$. We do so by letting $i_{\varphi,j}=C-1$ if for all $C-1$ rounds after the $i_{\delta,j}$$_\text{th}$ round (\ref{assum:phi:1}) and (\ref{assum:phi:2}) hold. Using this definition of $i_{\varphi,j}$, and Remark~\ref{remark1} with $K $ replaced with $(K+2)C$, it is easy to see that
\begin{align}
& P(i_{\varphi,j}=0 | i_{\delta,j}= C-1)=1, \label{d21} \\ & P(i_{\varphi,j}=0 | i_{\delta,j}\neq C-1)=\varrho_{\varphi}\label{d22}
  \end{align}
  and
  \begin{align}
&  P(i_{\varphi,j}=k|i_{\delta,j} \neq C-1)
  \nonumber \\ &\qquad \qquad =(1-\varrho_{\varphi})^{2k-1}(1-(1-\varrho_{\varphi})^{2}), \ \ 1\leq k \leq C-2,\label{d23}
   \end{align}
   and by the boundedness of $i_{\varphi,j}$,
   \begin{align}
& P(i_{\varphi,j}=C-1|i_{\delta,j} \neq C-1)
 \nonumber \\ & \qquad \qquad =1-\sum_{k=0}^{C-2} P(i_{\varphi,j}=k|i_{\delta,j} \neq C-1)\nonumber
 \\&\qquad \qquad=(1-\varrho_{\varphi})^{2(C-1)-1} .\label{d24}
   \end{align}

To complete the characterization of periods, we define the last round in the $j_{\text{th}}$ period to be the one immediately following the $i_{\delta,j}+ i_{\varphi,j}$$_{\text{th}}$ round in the $j_{\text{th}}$ period. This indicates that the $j_{\text{th}}$ period consists of $i_{\delta,j}+ i_{\varphi,j}+1$ rounds, and thus by the definition of $i_{\delta,j}$ and $i_{\varphi,j}$, its length is always less than $2C$.

Having introduced periods, we now define the sequence $\{p_{j}\}_{j=0}^{\infty}$, with $p_{0}=0$, as a subset of indices such that $\tau_{p_{j}}$, $j\geq 1$, is the number of timeslots form time $t$ to the last timeslot in the $j_{\text{th}}$ period. By definition, therefore, the $j_{\text{th}}$ period, $j\geq 1$, starts at $t+\tau_{p_{j-1}}+1$ and ends at $t+\tau_{p_{j}}+1$. Let $i_{K}$ be number of periods that are completely contained in the $K$ rounds after time $t$, i.e.,
\begin{align}
i_{K}=\max \{j: p_{j}< K, j\geq 0\}.
\end{align}
By virtue of the definitions for a scheduling period, $i_{\delta,j}$, and $i_{\varphi,j}$, we can see that
for all rounds after $i_{\delta,j}$$_{\text{th}}$ and before the last round in the $j_{\text{th}}$ period, all conditions to apply the discussions in Case~1 and Case~2 are met. Hence, considering (\ref{ineff:3}), for $1\leq j \leq i_{K}+1$ and $ i_{\delta,j}< i \leq i_{\delta,j}+i_{\varphi,j}$ with $p_{j-1}+i-1<K$ we have that
\begin{align}
|\varphi(t+\tau_{p_{j-1}+i-1}+1)-\tilde{\phi}(t+\tau_{0}+1)|<\gamma.\label{ineff:3b}
\end{align}
Note that $t+\tau_{p_{j-1}+i-1}+1$ is the start point of the $i_{\text{th}}$ round in the $j_{\text{th}}$ period, and we have set condition $p_{j-1}+i-1<K$ to consider only the first $K$ rounds after time $t$.

We now focus on finding a lower bound for $R_{K}$. Towards that goal, we use r.v.'s $i_{\delta,j}$ and  $i_{\delta,j}$ to define a new sequence of $N_{3}$ denoted by $N^{'}_{3}$ according to the following:
\begin{align}
&N^{'}_{3}(k=p_{j-1}+i-1)=N^{'}_{3}(j,i)\nonumber
\\
&=\left\{\begin{array}{l l}
L_{1} & (1 \leq i \leq i_{\delta,j}) \ \vee
\\ & \qquad \quad \qquad (i=i_{j}+1) \\
1 & (i=i_{\delta,j}+1) \wedge (i_{\varphi,j}=1) \\
2 & (i=i_{\delta,j}+1) \wedge (i_{\varphi,j}>1) \\
\min(\frac{2^{i}}{2^{i_{\delta,j}+2}},L_{1})  & (i_{\delta,j}+2 \leq i \leq i_{j})  \wedge
     (i_{\varphi,j}>1)
\end{array}
\right.\nonumber
\end{align}
where
\begin{align}
i_{j}=i_{\delta,j}+i_{\varphi,j} \nonumber.
 \end{align}
 Note that a round after time $t$ can be specified uniquely either as the $k_{\text{th}}$ round
after time $t$, or as the $i_{\text{th}}$ round in the $j_{\text{th}}$ period. We thus in the above have defined $N^{'}_{3}$ as a function of
the round number $k$ after time $t$, and also as a function of the pair $(j,i)$. Similarly, $N_{3}$ can be considered as a function of
either $K$ or $(j,i)$. In addition, note that the above definition of $N^{'}_{3}$ is mainly motivated by the method used to obtain (\ref{r_k:ideal}).

To simplify the analysis, we slightly modify the definition of $R_{K}$ such that
\begin{align}
R_{K} &= \frac{\sum_{j=1}^{i_{K}}\sum_{i=i_{\delta,j}+1}^{i_{j}} N_{c}N_{3}(j,i)}{\tau_{K}-\tau_{0}} \nonumber.
\end{align}
Hence, $R_{K}$ concerns only the rounds that are within the first $i_{K}$ periods, and
for which (\ref{ineff:3b}) holds. Considering the above definition, we can use a simple inspection to show that the above choices for $N^{'}_{3}$ ensure that w.p.1
\begin{align}
R_{K} &\geq R^{'}_{K}  = \frac{\sum_{j=1}^{i_{K}} \lambda_{r}^{C}(j)  }{\sum_{j=1}^{i_{K}+1}\lambda_{t}^{C}(j)},\label{rkk}
\end{align}
where
\begin{align}
\lambda_{t}^{C}(j)&=   \sum_{i=1}^{i_{j}+1}\ (1+N^{'}_{3}(j,i)) \nonumber  \\&=i_{j}+1+(i_{\delta,j}+1)L_{1} \nonumber \\
&+\mathbf{1}_{(i_{\varphi,j}=1)}
+(2+\sum_{i=0}^{i_{\varphi,j}-2}\min(2^{i},L_{1}))\mathbf{1}_{(i_{\varphi,j}>1)},
\end{align}
and
\begin{align}
\lambda_{r}^{C}(j)&=\sum_{i=i_{\delta^{'},j}+1}^{i_{j}}\ N^{'}_{3}(j,i)
\nonumber\\&=\mathbf{1}_{(i_{\varphi,j}=1)}
+\big(2+\sum_{i=0}^{i_{\varphi,j}-2}\min(2^{i},L_{1})\big)\mathbf{1}_{(i_{\varphi,j}>1)}.
\end{align}
As expected, $\lambda_{r}^{C}(j)$ denotes the minimum contribution of the $j_{\text{th}}$ period to the ratio $R_{K}$. The term $\lambda_{t}^{C}(j)$ is the total
length of the $j_{\text{th}}$ period that could potentially minimize the ratio $R_{K}$.
In addition, note that the inequality (\ref{rkk}) in general holds even when Remark~\ref{remark0} does not hold, and thus, when the distribution of $i_{\delta,j}$ and $i_{\varphi,j}$ is not given by (\ref{d11})-(\ref{d24}). However, as stated in Remark~\ref{remark0}, we consider the worst case which enables us to find a lower-bound for $R_{K}$ that holds with high probability.
We next show that the random variable $R^{'}_{K}$ is a function of i.i.d pairs, and in fact, is the average accumulated reward for a renewal process.

First, note that by definition $i_{\delta,j}>0$, and hence, a scheduling period, which consists of $i_{\delta,j}+i_{\varphi,j}+1$ rounds, at least contains of two rounds. This implies that the $K$ rounds under consideration consitute at most $\lfloor\frac{K}{2}\rfloor$ complete periods.
Consequently, $R^{'}_{K}$ is a function of at most $K_{P}=\lfloor\frac{K}{2}\rfloor+1$ periods, and thus, is completely characterized by
\begin{align}
\Big\{(i_{\delta,j},i_{\varphi,j}), \ 1\leq j \leq K_{P}\Big\}.\label{set}
\end{align}
We know that by definition a period consists of at most $2C-1$ rounds. Therefore, considering Remark~\ref{remark1} with $K$ replaced with $K_{P}(2C-1)$, we see that the above set is consisted of i.i.d pairs, with distribution given by (\ref{d11})-(\ref{d24}), if Statement~\ref{st2} holds for all $i$ with $0\leq i\leq K_{P}(2C-1)$. Recall that we started by assuming $\|\mathbf{X}_{t}\|>M''_{K}$ such that Statement~\ref{st2} holds for $0\leq i\leq (K+2)C$. But this means that Statement~\ref{st2} holds for all $i$ with $0\leq i\leq K_{P}(2C-1)$ since $K_{P}(2C-1)<(K+2)C$. Therefore, we have that the pairs in (\ref{set}) are i.i.d.\footnote{Note that if $C=\infty$, $i_{\delta,j}$ or $i_{\varphi,j}$ may take any finite value. Hence, a proper definition of $i_{\delta,j}$ or $i_{\varphi,j}$ with distributions given by (\ref{d11})-(\ref{d24}) requires Statement~\ref{st2} hold for all $i\geq0$, which cannot be true by assuming $\|\mathbf{X}_{t}\|>M''_{K}$, for any finite value of $M''_{K}$.}

 Next, observe that since the pair $(\lambda_{t}^{C}(j),\lambda_{r}^{C}(j))$ depends only on $(i_{\delta,j},i_{\varphi,j})$, the sequence $\{(\lambda_{t}^{C}(j),\lambda_{r}^{C}(j)) : \ 1\leq j\leq K_{P} \}$ also consists of i.i.d. pairs. This sequence is defined for $1\leq j\leq K_{P}$, but can be defined for $j>K_{P}$ by defining the pair $(\lambda_{t}^{C}(j),\lambda_{r}^{C}(j))$, for $j>K_{P}$, as an i.i.d. version of $(\lambda_{t}^{C}(1),\lambda_{r}^{C}(1))$. The resulting expanded sequence
\begin{align*}
\{(\lambda_{t}^{C}(j),\lambda_{r}^{C}(j)) : \ j\geq 1 \}
\end{align*}
defines a reward renewal process. For this renewal process, $\lambda_{t}^{C}(j)$ is the length of the $j_{\text{th}}$ inter-renewal interval, $\lambda_{r}^{C}(j)$ is the accumulated reward collected at the end of $j_{\text{th}}$ renewal interval, and
$R^{'}_{K}$ is the average accumulated reward prior to end of $i_{K}+1_{\text{th}}$ inter-renewal interval.

 Consider the extended sequence, and let $R^{'}_{k}$, for any $k>0$, be defined similar to $R^{'}_{K}$. Applying the strong law for the renewal process, and noting that $i_{k} \to \infty$, a.s., as $k \to \infty$, we obtain
 \begin{align*}
R^{C}_{\infty} \triangleq \lim_{k\to \infty}R^{'}_{k}=\frac{\mathds{E}[\lambda_{r}^{C}(1)]}{\mathds{E}[\lambda_{t}^{C}(1)]}, \ \ \text{a.s.}
 \end{align*}
Hence, by the almost surely convergence, for any given $\epsilon_{R}>0$ and $\varrho_{R}>0$, there exists a sufficiently large $n^{C}_{\epsilon_{R},\varrho_{R}}$ such that \cite{borovkov}
\begin{align}
P(\sup_{k\geq n^{C}_{\epsilon_{R},\varrho_{R}}} | R^{'}_{k}-R^{C}_{\infty} |<\frac{\epsilon_{R}}{2})> (1-\varrho_{R}) .\label{ineq:Rr1}
\end{align}
But since $\lim_{C\to \infty} R^{C}_{\infty}=R_{\infty}$, we can chose a sufficiently large $C$ such that
\begin{align*}
|R^{C}_{\infty}-R_{\infty}|<\frac{\epsilon_{R}}{2}.
\end{align*}
Considering (\ref{ineq:Rr1}) for this value of $C$, we have that
\begin{align}
P(\sup_{k\geq n^{C}_{\epsilon_{R},\varrho_{R}}} | R^{'}_{k}-R_{\infty} |<\epsilon_{R})> (1-\varrho_{R}) .\label{ineq:Rr2}
\end{align}

The above inequality and (\ref{rkk}) imply that there exists a sufficiently large $K_{\epsilon_{R},\varrho_{R}}$ such that for $K>K_{\epsilon_{R},\varrho_{R}}$ and $\|\mathbf{X}_{t}\|>M''_{K}$
\begin{align}
P(  R_{K}>R_{\infty}-\epsilon_{R}| \mathcal{H}_{t+\tau_{0}+1})> (1-\varrho_{R}).\label{p:r:k}
\end{align}
Here, we have stated the probability conditioned on $\mathcal{H}_{t+\tau_{0}+1}$ since all of the previous discussions are valid regardless of  $\mathcal{H}_{t+\tau_{0}+1}$.
The above inequality states that with probability close to one, $R_{K}$ is close to $R_{\infty}$ in the sense that $R_{K}>R_{\infty}-\epsilon_{R}$. This
is a generalized version of the result obtained in (\ref{lowerb:RK}), as desired.

We are finally in a position to derive a lower bound for the LHS of the inequality in the lemma, denoted by $\Sigma$. First, note that
\begin{align}
\Sigma&=\mathds{E}\big[\sum_{i=\tau_{0}+1}^{\tau_{K}}\mathbf{X}_{t+i}\mathbf{D}_{t+i}|
\mathcal{H}_{t}\big]\nonumber \\&\geq \mathds{E}[\sum_{k=1}^{p_{i_{K}}}\sum_{i=1}^{\tau_{K}-\tau_{k-1}-N_{c}}\mathbf{X}_{t+\tau_{k-1}+
N_{c}+i}\mathbf{D}_{t+\tau_{k-1}+N_{c}+i}|\mathcal{H}_{t}\Big]\label{sigma:1},
\end{align}
where we have simply used the fact that the product $\mathbf{X}_{t+i}\mathbf{D}_{t+i}$ is positive, and neglected the contributions due to the test intervals, and also the ones due to the
rounds of the last partially covered period.

To simplify the notation, let $t_{j,i,l}$ denote the start of $l_{th}$ timeslot of the $i_{th}$ round in the $j_{th}$ period, i.e.,
\begin{align*}
t_{j,i,l}=t+\tau_{p_{j-1}+i-1}+l.
\end{align*}
In addition, let $\delta_{j,i}$ denote the length of the $i_{th}$ round in the $j_{th}$ period, i.e.,
 \begin{align}
\delta_{j,i}=\tau_{p_{j-1}+i}-\tau_{p_{j-1}+i-1}
=N_{c}(1+N_{3}(j,i)),\nonumber
 \end{align}
 where
 \begin{align*}
N_{3}(j,i)=N_{3}(t+\tau_{p_{j-1}+i-1}+1).
 \end{align*}

 Considering the partition generated by the scheduling periods, and the above definitions, we can use (\ref{sigma:1}) to show that
\begin{align}
\Sigma &\geq
\mathds{E}\Big[\sum_{j=1}^{i_{K}}\sum_{i=i_{\delta,j}+1}^{i_{j}}
\sum_{l=1}^{\delta_{j,i}-N_{c}}
 \mathbf{X}_{t_{j,i,N_{c}+l}}\mathbf{D}_{t_{j,i,N_{c}+l}}|\mathcal{H}_{t}\Big]
\nonumber\\ &\geq
\mathds{E}\Big[\sum_{j=1}^{i_{K}}\sum_{i=i_{\delta,j}+1}^{i_{j}}
\sum_{l=1}^{\delta_{j,i}-N_{c}}
\|\mathbf{X}_{t}\|
\nonumber \\&\qquad \frac{\|\mathbf{X}_{t_{j,i,N_{c}+1}}\|}{\|\mathbf{X}_{t}\|}(\tilde{\phi}
(t+\tau_{0}+1)-\gamma)|\mathcal{H}_{t}\Big]
\label{a:phi},
\end{align}
where the last inequality follows from (\ref{ineff:3b}).
Using (\ref{inequal:opt:phi}) and assuming
 \begin{align*}
 \frac{\|\mathbf{X}_{t_{j,i,N_{c}+1}}\|}{\|\mathbf{X}_{t}\|}>(1-\epsilon_{3}),
 \end{align*}
  we obtain
\begin{align}
\Sigma&\geq
\mathds{E}\Big[\sum_{j=1}^{i_{K}}\sum_{i=i_{\delta,j}+1}^{i_{j}}
\sum_{l=1}^{\delta_{j,i}-N_{C}}
\|\mathbf{X}_{t}\|
(1-\epsilon_{3})(\tilde{\phi}
(t)-\epsilon_{1}-\gamma)|\mathcal{H}_{t}\Big]
\nonumber \\&=\mathds{E}\Big[\|\mathbf{X}_{t}\|(\tilde{\phi}(t)-\epsilon_{1}-\gamma)
\nonumber \\&\qquad \qquad (1-\epsilon_{3})
\sum_{j=1}^{i_{K}}\sum_{i=i_{\delta,j}+1}^{i_{j}} N_{c}N_{3}(j,i)
|\mathcal{H}_{t}\Big].
\label{relative:abs}
\end{align}
But, by using (\ref{p:r:k}) and adopting a method similar to the one in Lemma~\ref{prelemma}, for $K>K_{\epsilon_{R},\varrho_{R}}$ and $\|\mathbf{X}_{t}\|>M''_{K}$ we can show that
\begin{align}
&\mathds{E}\Big[\sum_{j=1}^{i_{K}}\sum_{i=i_{\delta,j}+1}^{i_{j}} N_{c}N_{3}(j,i)|\mathcal{H}_{t+\tau_{0}+1}\Big]
\nonumber \\ &\qquad \ \geq (1-\epsilon_{4})(R_{\infty}-\epsilon_{R})\mathds{E}\Big[(\tau_{K}-\tau_{0})\big|\mathcal{H}_{t+\tau_{0}+1}\Big],
\label{ineq:RKP}
 \end{align}
 where $\epsilon_{4} \to 0$, as $K \to \infty$.

 Using (\ref{relative:abs}) and (\ref{ineq:RKP}), we obtain
\begin{align}
\Sigma&\geq
\mathds{E}\Big[(\tau_{K}-\tau_{0})\|\mathbf{X}_{t}\|
(\tilde{\phi}(t)-\epsilon_{1}-\gamma)
\nonumber \\&
\qquad \quad \qquad \qquad
(R_{\infty}-\epsilon_{R})
(1-\epsilon_{3})(1-\epsilon_{4})|\mathcal{H}_{t}\Big]
.\nonumber
\end{align}
If we assume
\begin{align}
\frac{\tau_{K}-\tau_{0}}{\tau_{K}+1}>1-\epsilon_{5},\nonumber
\end{align}
then using the above inequality and the definition of $\gamma$,
 given in (\ref{gamma}), we have that
\begin{align}
\Sigma \geq \mathds{E}\Big[(\tau_{K}+1)\|\mathbf{X}_{t}\|
 (R_{\infty}(\tilde{\phi}(t)-\alpha-3\theta_{\varphi})-\epsilon)\big |\mathcal{H}_{t}\Big]\nonumber,
\end{align}
where $\epsilon>0$, and can be made arbitrarily small by choosing sufficiently small values for $\epsilon_{1}$, $\epsilon_{3}$, $\epsilon_{4}$, $\epsilon_{5}$, and $\epsilon_{R}$. Note that since $\|\mathbf{X}_{t}-\mathbf{X}_{t+i}\|<C'_{K}$, $0\leq i\leq \tau_{K}$, as discussed in the beginning of the proof of the lemma, $\epsilon_{3}$ can be assumed arbitrarily small if $\|\mathbf{X}_{t}\|$ is sufficiently large. Moreover, since $\tau_{0}\leq (1+L_{1})N_{c}$ and $\tau_{K}+1\geq 2KN_{c}$, $\epsilon_{5}$ can be made arbitrarily small by assuming a sufficiently large $K$. Thus, by considering the discussions for $\epsilon_{R}$ and $\epsilon_{4}$, we see that we can make $\epsilon_{R}$, $\epsilon_{4}$, and $\epsilon_{5}$ all
sufficiently small by choosing $K>K_{\epsilon}$, for sufficiently large $K_{\epsilon}$.
 Having selected $K$, we can find a lower bound $M_{\epsilon,K}>M''_{K}>M'_{K}$ for $\|\mathbf{X}_{t}\|$ such that $\epsilon_{1}$ and $\epsilon_{3}$ are also sufficiently small. Hence, $\epsilon$ can be arbitrarily small, completing the proof of the lemma.
\end{proof}
\end{Lemma}

\begin{Lemma}\label{lemma:l5}
Let $0 \leq \delta < 1$. We have
\begin{align*}
\prod_{i=1}^{\infty} (1-\delta^{i}) > \exp\big(-(\frac{\delta}{(1-\delta)}+\frac{\delta^{2}}{(1-\delta)^{2}(1-\delta^{2})})\big)>0.
\end{align*}
\begin{proof}
First note that by Taylor's theorem, we have
\begin{align*}
\ln(1-\delta)\geq -\big(\delta+\frac{\delta^{2}}{2(1-\delta)^{2}}\big).
\end{align*}
Taking $\ln$ and then $\exp$ of the product term in the lemma, and using the above inequality, we can easily
show that
\begin{align*}
&\prod_{i=1}^{\infty} (1-\delta^{i}) \geq
\exp \big(-\sum_{i=1}^{\infty}\delta^{i}+\sum_{i=1}^{\infty}\frac{\delta^{2i}}{(1-\delta)^{2}} \big)
\\
\nonumber&=\exp\big(-(\frac{\delta}{(1-\delta)}+\frac{\delta^{2}}{(1-\delta)^{2}(1-\delta^{2})})\big)>0,\nonumber
\end{align*}
proving the lemma.
\end{proof}
\end{Lemma}

\bibliographystyle{IEEEtran}
\bibliography{IEEEabrv,citations}

\end{document}